\documentclass{article}
\usepackage[a4paper, margin=1in]{geometry}
\usepackage[utf8]{inputenc}
\usepackage{amsmath}
\usepackage{amsthm}
\usepackage[hidelinks]{hyperref}
\usepackage{cleveref}
\usepackage{amsfonts}
\usepackage{float}
\usepackage{dirtytalk}
\usepackage{xcolor}
\usepackage{graphicx}
\usepackage{thm-restate}
\usepackage{array}
\usepackage{tikz}
\usetikzlibrary{arrows.meta, patterns}
\usepackage{pgfplots}
\pgfplotsset{compat=1.17}
\usetikzlibrary{patterns.meta,pgfplots.fillbetween}
\pgfdeclarelayer{background}
\pgfdeclarelayer{foreground}
\pgfsetlayers{background,main,foreground}
\usepackage[shortlabels]{enumitem}
\usepackage[numbers,square]{natbib}
\usepackage[normalem]{ulem}
\usepackage{subcaption}
\DeclareMathOperator*{\argmax}{arg\,max}

\usepackage{bm}
\usepackage{bbm}
\pgfplotsset{compat=1.17}

\usepackage{arydshln}
\allowdisplaybreaks
\newtheorem{theorem}{Theorem}[section]
\newtheorem{lemma}[theorem]{Lemma}
\newtheorem{proposition}[theorem]{Proposition}
\newtheorem{corollary}[theorem]{Corollary}
\newtheorem{definition}[theorem]{Definition}
\newtheorem{remark}[theorem]{Remark}

\usepackage[para]{footmisc}

\newcommand{\setoutcomes}{[m]}
\newcommand{\numoutcomes}{m}
\newcommand{\setactions}{[n]}
\newcommand{\numactions}{n}
\newcommand{\Clinear}{\mathcal{C}_{\mathrm{linear}}}
\newcommand{\Cunbounded}{\mathcal{C}_{\mathrm{unbounded}}}
\newcommand{\Thetaall}{\Theta_{\mathrm{all}}}
\newcommand{\Cbounded}{\mathcal{C}_{\mathrm{bounded}}}
\newcommand{\numcombinatorialactions}{\bar{n}}
\newcommand{\numcombinatorialactionsexpanded}{n}
\newcommand{\pdim}[1]{\mathrm{Pdim}(#1)}
\newcommand{\Bepsilon}{\mathcal{B}_{\epsilon}}
\newcommand{\Lepsilon}{\mathcal{L}_{\epsilon}}

\title{The Pseudo-Dimension of Contracts}
\author{Paul D\"utting\thanks{Google Research} 
\and Michal Feldman\thanks{Tel Aviv University} $^{,}$\thanks{Microsoft ILDC}
\and Tomasz Ponitka\footnotemark[2]
\and Ermis Soumalias\thanks{University of Zurich} $^{,}$\thanks{ETH AI Center}
}
\date{January 23, 2024}

\begin{document}

\maketitle

\begin{abstract}
Algorithmic contract design 
studies scenarios 
where a principal incentivizes an agent to exert effort on her behalf.
In this work, we focus on settings where the agent's type is drawn from an unknown distribution, and formalize an offline learning framework for learning 
near-optimal 
contracts 
from sample agent types.
A central tool in our analysis is the notion of \emph{pseudo-dimension} from statistical learning theory. 
Beyond its role in establishing upper bounds on the sample complexity, pseudo-dimension measures the intrinsic complexity of a class of contracts, offering a new perspective on the tradeoffs between simplicity and optimality in contract design. 
{Our main results provide} 
essentially optimal tradeoffs 
between pseudo-dimension and representation error (defined as the loss in principal's utility) with respect to 
linear and bounded contracts. 
Using these tradeoffs, we derive sample- and time-efficient learning algorithms, 
and demonstrate their near-optimality by 
providing almost matching lower bounds on 
the sample complexity.
Conversely, for unbounded contracts, we prove an impossibility result showing that no learning algorithm exists.

Finally, we extend our techniques in three important ways.  
First, we provide refined pseudo-dimension and sample complexity guarantees for the \emph{combinatorial actions} model, revealing a novel connection between the number of critical values and sample complexity.
Second, we extend our results to  
\emph{menus} of contracts, showing that their pseudo-dimension scales linearly with the menu size.
Third, we adapt our algorithms to the \emph{online learning} setting, where we show that,  
a polynomial number of type samples 
suffice to learn near-optimal bounded contracts. 
Combined with prior work, this establishes a formal separation between expert advice and bandit feedback for this setting.
\end{abstract}

\section{Introduction}

Contract theory plays a foundational role in understanding markets for services, much like mechanism design and auctions underpin our understanding of markets for goods. 
The importance of this framework is reflected in its widespread application and recognition in economics, with the 2016 Nobel Prize awarded to Hart and Holmström \cite{Nobel2016}. 
As more of the classic applications are moving online and growing in scale, and as new applications and opportunities arise in online markets, this mandates an algorithmic approach to contract theory (see, e.g., \cite{babaioff2006combinatorial,ho2014adaptive,dutting2019simple,ZhuBYWJJ23}). 

A prime example of the growing relevance of contract theory in the online era is the \emph{creator economy},  projected to surpass half a trillion dollars by 2027 \citep{goldmansachs2022creatorEconomy}.
The creator economy refers to the ecosystem of individual content creators---such as YouTubers, Twitch streamers, TikTok influencers, and independent writers---who monetize their work through platforms that generate value for users and revenue for themselves. Online platforms aim to maximize long-term growth by encouraging creators to produce high-quality content. 
However, creators often prioritize minimizing production costs (e.g., by producing cheap, click-bait content or fake news) leading to a misalignment of incentives. 

Scenarios of this nature are captured by the classic hidden-action principal-agent model \citep{holmstrom1979,grossman1983}, where a principal (the platform) delegates a task (producing content) to an agent (the creator). 
In this model, the agent chooses from a set of actions (e.g., levels of effort in content production), each associated with a cost and a distribution over outcomes (e.g., engagement metric and audience growth) which benefit the principal.  
Since the principal 
observes only these outcomes, not the 
agent's actions, an asymmetry of information arises, creating a moral hazard problem: the agent 
has no inherent motivation to undertake costly actions that align with the principal's 
goals. Hence, the principal  
has to design a \textit{contract} (e.g., a revenue-sharing scheme) that specifies transfers---monetary rewards on observed outcomes---incentivizing agents 
to exert effort.

An additional challenge in applying this framework 
is that the principal typically interacts with a population of (independent) agents, 
each with their own \emph{type} 
(e.g., expertise, production costs, and audience profile). 
One possible approach is to adopt a \emph{Bayesian perspective}, 
assuming the principal has perfect knowledge of the distribution over agent types, and design a contract that performs well in expectation. This approach has been extensively studied by
\cite[e.g.,][]{alon2021contracts,castiglioni2022designing,GuruganeshSW023,AlonDLT23,CastiglioniCLXS25,GSW21}.

However, access to the full distribution of agent types is often an unrealistic assumption. A more practical and feasible approach is to use historical data in the form of {\em samples}, where each sample corresponds to an individual agent's type.
A similar shift from relying on full distributional knowledge to adopting a sample-based approach
has fueled the rise of \textit{automated mechanism design} \cite[e.g.,][]{ColeR14,balcan16,morgenstern2015pseudodimension,BalcanSV18,curry2024automateddesignaffinemaximizer} and \emph{data-driven algorithm design} \cite[e.g.,][]{BalcanDDKSV21,DBLP:books/cu/20/Balcan20}.
Inspired by this {line of} work, we {introduce a framework for offline learning in contract design settings}\footnote{There is an active line of work on \emph{online} learning for contracts (which we discuss in more detail below). {Some of our insights stem from comparisons with this earlier literature and the distinct feedback model it employs}.},
answering the following central question:

\begin{center}
     \textbf{Question 1:} \emph{How many samples from the agent type distribution are\\ required to learn a near-optimal contract {with high probability}?}
\end{center}

To address this question, it is essential to recognize that the principal's choices might be restricted to a specific class of contracts, e.g., due to legal regulations or the cognitive difficulty that general contracts might induce.
For example, a natural candidate for such a class is the class of \emph{linear} contracts, where the agent receives a fixed fraction of the principal's revenue.

Building on this, we propose a systematic approach to establishing sample complexity guarantees for specific contract classes by leveraging the \emph{pseudo-dimension} --- a combinatorial measure of complexity for a class of contracts.
Notably, understanding the intrinsic complexity of key contract classes is an important question in its own right, extending beyond its primary role in determining sample complexity guarantees, see \cite[e.g.,][]{dutting2019simple,caroll2015}.
This leads us to pose the following central question:

\begin{center}
\textbf{Question 2:} \emph{What is the pseudo-dimension of key contract classes?}    
\end{center}

We further explore whether the inherent complexity of key contract classes can be mitigated by approximating them with simpler classes. The quality of this approximation is measured by the  resulting loss in the principal's utility, referred to as the \emph{representation error} with respect to the original class.

\begin{center}
     \textbf{Question 3:} \emph{{Are there contract classes  with low pseudo-dimension \\ that closely approximate key contract classes?}}
\end{center}

In this work we provide answers to all three questions.

Our work also connects to the rapidly growing and impactful area of \emph{online} contract learning \cite[e.g.,][]{ho2014adaptive,ZhuBYWJJ23,bcmg23,ChenEtAl2024,scheid24incentivizedLearning}. 
The primary distinction between this line of research and our model lies in the feedback type: 
these works focus on repeated interactions with agents drawn from the underlying distribution, where the principal observes only the realized outcomes ({a setup commonly referred to as} {\emph{bandit feedback}}). 
In contrast, we study a one-shot problem in which a sample corresponds to the agent's type ({a setup often termed} {\emph{expert advice}} {or \emph{full feedback}}).
A key finding {from the online learning literature} is that finding a near-optimal contract {under the} bandit feedback model requires exponentially many rounds of adaptation \citep{ZhuBYWJJ23}. 
This lower bound applies even {in the case of} a single agent with an unknown type, highlighting that the {primary} challenge in this feedback model lies in learning the distributions over outcomes induced by individual actions, rather than learning the distribution of agent types.
In contrast, {offline learning with expert advice} focuses on the problem of learning the agent type distribution, which has the natural interpretation of learning a (near-)optimal contract from a small number of sample agents. 
We include a detailed comparison with this prior work in \Cref{sec:related_work,section:online}.

\subsection{Main Results: Bounds on Pseudo-Dimension and Sample Complexity}

We consider the hidden-action principal-agent model with $n$ actions and $m$ outcomes. 
The agent is characterized by a \emph{private} type, encoding their costs and associated outcome distributions for all potential actions. We assume that the agent's type is drawn from an \emph{unknown} distribution $\mathcal{D}$. The principal has sample access to $\mathcal{D}$ but does not know the agent's type or the actual distribution $\mathcal{D}$. Our goal is to learn a (single) contract from a prespecified contract class $\mathcal{C}$ that maximizes the principal's expected utility, with the expectation taken over the randomness of the agent's type.

\begin{table}[t]
\vspace{-1cm}
    \centering
    \renewcommand{\arraystretch}{1.3}
    \begin{tabular}{|c|c|c|c|}
    \hline
    \textbf{Contract Class} & \multicolumn{2}{c|}{\textbf{Pseudo-Dimension}} & \textbf{Sample Complexity} \\ 
    \hline 
    & \textbf{Full Class} & \textbf{$\bm{\epsilon}$-Approximation} & \\ 
    \hline
    $\mathcal{C}_{\mathrm{linear}} = \{\alpha \cdot r : \alpha \in [0,1]\}$ & $\Theta(\log n)$ & ${\Theta}(\log (1/\epsilon))$ & ${\widetilde{\Theta}((1/\epsilon^2) \log (1/\delta))}$ \\
     & {\footnotesize (\Cref{theorem:linear_contract_pdim,theorem:special_case})} & {\footnotesize(\Cref{theorem:t_contracts_true_pd,thm:pdim_lb_for_rep_error_eps})} & {\footnotesize(\Cref{thm:efficient_linear,thm:sample_lower_bound_linear})} \\ \hline
    $\Cbounded = [0,1]^m$ & $\widetilde{\Theta}(m \log n)$ & $\widetilde{\Theta}(m \log (1/\epsilon))$ & $\widetilde{\Theta} ((1/\epsilon^2)(m + \log(1/\delta))$ \\
     & {\footnotesize (\Cref{theorem:general_contract_pdim,theorem:special_case})} & {\footnotesize(\Cref{lemma:se_contracts_pseudo-dimension,thm:pdim_lb_for_rep_error_eps})} & {\footnotesize(\Cref{thm:cone_efficient,thm:sample_complexity_lb_bounded_contracts})} \\ \hline
    $\Cunbounded = [0, \infty)^m$ & $\widetilde{\Theta}(m \log n)$ & - & $\infty$ \\
     & {\footnotesize(\Cref{theorem:general_contract_pdim,theorem:special_case})} & - & {\footnotesize(\Cref{thm:sample_impossibility_unbounded_contracts})} \\ \hline
    \end{tabular}
    \caption{Overview of our upper and lower bounds on the pseudo-dimension and sample complexity of key classes of contracts. Here, the $\widetilde{\Theta}$ notation suppresses {the $\log(1/\epsilon)$ and $\log m$ factors in the bounds that already depend polynomially on $1/\epsilon$ and $m$, respectively. 
    The pseudo-dimension lower bounds hold for any $\epsilon$-approximation of the corresponding full contract class. 
    }
    }
    \label{tab:overview}
\end{table}

\begin{figure}[!ht]
\centering
\vskip -60pt
\includegraphics[width=1\textwidth]{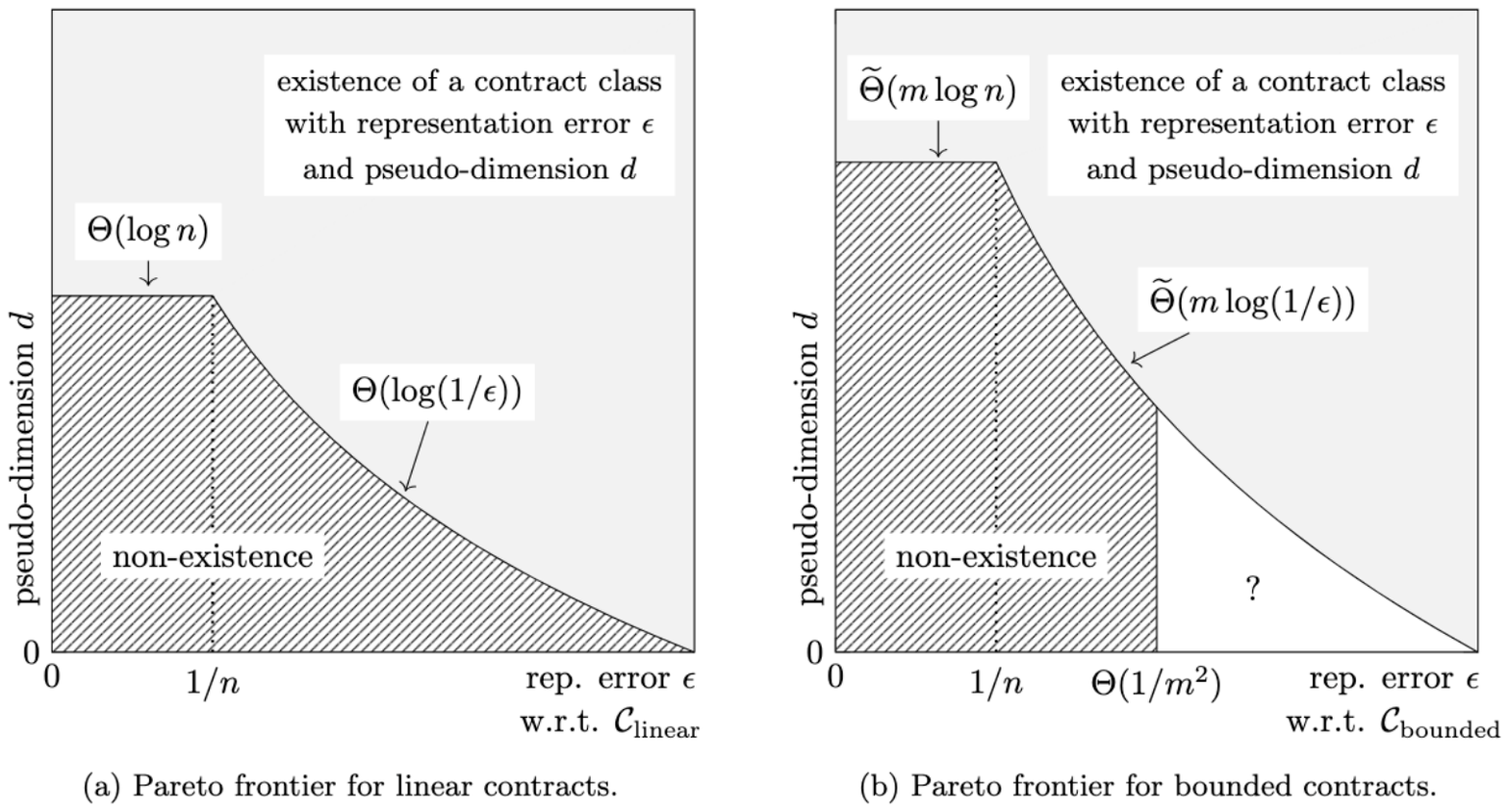}
\vskip -60pt
\caption{Trade-offs between pseudo-dimension and representation error for linear and bounded contracts in the case where $n > 24 \cdot m^2$.
For every point $(\epsilon,d)$, an $\epsilon$-approximation with pseudo-dimension $d$ exists if $(\epsilon,d)$ is in the gray region, and it does not exist if $(\epsilon,d)$ is in the diagonally-marked region.
For clarity, we suppress the logarithmic factors in this figure. Specifically, the boundaries between the two regions are tight up to a constant for linear contracts and up to a factor of $\log m$ for bounded contracts.}
    \label{fig:tradeoffs_summary}
\end{figure}
\vskip -50pt

{We focus on three main classes of contracts: linear, bounded, and unbounded. For each class, we provide two key results: trade-offs between the pseudo-dimension and the representation error with that class and sample complexity guarantees for learning a near-optimal contract. Additionally, we show that many of our results are nearly tight.
Those results are presented in \Cref{tab:overview} and \Cref{fig:tradeoffs_summary}.
}

While the results for linear, bounded, and unbounded {contracts} are the {centerpiece} of this work, our framework and technical tools (e.g., \Cref{lem:learning_with_oracle,{lem:online_learning_reduction}}) 
are applicable to a wide range of settings. 
We illustrate the versatility of our techniques through their extension to additional settings, as discussed in \Cref{sec:furtherresults}.

\paragraph{Linear Contracts.}
{A linear contract sets the transfer for each outcome to be {an} $\alpha$-fraction of the principal's reward, for some fixed $\alpha \in [0,1].$}
We show that the pseudo-dimension of the class of linear contracts is  $\Theta(\log n)$ (\Cref{theorem:linear_contract_pdim,theorem:special_case}). 
Additionally, we show that it is possible to alleviate the dependence on the number of actions $n$ from the pseudo-dimension by considering an $\epsilon$-approximation of the full class of linear contracts.
In \Cref{theorem:t_contracts_true_pd,thm:pdim_lb_for_rep_error_eps}, we map the Pareto frontier between pseudo-dimension and representation error, showing that for representation error $\epsilon$, the optimal pseudo-dimension is $\Theta(\min(\log n, \log (1/\epsilon))$.
A key insight from these results is that for a representation error $\epsilon < 1/n$, there is no discretization that gives a better pseudo-dimension than the full class of linear contracts. 
These results lead to a sample-efficient and polynomial-time algorithm for learning linear contracts:  Having access to  $ N = O\left( (1 / \epsilon^2) \left( \log \left( 1 / \epsilon \right) + \log \left( 1 / \delta  \right) \right)\right)$ samples, we can find, with probability $1-\delta$, a contract with utility at most $\epsilon$ away from that of the optimal linear contract (\Cref{thm:efficient_linear}). 
We further establish a nearly matching sample complexity lower bound, showing that $ \Omega((1/\epsilon^2) \log (1/\delta))$ samples are necessary (\Cref{thm:sample_lower_bound_linear}). 

\paragraph{Bounded Contracts.}
A bounded contract sets the transfer for each outcome to a value in $[0,1]$.
We show that the pseudo-dimension of the class of bounded contracts is 
$O (m\log n + m \log m)$ and $\Omega(m \log n)$ (\Cref{theorem:general_contract_pdim,theorem:special_case}). 
Similarly to our results for linear contracts, we show that it is possible to alleviate the dependence on the number of actions $n$ in the pseudo-dimension bound. 
We map the Pareto frontier between pseudo-dimension and representation error for the region where 
$\epsilon < 1/(24 \cdot m^2)$, showing that for representation error $\epsilon$, the optimal pseudo-dimension is $\Theta(\min(m \log n, m \log (1/\epsilon))$.
As with linear contracts, 
we observe that for $\epsilon < 1/n$, no discretization improves the pseudo-dimension compared to the full class of bounded contracts.
Notably, our upper bound applies to any $\epsilon$, not necessarily bounded by $1/(24 \cdot m^2)$. 
These results translate into a sample-efficient learning algorithm, which can also be made time-efficient with appropriate oracle access (\Cref{thm:bounded_efficient}): using $ N = O\left( (1 / \epsilon^2) \left( m \log \left( 1 / \epsilon \right) + \log \left( 1 / \delta  \right) \right)\right)$ samples, we can find, with probability $1-\delta$, a contract with utility at most $\epsilon$ away from that of the optimal linear contract (\Cref{thm:efficient_linear}). 
Finally, we establish a nearly matching sample complexity lower bound of ${\Omega}((1/\epsilon^2)(m + \log(1/\delta))$ (\Cref{thm:sample_complexity_lb_bounded_contracts}).

\paragraph{Unbounded Contracts.}
{Unbounded contracts represent the most general class, allowing arbitrary transfers for each outcome.}
For the class of unbounded contracts, perhaps surprisingly, we show that the pseudo-dimension is also 
$\Theta (m\log n + m \log m)$  (\Cref{theorem:general_contract_pdim,theorem:special_case}), identical to that of the class of bounded contracts. However, in strong contrast to the case of bounded contracts, we prove that no learning algorithm with finite sample complexity exists for this setting (\Cref{thm:sample_impossibility_unbounded_contracts}), revealing a fundamental difference between bounded and unbounded contracts. 

\paragraph{Discussion.}

Taken together, for the most widely used contract classes---linear and bounded---we derive tight trade-offs (up to logarithmic factors) between {contract expressiveness (representation error) and sample complexity (pseudo-dimension)}, providing a principled framework for balancing these objectives. In contrast, for unbounded contracts, we establish a fundamental negative result: no learning algorithm with finite sample complexity exists in this setting, {highlighting} a critical distinction from bounded contracts. Furthermore, these results extend the scope of the {prior} contract learning literature \cite[e.g.,][]{ZhuBYWJJ23,bcmg23,scheid24incentivizedLearning}.
{The complexity bounds in prior work were dominated by the {difficulty} of learning the distribution {over outcomes induced by the actions of a fixed agent type} under bandit feedback. In contrast, our results capture the difficulty of learning the distribution over agent types themselves.}

\subsection{Extensions of the Main Results}\label{sec:furtherresults}
In this section, we present extensions of our core results to other settings, highlighting the novel challenges addressed in this paper, the flexibility of our techniques, and their applicability to more complex and less-explored settings.

\paragraph{Combinatorial Contracts.}
In \Cref{sec:combinatorial_actions}, we extend our results to the well-established setting with combinatorial actions \cite[e.g.,][]{Duetting2021CombinatorialC,dutting2024combinatorial,ezra2023Inapproximability,DEFK24,DBLP:journals/corr/abs-2403-09794}, where the agent can choose a set of actions rather than a single action/effort level. 
We demonstrate how the pseudo-dimension analysis enables the designer to leverage structural insights about optimal contracts  \cite{Duetting2021CombinatorialC} to achieve significantly improved sample complexity compared to naive approaches. 
Specifically, we discover a novel connection between the number of critical values and the pseudo-dimension, demonstrating that the pseudo-dimension---and consequently the sample complexity---scales polynomially with the number of critical values rather than the number of actions. This yields an exponential improvement in many well-studied scenarios (e.g., when the success probability function is gross substitutes).

\paragraph{Menus of Contracts.} In \Cref{sec:menus}, we illustrate how pseudo-dimension can serve as a powerful tool for analyzing less well-understood contract classes. As an example, we provide pseudo-dimension and sample complexity upper bounds for menus of contracts with a fixed size, offering new insights into the learnability of these more general and powerful classes \citep{GuruganeshSW023}. 
Importantly, we show that the pseudo-dimension grows linearly with the menu size.

\paragraph{Online Learning.} 
In \Cref{section:online}, we adapt our offline learning framework to the online setting. Unlike prior work on online contract learning \cite[e.g.,][]{ZhuBYWJJ23,bcmg23,scheid24incentivizedLearning}, which {typically} operates under a bandit feedback model, our approach assumes a stronger \textit{expert advice model}. 
This enables our algorithm to deliver significantly improved performance, achieving a regret of $\widetilde{O}(\sqrt{m \cdot T})$ for bounded contracts, {corresponding} to a sample complexity that is polynomial in $1/\epsilon$.
In contrast, \citet*{ZhuBYWJJ23} demonstrate that in the bandit feedback setting, any algorithm incurs a regret of at least $\Omega(T^{1-1/(m+2)})$, {corresponding} to a {sample complexity that is} exponential  in $1/\epsilon$, even when the agent type distribution is concentrated on a single type.
The key takeaway from this comparison is that, in terms of sample complexity, the cost of learning the distribution of agent types is only polynomial, while the cost of learning the distribution of outcomes for individual actions is exponential.

\subsection{Our Techniques}

\paragraph{Upper Bounds and Delineability.} The key technique for our pseudo-dimension upper bounds is to apply the  \emph{delineability framework} introduced by \citet*{BalcanSV18}. 
{They observed that, in related settings in mechanism design, the objective function is linear over a bounded number of regions of the type space and demonstrated how to leverage this structure to derive pseudo-dimension bounds.}
In this work, we extend this idea to the contract design problem, showing that our model exhibits a similar structure to that of mechanism design. 
{Notably, while we demonstrate that this structure holds for general contracts, we find that in the case of linear contracts, the boundaries between those regions correspond precisely to the \emph{critical values} \cite{Duetting2021CombinatorialC}}.
While previous work has established a crucial connection between the number of critical values and the \emph{time} complexity of finding optimal (linear) contracts \cite{Duetting2021CombinatorialC,dutting2024combinatorial}, in this work, we uncover a novel connection between the number of critical values and the \emph{sample} complexity of learning near-optimal linear contracts.

\paragraph{Pseudo-Dimension Lower Bounds.}
One of our main technical achievements is showing that the {pseudo-dimension} upper bounds from the delineability framework are near-tight. Our lower bound construction for 
bounded contracts proceeds in two steps:
First, by considering suitable distributions over agent types, we prove that any $\epsilon/m$-approximation with $1/\epsilon \leq n$ contains a structured grid of $(1/\epsilon)^{m-1}$  contracts 
(\Cref{lem:pdim_lb_contract_construction}). Then, by carefully adjusting the 
{outcome distributions and costs of the different actions,}
we construct agent types that have high utility on any chosen subset of the grid and low utility elsewhere (\Cref{lem:pdim_lb_type_construction}). These observations directly yield a shattering construction.

\paragraph{Sample Complexity Lower Bounds.} {Beyond the pseudo-dimension lower bounds discussed above,}
our second technical contribution is demonstrating that pseudo-dimension yields nearly optimal bounds on sample complexity. Our lower bounds are established through a reduction of the contract design problem to distinguishing between a set of sufficiently close distributions of agent types based on samples.
Notably, while the distributions of agent types are close in total variation distance, the corresponding optimal contracts differ significantly.
Leveraging standard techniques from statistical learning theory, we establish lower bounds for distinguishing between these distributions (\Cref{lem:bernoulli_lb,lem:discrete_lb}).

\subsection{Related Literature} \label{sec:related_work}
Over several decades, a rich body of literature has developed around contract theory (see, e.g., \citet*{holmstrom1979,grossman1983}). 
Recently, there has been a significant focus on the computational aspects of contract design, as summarized in the survey of \citet*{DuettingFTC24}. For instance,~\citet*{babaioff2006combinatorial} examine the computational challenges involved in incentivizing a group of agents to coordinate on a task. Additional work on this and other combinatorial contract problems includes \cite{dutting2021complexity,Duetting2021CombinatorialC,duetting2022multi,ezra2023Inapproximability,dutting2024combinatorial,deo2024supermodular,DEFK24,DBLP:journals/corr/abs-2403-09794}.

Another important line of research explores how well simple contracts approximate the optimal utility. \citet*{dutting2019simple} find that the best linear contracts can closely match the performance of the optimal contracts under specific conditions, while \citet*{caroll2015} shows that linear contracts are robust even when actions are unknown.

A rich body of work, initiated by \citet*{CastiglioniM021}, \citet*{GSW21}, and \citet*{alon2021contracts} explores 
the problem of optimal contract design for a known distribution over agent types. 
Additional work on this topic includes \cite{alon2021contracts,castiglioni2022designing,GuruganeshSW023,AlonDLT23,CastiglioniCLXS25,GSW21}. A key finding of this line of work is that computing optimal deterministic contracts is typically NP-hard, while computing an (almost) optimal randomized contract is feasible.
The main differentiation of our work from that line of research is that we make the weaker assumption of only having samples from that distribution, and 
focus on deriving sample-efficient learning algorithms. 

Beyond computational issues, there has been interest in integrating statistical inference with contract theory. \citet*{schorfheide2012,schorfheide2016} address the problem of delegating statistical estimation to a strategic agent, and \citet*{frazier14}  investigate incentivizing learning agents in a bandit setting. \citet*{bates2022} recently demonstrated how hypothesis testing can be conducted through interactions between a principal and multiple agents.

The intersection of contract design and learning was pioneered in \citet*{ho2014adaptive}, who approach the problem from an online learning perspective. More recently, this problem was studied by \citet*{ZhuBYWJJ23}. One of the main results of this line of work is an exponential lower bound for learning bounded contracts from bandit feedback. This lower bound applies even if the agent's type is fixed, but requires an exponential number of actions. Motivated by this, \citet*{ChenEtAl2024} and \citet*{bcmg23} show that near-optimal bounded contracts can be learned with polynomially many samples (in the bandit feedback model), when either additional ``regularity'' assumptions are satisfied or there are only constantly many actions.
Additional work on online learning in contract design includes 
\cite{cohen2022,DuettingGuruganeshSchneiderWang23,burkett2024,scheid24incentivizedLearning}.
The main distinguishing feature of our work is the focus on type samples (expert advice).

This focus on learning from samples of the agent's type distribution aligns our work with the tradition of \textit{automated mechanism design}, a research direction initiated by \citet*{ColeR14}.
Particularly relevant are the studies by \citet*{morgenstern2015pseudodimension}, who introduce the notion of pseudo-dimension to auction design, and \citet*{BalcanSV18}, along with follow-up work \cite{BalcanDDKSV21}, which leverage structural properties of auction classes to derive pseudo-dimension bounds.

A related body of work focuses on machine learning-powered preference elicitation \citep{brero2021machinelearningpowerediterativecombinatorial}. Rather than relying on historical data, this line of research involves the auctioneer interacting iteratively with bidders, using the auction’s dynamic history to elicit bids that enable more efficient allocations. This framework has been applied to various auction formats \cite[e.g.,][]{soumalias2024pricesbidsvalueseverything,Soumalias_2024} and extended to broader combinatorial settings, such as assignment problems \cite[e.g.,][]{soumalias2024machinelearningpoweredcourseallocation}.

\section{Contract Design Preliminaries}\label{sec:contract_prelims}

Throughout this work, we use $[k]$ to denote the set $\{0, \ldots, k-1\}$ and $\Delta^k$ to denote the set of all probability distributions over $[k]$.

\subsection{The Principal-Agent Model}

In this {section}, we define the hidden-action principal-agent model \citep{holmstrom1979,grossman1983}.
{The principal delegates the execution of a costly project to an agent, who can choose an action from a set of actions $\setactions = \{0, \ldots, n - 1\}$.
The set of possible outcomes is $\setoutcomes = \{ 0, \ldots , m - 1 \}$ {for $m \geq 2$}.
Every outcome $j \in \setoutcomes$ is associated with a reward $r_j \in \mathbb{R}_{\geq 0}$, known to the agent and the principal.}
We assume that the reward for outcome $0$ is $r_0 = 0$, and that the is some outcome $j \in [m]$ with positive reward $r_j > 0$.
{The action chosen by the agent is hidden from the principal, who can observe only the final outcome.}

{Every action $i \in [n]$ is associated with a {\em private} probability distribution $f_i = (f_{i,0}, \ldots, f_{i, \numoutcomes - 1 }) \in \Delta^\numoutcomes$ over the outcomes, and also with a 
{\em private} cost $c_i \in \mathbb{R}_{\geq 0}$ that the agent incurs for taking action $i$.
{We refer to the ensemble of probability distributions} $f=(f_0,\ldots,f_{\numactions - 1})$ {as the agent's \textit{production function}}, 
and {to} $c=(c_0,\ldots,c_{\numoutcomes - 1})$ {as} 
the {agent's costs.} 
The agent's type is the tuple $\theta =(f,c)$ and is known only by the agent.
The principal observes the final outcome, but not the agent's type $\theta$ nor the action taken by the agent. }

{To motivate the {agent} to exert effort, the principal designs a contract}
$t = (t_0, \ldots, t_{\numoutcomes-1}) \in \mathbb{R}^\numoutcomes_{\geq 0}$ determining the agent's payment for each outcome. For a given contract $t$, if the agent with type $\theta = (f, c)$ takes action $i$ resulting in outcome $j$, the agent's utility is  $t_j - c_i$ and the principal's utility is $r_j - t_j$. We define the agent's and principal's expected utilities for action $i$ in contract $t$ as:
\begin{align*}
    u_a((f, c), t, i) &= \sum_{j \in \setoutcomes} f_{i,j} t_j - c_i \\
    u_p((f, c), t, i) &= \sum_{j \in \setoutcomes} f_{i,j} (r_j - t_j)
\end{align*}
We denote the agent's best response to contract $t$ by $i^\star(\theta, t) \in \argmax_{i \in \setactions} u_a(\theta, t, i)$. 
We assume that the agent always selects the best response action, breaking ties in {favor} of the principal, i.e., out of all actions in $\argmax_{i \in \setactions} u_a(\theta, t, i)$, action
$i^\star(\theta, t)$ is the one that maximizes $u_p(\theta, t, i)$. 
Slightly abusing notation, we denote the principal's expected utility for contract $t$ by $u_p(\theta, t) = u_p(\theta, t, i^\star(\theta, t))$, and the principal's expected reward as 
$r_p(\theta, t) = \sum_{j \in [m]} f_{i^\star(\theta, t), j}  r_j$.

As is common in the literature, we assume that for every agent's type $\theta = (f,c)$ the cost of action $0$ is $c_0 = 0$. This ensures that the agent's utility from a contractual relationship is at least $0$ (by non-negativity of payments).
We denote by $\Thetaall = (\Delta^{m})^n \times (\{0\} \times \mathbb{R}_{\geq 0}^{n-1})$ the space of all permissible agent types.

\subsection{Contract Classes} 

An important special case of the contracts defined above is linear contracts, where the principal transfers a fixed fraction of her reward to the agent.

\begin{definition}[Linear Contracts] \label{definition:linear_contract}
We say that a contract $t$ is linear if  $t = (\alpha \cdot r_0, \ldots, \alpha \cdot r_{m-1})$ for some $\alpha \in [0,1]$. We will refer to the linear contract by its parameter $\alpha$.
\end{definition}

The key concept in the analysis of linear contracts is that of critical values \citep{Duetting2021CombinatorialC}.

\begin{definition}[Critical Values]\label{def:critical_definition}
For a given agent type $\theta = (f,c)$, a value of $\alpha \in (0, 1]$ is called a critical value if the principal's expected reward changes at $\alpha$, i.e., 
\begin{align*}
r_p ((f,c), \alpha) \neq r_p ((f,c), \alpha - \epsilon) \quad \text{for any } \epsilon > 0.
\end{align*}
We denote the set of all critical values as ${C}_{f,c} = \{\alpha \in (0, 1]: \alpha \text{ is critical}\}$.
\end{definition}

\citet*{Duetting2021CombinatorialC} made the important observation that $r_p(\theta, \cdot)$ is a non-decreasing step function with the number of steps bounded by $n-1$. 
Formally, they proved: 

\begin{lemma}[\cite{Duetting2021CombinatorialC}]\label{corollary:linear_contracts_critical_values}
For any agent type $\theta = (f,c)$, there exists some $k < n$ and a series of $\alpha$ values $0 = \alpha_0 < \alpha_1 < \ldots < \alpha_k$ such that for every $x \in [0, 1]$, we have $r_p(\theta, x) = r_p(\theta, \alpha_i)$, where $i$ is the maximal index satisfying $\alpha_i \leq x$. In particular, the number of critical values is $|C_{f,c}| \leq n-1$. Moreover, the principal's expected reward $r_p(\theta, \alpha)$ is monotonically non-decreasing in $\alpha$.
\end{lemma}

Another significant class of contracts are those with bounded transfers, where without loss of generality we set
the bound to be $1$.
In this work, we study all three main classes of contracts, summarized in the following definition.

\begin{definition}[Contract Classes] \label{def:contract_classes}
We define: 
\begin{itemize}
    \item The class of linear contracts to be $\mathcal{C}_{\mathrm{linear}} = \{ \alpha \cdot r \,|\, \alpha \in [0,1]\}$.
    \item The class of bounded contracts to be $\Cbounded = [0,1]^m$.
    \item The class of unbounded contracts to be $\Cunbounded=\mathbb{R}_{\geq 0}^m$.
\end{itemize}
\end{definition}

Note that $\Cbounded \subseteq \Cunbounded$ and that $\Clinear \subseteq \Cbounded$ if $r_j \leq 1$ for all outcomes $j \in [m]$.
For any class of contracts $\mathcal{C}$, we define $\mathrm{OPT}(\mathcal{C}, \theta) = \max_{t \in \mathcal{C}} u_p(\theta, t)$ for any agent type $\theta = (f,c)$, and we define $\mathrm{OPT}(\mathcal{C}, \mathcal{D}) = \max_{t \in \mathcal{C}} \mathbb{E}_{\theta \sim \mathcal{D}}[u_p(\theta, t)]$ for any distribution $\mathcal{D}$ over agent types.

We also define the notion of a representation error, which measures the loss in the principal's utility incurred by restricting the contract space to some subset.

\begin{definition}[Representation Error {and $\epsilon$-Approximation}]\label{def:representation_error} Consider two contract classes $\bar{\mathcal{C}} \subseteq \mathcal{C}$. We say that the representation error of $\bar{\mathcal{C}}$ with respect to $\mathcal{C}$ is at most $\epsilon$ if for every distribution $\mathcal{D}$ over the agent's type space, we have $\mathrm{OPT}(\bar{\mathcal{C}}, \mathcal{D}) \geq \mathrm{OPT}(\mathcal{C}, \mathcal{D}) - \epsilon$.
{We say that $\bar{\mathcal{C}}$ is an $\epsilon$-approximation of $\mathcal{C}$ if it has representation error of at most $\epsilon$ with respect to $\mathcal{C}$.}
\end{definition}

\subsection{Binary Outcome Model}\label{sec:binaryoutcome}
A notable special case of the contracting model occurs when there are only two outcomes, i.e., $m = 2$, with rewards $r_0 = 0$ and $r_1 > 0$. We refer to this case as the binary outcome model, and we say that outcome 0 corresponds to failure, and outcome 1 corresponds to success. This special case encompasses many significant scenarios and has been extensively studied (see, e.g., \cite{Duetting2021CombinatorialC,dutting2024combinatorial}).

Previous work (see, e.g., \cite{Duetting2021CombinatorialC,DuettingFTC24}) made two important observations: 
first, that linear contracts are always optimal in the binary outcome model; and second, that when restricting attention to linear contracts, it is without loss of generality to consider the binary outcome model. We formalize and prove these observations as \Cref{lem:binary_outcome} in \Cref{sec:proofs_for_contracts_prelims}.

\section{Learning Theory Preliminaries}\label{sec:learning_theory_preliminaries}

In this section, we describe a general model for learning a parameterized class of real-valued functions with input space $\mathcal{I}$ and parameter space $\mathcal{P}$.
The class of functions is denoted as $\mathcal{F} : \mathcal{I} \times \mathcal{P} \to \mathbb{R}$, and for any parameter $p \in \mathcal{P}$, $\mathcal{F}(\cdot, p)$ is the function parameterized by $p$.

We describe the learning model in terms of a general class of functions $\mathcal{F}$ for clarity, but our focus is on the principal's utility as a function of the agent's type parameterized by chosen contracts. Specifically, we define $\mathcal{F} = u_p$ with the input space $\mathcal{I} \subseteq \Thetaall = (\Delta^{m})^n \times (\{0\} \times \mathbb{R}_{\geq 0}^{n-1})$ representing the agent's type space and the parameter space $\mathcal{P} \subseteq \Cunbounded$ representing the contract space.

Our primary objective is to solve the following problem: There is some unknown distribution $\mathcal{D}$ over the agent types. We want to design efficient learning algorithms that take a limited number of independent samples from $\mathcal{D}$ as input and return some contract $t$ that nearly maximizes the principal's expected utility $\mathbb{E}_{\theta \sim \mathcal{D}}[u_p(\theta, t)]$ with high probability.

In the general model, this corresponds to the following problem: There is a distribution $\mathcal{D}$ over the input space $\mathcal{I}$, and we aim to find a parameter $p \in \mathcal{P}$ that  nearly maximizes $\mathbb{E}_{i \sim \mathcal{D}}[\mathcal{F}(i,p)]$.
Slightly abusing the notation,  we denote $\mathrm{OPT}(\mathcal{F}, \mathcal{D}) = \max_{p \in \mathcal{P}} \mathbb{E}_{i \sim \mathcal{D}}[\mathcal{F}(i,p)]$.

\subsection{Uniform Learnability} 

The first crucial concept for establishing our result is uniform learnability.

\begin{definition}[Uniform Learnability]
    We say that $\mathcal{F}$ is $(\epsilon, \delta, N)$-uniformly learnable for some $\epsilon, \delta, N > 0$ if for every distribution $\mathcal{D}$ over the input space $\mathcal{I}$ it holds that:
    \begin{align*}
      \mathbb{P}_{S \sim \mathcal{D}^N}\left[\forall_{p \in \mathcal{P}} \left| \frac{1}{|S|} \sum_{i \in S} \mathcal{F}(i,p) - \mathbb{E}_{i \sim \mathcal{D}}[\mathcal{F}(i,p)]\right| \leq \epsilon\right] \geq 1-\delta.
  \end{align*}
\end{definition} \label{def:uniform_learnability}

In words, $\mathcal{F}$ is $(\epsilon, \delta, N)$-uniformly learnable if, given a set of $N$ samples from the input distribution, with probability at least $1 - \delta$, the expected value of all members of that class $\mathcal{F}$ is at most $\epsilon$ away from their empirical value on the sample set. 

The essential connection between our problem and uniform learnability is that, for a uniformly learnable class of functions, the parameter that maximizes the function with respect to the empirical distribution is near-optimal with respect to the true distribution with high probability. This is formalized in the following standard lemma. 
We include a proof for completness in \Cref{sec:proofs_from_learning_theory_preliminaries}.

\begin{restatable}
{lemma}{empiricalmaximization}\label{lem:empirical_maximization}
Let $\mathcal{F}$ be a class of functions that is $(\epsilon, \delta, N)$-uniformly learnable for some $\epsilon, \delta, N > 0$. Consider any distribution $\mathcal{D}$ over $\mathcal{I}$. Let $S \sim \mathcal{D}^N$ be a collection of $N$ independent samples drawn from $\mathcal{D}$, and let $\widehat{p} \in \mathcal{P}$ be the parameter that maximizes $(1/|S|) \cdot \sum_{i \in S} \mathcal{F}(i, p)$ over all parameters $p \in \mathcal{P}$. Then, with probability at least $1-\delta$ over the random samples, we have $\mathbb{E}_{i \sim \mathcal{D}}[\mathcal{F}(i, \widehat{p})] \geq  \mathrm{OPT}(\mathcal{F}, \mathcal{D}) - 2 \cdot \epsilon$.
\end{restatable}

\subsection{Pseudo-Dimension}

In order to establish uniform learnability guarantees, we use the notion of the pseudo-dimension. Before we define pseudo-dimension, we need to define the concept of shattering.

\begin{definition}[Shattering]
The class of functions $\mathcal{F}$ shatters a set of $k$ inputs $\{i_1, \ldots, i_k\} \subseteq \mathcal{I}$ if there are some thresholds $\tau_1, \ldots, \tau_k \in \mathbb{R}$ such that for all $S \subseteq \{1, \ldots, k\}$ there is a parameter $p \in \mathcal{P}$ so that (1) for all $j \in S$ we have $F(i_j, p) \geq \tau_j$ and (2) for all $j \notin S$ we have $F(i_j, p) < \tau_j$.
\end{definition}

We are now ready to define pseudo-dimension of $\mathcal{F}$.

\begin{definition}[Pseudo-Dimension of Real-Valued Functions]
    The pseudo-dimension of a class of functions $\mathcal{F}$ is the size of the largest subset of the input space $\mathcal{I}$ that $\mathcal{F}$ can shatter. We denote the pseudo-dimension of $\mathcal{F}$ as $\pdim{\mathcal{F}}$.
\end{definition} 

This definition can be extended to classes of contracts as follows.

\begin{definition}[Pseudo-Dimension of Contracts]\label{def:pdim_of_contracts}
    The pseudo-dimension of a class of contracts $\mathcal{C} \subseteq \mathbb{R}^m_{\geq 0}$ 
    for an agent's type space $\Theta \subseteq \Thetaall$
    is the pseudo-dimension of the class of functions given by $\mathcal{F} = u_p$ where the input space is $ \mathcal{I} = \Theta$ and the parameter space is $\mathcal{P} = \mathcal{C}$. Slightly abusing the notation, we denote the pseudo-dimension of $\mathcal{C}$ for type space $\Theta$ as $\pdim{ \mathcal{C}, \Theta}$.
    When discussing the general type space $\Theta_{\mathrm{all}}$, we omit the type space, and  simply write $\pdim{\mathcal{C}} = \pdim{\mathcal{C}, \Theta_{\mathrm{all}}}$.
\end{definition}

We first prove two basic observations following from the definition of the pseudo-dimension. For proofs, see \Cref{sec:proofs_from_learning_theory_preliminaries}.

\begin{restatable}{lemma}{pdimproperties}\label{lem:pdim_properties}
Let $\mathcal{F}$ be a class of functions with input space $\mathcal{I}$ and parameter space $\mathcal{P}$. Then:
    \begin{enumerate}[label=(\alph*)]
        \item It holds that $\pdim{\mathcal{F}} \leq \log |\mathcal{F}|$. \label{enum:pdim_size}
        \item For any classes of functions such that  $\mathcal{F'} \subseteq \mathcal{F}$ it holds that $\pdim{\mathcal{F'}} \leq \pdim{\mathcal{F}}$.
        \label{enum:pdim_mono}
    \end{enumerate}
\end{restatable}

The following theorem says that pseudo-dimension bounds allow us to establish uniform learning guarantees, and hence sample complexity bounds (see \Cref{lem:empirical_maximization}).
For a proof, see, e.g., \citet*{DBLP:journals/jcss/LiLS01}.

\begin{theorem}[\cite{DBLP:journals/jcss/LiLS01}] \label{theorem:uniform_learnability}
    Fix any class of functions $\mathcal{F}$ with pseudo-dimension $\pdim{\mathcal{F}}$.
    Let $H \in \mathbb{R}_{> 0}$ be an upper bound on the absolute values of $\mathcal{F}$, i.e., $H$ is so that $|\mathcal{F}(i,p)| \leq H$ for all $i \in \mathcal{I}$ and $p \in \mathcal{P}$. We get the following equivalent guarantees:
    \begin{enumerate}
        \item For every $\epsilon > 0$ and $\delta > 0$, the class of functions $\mathcal{F}$ is $(\epsilon, \delta, N)$-uniformly learnable for:
        \begin{align*}
             N = O\left(\left(\frac{H}{\epsilon}\right)^2\left( 
             \pdim{\mathcal{F}} + \ln \left( \frac{1}{\delta}\right)\right) \right).
        \end{align*}
        \item  For every $\delta > 0$ and $N > 0$, the class of functions $\mathcal{F}$ is $(\epsilon, \delta, N)$-uniformly learnable for:\label{enum:epsilon_unif_learnability}
        \begin{align*}
             \epsilon = O\left(\frac{H}{\sqrt{N}}\left( \sqrt{\pdim{\mathcal{F}}} + 
            \sqrt{\ln \left( \frac{1}{\delta}\right)} \right) \right).
        \end{align*}
    \end{enumerate}
\end{theorem}

\subsection{Delineability}

The main technical tool that we use to derive pseudo-dimension bounds is the notion of delineability due to \citet*{BalcanSV18}.

\begin{definition}[Delineability] \label{definition:delineable}
    The class of functions $\mathcal{F}$ is $(d,k)$-delineable for some $d, t \in \mathbb{N}$ if:
    \begin{enumerate}
        \item the parameter space $\mathcal{P} = \mathbb{R}^d$, and
        \item for every input $i \in \mathcal{I}$, 
        there is a set $\mathcal{H}$ of $k$ hyperplanes so that the function $\mathcal{F}(i, \cdot)$ on the parameter space $\mathcal{P}$ is linear over every connected component of $\mathcal{P} \setminus \mathcal{H}$.
    \end{enumerate}
\end{definition} 

\citet*{BalcanSV18} show the following crucial result:

\begin{theorem}[\cite{BalcanSV18}]\label{theorem:delineable_pseudo-dimensions}
    For any $(d,k)$-delineable class of functions $\mathcal{F}$, we have that the pseudo-dimension of $\mathcal{F}$ is $\pdim{\mathcal{F}} \leq O(d\log(dk))$.
\end{theorem}

\section{Pseudo-Dimension Results}\label{sec:pdim}

In this section, we provide a tight analysis of the pseudo-dimension of the main contract classes ($\Clinear$, $\Cbounded$, $\Cunbounded$). Moreover, by considering suitable discretizations of linear and bounded contracts ($\Lepsilon$ for $\Clinear$ and $\Bepsilon$ for bounded contracts), we derive the Pareto frontier between the pseudo-dimension and the representation error with respect to both $\Clinear$ and $\Cbounded$.

First, in \Cref{sec:pdim_upper_bounds}, we provide upper bounds on the pseudo-dimension of the full contract classes ($\Clinear$ , $\Cbounded$ and $\Cunbounded$) and their discretizations ($\Lepsilon$ and $\Bepsilon$). Then, in \Cref{sec:pdim_lower_bounds}, we establish nearly matching lower bounds, demonstrating the tightness of our results.

\subsection{Upper Bounds on the Pseudo-Dimension}\label{sec:pdim_upper_bounds}

\paragraph{Linear Contracts.} We first state the upper bound on the pseudo-dimension of the class of linear contracts. Crucially, our upper bound depends on the number of critical values (see \Cref{def:critical_definition}), which is at most $n-1$ in general (see \Cref{corollary:linear_contracts_critical_values}). Moreover, better bounds on the number of critical values are known for important special cases, which we discuss in \Cref{sec:combinatorial_actions}.

\begin{theorem}[Pseudo-Dimension Upper Bound for $\Clinear$] \label{theorem:linear_contract_pdim}
For any agent's type space $\Theta$, let $k$ be an upper bound on the number of critical values for any type in $\Theta$, i.e., $k$ is such that $|C_{f,c}| \le k$ for all $\theta = (f,c) \in \Theta$.
Then, the pseudo-dimension
of the class of linear contracts for the type space $\Theta$ is 
$\pdim{\mathcal{C}_{\mathrm{linear}}, \Theta} = O (\log k)$. In particular, it holds that $\pdim{\mathcal{C}_{\mathrm{linear}}} = O (\log n)$. 
\end{theorem}

\begin{proof}
We first argue that the class of linear contracts $\mathcal{C}_{\text{linear}}$
is $(1, k + 1 )$-delineable.

Fix any agent type $\theta = (f,c) \in \Theta$.
According to
\Cref{corollary:linear_contracts_critical_values},
there exists a series of $\alpha$ values $0 = \alpha_0 < \alpha_1 < \ldots < \alpha_k \le 1$ such that $r_p((f,c), \cdot)$ is a step function with steps exactly at those $\alpha$ values.
The principal's expected utility for agent type $(f, c)$ is $u_p((f,c), \alpha) = (1 - \alpha) \cdot r_p((f,c), \alpha)$. 
 Thus, for any $\alpha \in [\alpha_j, \alpha_{j+1})$, the principal's utility is linear in $\alpha$.
Therefore, there exists a set $\mathcal{H}$ of at most $k+1$ hyperplanes (precisely the critical values of $\alpha$), such that for any connected component $\mathcal{P}'$ of $\mathcal{P} \setminus \mathcal{H}$, the principal's expected utility for agent type $(f, c)$ is linear over $\mathcal{P}'$.
According to \Cref{definition:delineable}, $\mathcal{C}_{\text{linear}}$ is $(1, k + 1)$-delineable.

From \Cref{theorem:delineable_pseudo-dimensions}, we have that $\pdim{\Clinear, \Theta } = O (\log k).$
For the second result, we can apply \Cref{corollary:linear_contracts_critical_values} to get that the number of critical values is always at most $n-1$. 
\end{proof}

\paragraph{Discretizing Linear Contracts.} 
Our first pseudo-dimension bound in this section (\Cref{theorem:linear_contract_pdim}) depends on the number of actions, $\numactions$, which can be problematic in settings with large or continuous action spaces. 
To address this, we restrict oursearch space to a suitably discretized subclass of linear contracts. We begin by defining this discretized class.

\begin{definition}[$\epsilon$-Discrete Linear Contracts]\label{def:t-discretized-linear}
{Fix any $0 < \epsilon \leq 1$.
We say that a linear contract $\alpha$ is $\epsilon$-discrete if $\alpha = 1$ or  $\alpha = k \cdot \epsilon$ for some $k \in \mathbb{N}_0$ such that $k \cdot \epsilon \le 1$. 
We denote by $\Lepsilon = \{k \cdot \epsilon: k \in \mathbb{N}_{0}, k \cdot \epsilon \le 1 \} \cup \{ 1 \} $ the set of all $\epsilon$-discrete linear contracts.}
\end{definition}

First, we establish a representation error guarantee (see \Cref{def:representation_error}), showing that there is always an $\epsilon$-discrete linear contract at most $\epsilon$ away from the optimal contract.

\begin{lemma}
\label{lemma:t_contracts_representation_error}
The class of $\epsilon$-discrete contracts $\Lepsilon$ is an $\epsilon$-approximation of $\Clinear$, i.e., the representation error of $\Lepsilon$ with respect to $\Clinear$ is at most $\epsilon$.
\end{lemma}
\begin{proof}
Let $\mathcal{D}$ be a distribution over the agent types.
Let $\alpha^\star \in \mathcal{C}_{\text{linear}}$ be the optimal linear contract for $\mathcal{D}$.
Let $\widehat{\alpha} = \min \{ \alpha \in \mathcal{L}_\epsilon: \alpha \ge \alpha^\star \}$.
Note that $\widehat{\alpha}$ is well-defined because $\alpha^\star \le 1$ and $1 \in \mathcal{L}_\epsilon$.
Furthermore, $\alpha^\star \le \widehat{\alpha} \le \alpha^\star + \epsilon$. Observe that:
\begin{align*}
    u_{p}(\mathcal{D}, \widehat{\alpha}) &= \mathbb{E}_{\theta \sim \mathcal{D}}\left[ u_p(\theta, \widehat{\alpha}) \right] \\
    & = \mathbb{E}_{\theta \sim \mathcal{D}}\left[ (1 - \widehat{\alpha}) \cdot r_p (\theta, \widehat{\alpha}) \right]  \\ 
    & \ge  \mathbb{E}_{\theta \sim \mathcal{D}}\left[ (1 - \widehat{\alpha}) \cdot r_p (\theta, \alpha^\star) \right]  && (\text{by \Cref{corollary:linear_contracts_critical_values} since $\alpha^\star \leq \widehat{\alpha}$})\\ 
    & \ge \mathbb{E}_{\theta \sim \mathcal{D}}\left[ (1 - \alpha^\star - \epsilon ) \cdot r_p (\theta, \alpha^\star) \right]  && (\text{since } \widehat{\alpha} \leq \alpha^\star + \epsilon) \\ 
    & = \mathbb{E}_{\theta \sim \mathcal{D}}\left[ u_{p}(\theta, \alpha^\star) - \epsilon \cdot r_p (\theta, \alpha^\star) \right]  \\ 
    & \ge  \mathbb{E}_{\theta \sim \mathcal{D}}\left[ u_{p}(\theta, \alpha^\star) - \epsilon \cdot 1 \right] && (\text{since $r_j \leq 1$ for all $j \in \setoutcomes$})\\
    &= \mathrm{OPT}(\mathcal{C}_{\mathrm{linear}}, \theta) - \epsilon && (\text{by the choice of } \alpha^\star)
\end{align*}
which proves the lemma.
\end{proof}

Next, we establish a pseudo-dimension upper bound for the class:

\begin{theorem}[Pseudo-Dimension Upper Bound for $\Lepsilon$] \label{theorem:t_contracts_true_pd}
    The class of $\epsilon$-discrete linear contracts $\Lepsilon$ has pseudo-dimension at most $\pdim{\Lepsilon} = O(\min( \log n, \log (1/\epsilon)))$.
\end{theorem}

\begin{proof}
Note that $\Lepsilon$ contains $\lceil 1/\epsilon \rceil + 1$ contracts. Hence, by \Cref{lem:pdim_properties}\ref{enum:pdim_size}, we have that $\pdim{\Lepsilon} \leq \log(1/\epsilon)$. Additionally, by \Cref{theorem:linear_contract_pdim} and \Cref{lem:pdim_properties}\ref{enum:pdim_mono}, we have that $\pdim{\Lepsilon} \leq \pdim{\mathcal{C}_{\mathrm{linear}}} = O(\log n)$  since $\Lepsilon \subseteq \mathcal{C}_{\mathrm{linear}}$. Therefore, the result follows.
\end{proof}

Given the results of \Cref{lemma:t_contracts_representation_error} and \Cref{theorem:t_contracts_true_pd}, the discretization hyperparameter $\epsilon$ in $\Lepsilon$ provides a principled way to trade off the simplicity of a contract class, as measured by its pseudo-dimension, and its representation error.

\paragraph{Bounded and Unbounded Contracts.} 
We now establish an upper bound on the pseudo-dimension of $\Cunbounded$, which immediately gives an upper bound on the pseudo-dimension of $\Cbounded$.

\begin{theorem}[Pseudo-Dimension Upper Bound for $\Cunbounded$ and $\Cbounded$]\label{theorem:general_contract_pdim}
The pseudo-dimension of $\Cunbounded$, and hence also of $\Cbounded$, is at most 
$O (m\log n + m \log m)$.
\end{theorem}

\begin{proof}
We first show that $\Cunbounded$ is $(m,n^2)$-delineable.

Fix any agent type $\theta = (f,c)$.
A contract $t$ is parameterized by $(t_0, \ldots, t_{\numoutcomes - 1}) \in \mathbb{R}_{\ge 0}^m$.
We define $\binom{n}{2}$ hyperplanes as follows. For 
any two actions $i,w \in \setactions$, we define $H_{i,w} = \{ t \in \mathbb{R}_{\ge 0 }^m \,:\, t \cdot f_i - c_i = t \cdot f_w - c_w\}$, which is the hyperplane of contracts where the agent with type $\theta$ is indifferent between taking action $i$ and action $w$. 
Clearly, for any contract $t$ in a connected component of $\mathbb{R}^m \setminus \bigcup_{i,w\in\setactions} H_{i,w}$, 
the best response $i^{*}(\theta, t )$ is constant, 
and hence the principal's utility is linear across that connected component, since it's equal to 
$(r-t) \cdot f_{i^{\star}(\theta, t)}$. 
Thus, $\Cunbounded$ is $(\numoutcomes,\numactions^2)$-delineable.

Since, as we have just shown,
the class of all contracts $\Cunbounded$ is $(m,n^2)$-delineable,
\Cref{theorem:delineable_pseudo-dimensions}  implies that 
$\pdim{\Cunbounded} = O (m \log n +  m \log m) $.
The result for $\Cbounded$ follows by \Cref{lem:pdim_properties}\ref{enum:pdim_mono} since $\Cbounded \subseteq \Cunbounded$.
\end{proof}

\paragraph{Discretizing Bounded Contracts.} 
Similarly to our approach for linear contracts, we now demonstrate how to eliminate the dependence on the number of actions $n$ in the pseudo-dimension bound for $\Cbounded$. This is achieved by using a discretized class introduced by \citet*{ZhuBYWJJ23} which  ensures low representation error relative to 
$\Cbounded$ (see \Cref{def:representation_error}).

\begin{theorem}[Discretization of \cite{ZhuBYWJJ23}]  \label{thm:discretization}
 Assume that the rewards are $r_0 = 0$ and $r_1, r_2, \ldots, r_{m-1} \leq 1$. 
Let $\epsilon > 0$ be such that $1/\epsilon$ is an integer.  There is some set of unit vectors $\mathcal{V}_{\epsilon}\subseteq\mathbb{R}^m$ of size at most $\Theta((1/\epsilon)^{C \cdot m})$ for some constant $C$
such that for any contract $t \in \Cbounded$, one can find a vector $\gamma \in \mathcal{V}_{\epsilon}$ such that $\langle t - r, \gamma \rangle \geq \cos(\epsilon) \cdot \|t - r\|_2$ and $\|\gamma\|_2 = 1.$
 Moreover, $\mathcal{V}_{\epsilon}$ can be computed in time polynomial in $(1/\epsilon)^{C \cdot m}$. Let $\bar{\mathcal{B}}_{\epsilon} = \{ {r} + \sqrt{m} \cdot \beta \cdot {\gamma} \,|\, \beta \in \{\epsilon, 2 \epsilon, \ldots, 1\}, {\gamma} \in \mathcal{V}_{\epsilon^2} \} \cap \Cbounded$. 
 Then, {the discretized class $\bar{\mathcal{B}}_{\epsilon}$ is an $20\epsilon\sqrt{m}$-approximation of $\Cbounded$, i.e.,} the representation error of $\bar{\mathcal{B}}_{\epsilon}$ {with respect to $\Cbounded$} is at most $20 \epsilon \sqrt{m}$. 
\end{theorem}

Next, we prove that this class allows us to achieve low representation error with respect to $\Cbounded$ while maintaining a low pseudo-dimension:

\begin{theorem} [Pseudo-Dimension Upper Bound for $\Bepsilon$]\label{lemma:se_contracts_pseudo-dimension}
 Assume that the rewards are $r_0 = 0$ and $r_1, r_2, \ldots, r_{m-1} \leq 1$. 
Let $\Bepsilon = \bar{\mathcal{B}}_{\epsilon/(20  \sqrt{m})}$. 
Then, the discretized class $\mathcal{B}_{\epsilon}$ is an $\epsilon$-approximation of $\Cbounded$, i.e., the representation error of $\Bepsilon$ with respect to $\Cbounded$ is at most $\epsilon$, and 
it holds that 
$\pdim{\Bepsilon} = O \left ( m \log \left ( {1}/{\epsilon} \right ) + m \log m \right )$.
\end{theorem}
\begin{proof}
    Let $\bar{\epsilon} = \epsilon/(20  \sqrt{m})$.
    By \Cref{thm:discretization}, the representation error of $\Bepsilon$ with respect to $\Cbounded$ is $20\bar{\epsilon}\sqrt{m} = \epsilon$. Moreover, the cardinality of $\Bepsilon$ is $|\Bepsilon| = |\mathcal{V}_{\bar{\epsilon}^2}| \cdot (1/\bar{\epsilon}) = O((1/\bar{\epsilon}^2)^{C \cdot m} \cdot (1/\bar{\epsilon})) = O((400 m/\epsilon^2)^{C \cdot m} \cdot (20\sqrt{m}/\epsilon))$.
    By 
\Cref{lem:pdim_properties}\ref{enum:pdim_size}, we get that $\pdim{\Bepsilon} \leq \log |\Bepsilon| = O(m \log (1/\epsilon) + m \log m)$, as needed.
\end{proof}

\subsection{Lower Bounds on the Pseudo-Dimension}\label{sec:pdim_lower_bounds}

{In this section, we establish nearly matching lower bounds for the pseudo-dimension upper bounds we developed in \Cref{sec:pdim_upper_bounds}.

First, in \Cref{lem:pdim_lb_contract_construction}, we show that any contract class with low representation error relative to $\Cbounded$ must include contracts with a specific structured form.}

\begin{restatable}
{lemma}{pdimlbcontractconstruction}\label{lem:pdim_lb_contract_construction}
Set the rewards to $r_0 = 0$ and $r_1, r_2, \ldots, r_{m-1} \geq 1$. 
Fix any sequence $\alpha_0, \alpha_1, \ldots, \alpha_{m-1} \geq 0$ and any $\epsilon <  \min_{j \in [m] \setminus \{0\}} (r_j - \alpha_j)$.
Let $\mathcal{C}$ be any $\epsilon/\numoutcomes$-approximation of $\Cbounded$, i.e., the representation error of $\mathcal{C}$ with respect to $\Cbounded$ is at most $\epsilon / \numoutcomes$.
Then, there must exist a contract $t \in \mathcal{C}$ such that $t_0 \in [0, \epsilon]$ 
and $t_j \in [\alpha_j, \alpha_j + \epsilon]$ for all $1 \leq j \leq m-1$. 
\end{restatable}

\begin{proof}
    The proof is deferred to \Cref{sec:proofs_for_pdim}.
\end{proof}

{Next, in \Cref{lem:pdim_lb_type_construction}, we construct agent types that allow us to differentiate between the structured contracts identified in \Cref{lem:pdim_lb_contract_construction}.} 

\begin{restatable}
{lemma}{pdimlbtypeconstruction} \label{lem:pdim_lb_type_construction}
Assume that the rewards are $r_0 = 0$ and $r_1, r_2, \ldots, r_{m-1} \geq 1$.
    Fix a number $\ell < \numactions$, an outcome $j \in [m] \setminus \{0\}$, and any 
    sequence $0 = \alpha_0 < \alpha_1 < \alpha_2 < \ldots < \alpha_{\ell} < \alpha_{\ell+1} = r_j$.
    Let $\rho = (1/3) \cdot \min(r_j-\alpha_\ell, 1) \cdot \min_{0 \leq i \leq \ell-1} (\alpha_{i+1} - \alpha_i)$.
    Then, for any subset $S \subseteq \{1, 2, \ldots, \ell\}$, there is some agent type $\theta^{(\alpha, S, j)} = (f^{(\alpha, j)},c^{(\alpha, S)})$ such that for every contract $t$ and $i \in \{0, 1, \ldots, \ell\}$:
    \begin{itemize}
        \item if $t_j - t_0 \in [\alpha_i, \alpha_{i} + \rho]$ and $i \in S$ and $t_0 < \rho$, we have $u_p(\theta^{(\alpha, S, j)}, t) \geq (r_j-\alpha_{\ell}) - 2 \cdot \rho$, and
        \item if $t_j - t_0 \in [\alpha_i,\alpha_{i+1})$ and $i \notin S$, we have $u_p(\theta^{(\alpha, S, j)}, t) \leq (r_j-\alpha_{\ell}) - 3 \cdot \rho$.
    \end{itemize}
\end{restatable}

\begin{proof}
    The proof is deferred to \Cref{sec:proofs_for_pdim}.
\end{proof}

{By combining these two lemmas, we can construct sufficiently large shatterings of agent types, which directly imply the desired pseudo-dimension lower bounds of  $\Clinear$, $\Cbounded$, and $\Cunbounded$. }

\begin{theorem}[Pseudo-Dimension Lower Bound for Full Contract Classes]\label{theorem:special_case}
     Assume that the rewards are $r_0 = 0$ and $r_1, r_2, \ldots, r_{m-1} \geq 1$. 
     {Then, the pseudo-dimension of any $1/(24 \cdot n \cdot m)$-approximation of $\Cbounded$ is at least $\Omega(m \log n)$.} 
     In particular, the pseudo-dimension of $\Cbounded$, and thus $\Cunbounded$, is $\Omega(m \log n)$. Moreover, the pseudo-dimension of $\Clinear$ is $\Omega(\log n)$. 
\end{theorem}
\begin{proof}
If $n$ is not a power of $2$, we can reduce $n$ to the highest power of $2$ smaller than $n$. Henceforth, assume $n$ is a power of $2$.
Define the subsets \( S_1, \ldots, S_{\log n} \subseteq \{0, \ldots, n-1\} \) in the following way:
\begin{align*}
    S_b = \{ i \in \{0, \ldots, n-1\} : \text{the $b$-th bit of $i$ in binary representation is $1$} \}
\end{align*}
for \( 1 \leq b \leq \log n \). Note that $0 \notin S_b$ for all $1 \leq b \leq \log n$.

We consider a collection of $(m-1)\log n$ agent types constructed using \Cref{lem:pdim_lb_type_construction} as follows.
First, we define a sequence $\alpha$ by letting $\alpha_i = i / (2 \cdot n)$ for $i \in \{0, 1, \ldots, n-1\}$ and $\alpha_{n} = 1$. Let $\rho_j = (1/3) \cdot \min(r_{j} - \alpha_{n-1}, 1) \cdot \min_{0 \leq i \leq n-2} (\alpha_{i+1}-\alpha_i) = \min(r_{j} - \alpha_{n-1}, 1)/(6 \cdot n)$. Note that $\rho_j \leq 1/ (12 \cdot n)$ for all $j \in [m] \setminus \{0\}$, since $\alpha_{n-1} < 1/2$ and $\alpha_{i+1} - \alpha_{i} = 1/(2 \cdot n)$ for all $i \in \{0, \ldots, n-2\}$.
Since $S_b \subseteq \{1, \ldots, n-1\}$, 
by \Cref{lem:pdim_lb_type_construction}, for each outcome $j \in [m]\setminus \{0\}$ and $1 \leq b \leq \log n$, there is some agent type $\theta^{(\alpha, S_b, j)}$ such that for every contract $t$ and $i \in \{0, \ldots, n-1\}$:
\begin{itemize}
    \item if $t_j - t_0 \in [i/(2 \cdot n), i/(2 \cdot n) + 1/(12 \cdot n)]$ and $i \in S$ and $t_0 < 1/(12 \cdot n)$, we have $u_p(\theta^{(\alpha, S_b, j)}, t) \geq (r_j - \alpha_{n-1}) - 2/(12 \cdot n)$, and
    \item if $t_j - t_0 \in [i/(2 \cdot n), i/(2 \cdot n) + 1/(2 \cdot n))$ and $i \notin S$, we have $u_p(\theta^{(\alpha, S_b, j)}, t) \leq (r_j - \alpha_{n-1}) - 3/(12 \cdot n)$.
\end{itemize}
 To prove the lemma, it is sufficient to show that $\mathcal{C}$ shatters this collection of types. 

To prove this, take any subset of the agent types $X \subseteq \{1, \ldots, \log n\} \times \{1, \ldots, m-1\}$. 
Let $X_j = \{b \in \{1, \ldots, \log n\} : (b,j) \in X\}$. 
Define $k_j = \sum_{b \in X_j} 2^{b-1}$. Note that 
$0 \leq k_j \leq n-1 = \sum_{1 \le b \le \log n} 2^{b-1} $, and 
$k_j \in S_b$ if and only if $b \in X_j$, since the $b$-th bit of $k_j$ in the binary representation is $1$ if and only if $b \in X_j$.

Let $\epsilon = 1/(24 \cdot n)$. Note that $\epsilon < \min_{j \in [m] \setminus \{0\}} (r_j - k_j/(2 \cdot n))$ since $r_j \geq 1$ and $k_j / (2 \cdot n) \leq 1/2$.
Hence, since $\mathcal{C}$ has representation error $1/(24 \cdot n \cdot m)$ with respect to $\Cbounded$, we can apply \Cref{lem:pdim_lb_contract_construction} to get that there must be a contract $t \in \mathcal{C}$ such that $t_0 \in [0, 1/(24 \cdot n)]$, and for all $1 \leq j \leq m-1$, we have $t_j \in [k_j/(2 \cdot n) + 1/(24 \cdot n), k_j/(2 \cdot n) + 1/(12 \cdot n)]$. In particular, we have $t_j - t_0 \in [k_j/(2 \cdot n), k_j/(2 \cdot n) + 1/(12 \cdot n)]$.
Combining these observations with the guarantees we got from \Cref{lem:pdim_lb_type_construction} and letting $\tau_{(b,j)} = (r_j - \alpha_{n-1}) - (5/2) \cdot \rho_j$, we get:
\begin{itemize}
    \item if $(b, j) \in X$, then $b \in X_j$, hence $k_j \in S_b$, so $u_p(\theta^{(S_b,j)},t) \geq (r_j - \alpha_{n-1}) - 2 \cdot \rho_j > \tau_{(b,j)}$.
    \item if $(b, j) \notin X$, then $b \notin X_j$, hence $k_j \notin S_b$, so $u_p(\theta^{(S_b,j)},t) \leq (r_j - \alpha_{n-1}) - 3 \cdot \rho_j < \tau_{(b,j)}$.
\end{itemize}
This proves the claim. It follows that $\pdim{\Cbounded} = \Omega(m \log n)$ and $\pdim{\Cunbounded} = \Omega(m \log n)$.

Finally, by \Cref{lem:binary_outcome}\ref{enum:binary_general} it suffices to bound the pseudo-dimension of $\Clinear$ in the binary outcome model, and by \Cref{lem:binary_outcome}\ref{enum:lin_opt}, we know that $\Clinear$ in the binary outcome model has representation error of $0$ with respect to $\Cbounded$. Therefore, $\pdim{\Clinear} = \Omega(2 \cdot \log(n)) = \Omega(\log n)$.
\end{proof}

We can use the lower bound of \Cref{theorem:special_case} to establish the following lower bound on the pseudo-dimension of any class with representation error below $O(1/m^2)$\footnote{This result can be extended to $\epsilon < O(1/m^{1+\rho})$ for any $\rho > 0$, but the constant in the big-O notation then depends on $\rho$.}.

\begin{theorem}[Pseudo-Dimension Lower Bound for Subclasses of $\Clinear$ and $\Cbounded$]\label{thm:pdim_lb_for_rep_error_eps}
    Set the rewards as $r_0 = 0$ and $r_1, r_2, \ldots, r_{m-1} \geq 1$.
    Then, for any $0 < \epsilon < 1/(24 \cdot m^2)$, 
    the pseudo-dimension of any $\epsilon$-approximation of $\Cbounded$ is at least $\Omega(\min(m \log(1/\epsilon), m \log n)$.
    In particular, for any $0 < \epsilon < 1$, 
    the pseudo-dimension of any $\epsilon$-approximation of $\Clinear$ is at least $\Omega(\min(\log(1/\epsilon), \log n))$.
\end{theorem}
\begin{proof}
If $\epsilon < 1/(24 \cdot n \cdot m)$, then by \Cref{theorem:special_case}, we get $\pdim{\mathcal{C}} =  \Omega(m \log n)$.

If $\epsilon > 1/(24 \cdot n \cdot m)$, 
let $\bar{n}$ be the smallest positive integer such that $\epsilon \leq 1/(24 \cdot \bar{n} \cdot m)$. 
Such $\bar{n}$ must exist since $\epsilon \leq 1/(24 \cdot 1 \cdot m^2) \leq 1/(24 \cdot 1 \cdot m)$.
Note that $\bar{n} = \Omega((1/\epsilon) \cdot (1/m))$ and $\bar{n} \leq n$.

We can remove $n - \bar{n}$ actions by setting their cost to $\infty$. Then, apply \Cref{theorem:special_case} for the case with $\bar{n}$ actions and $m$ outcomes, which gives a lower bound of $\pdim{\mathcal{C}} = \Omega(m \log \bar{n})$.

Recall that $1/\epsilon \geq 24 \cdot m^2$ which implies $\log(1/\epsilon) \geq 2 \cdot \log m$.
Thus, since $\bar{n} = \Omega((1/\epsilon) \cdot (1/m))$:
\[
\pdim{\mathcal{C}} = \Omega(m \log((1/\epsilon) \cdot (1/m))) = \Omega(m \log (1/\epsilon) - m \log m) = \Omega(m \log(1/\epsilon)),
\]
which completes the proof for bounded contracts.

Finally, similarly to the proof of \Cref{theorem:special_case}, we note that by \Cref{lem:binary_outcome}\ref{enum:binary_general} it suffices to bound the pseudo-dimension of subclasses of $\Clinear$ in the binary outcome model, and by \Cref{lem:binary_outcome}\ref{enum:lin_opt}, we know that in the binary outcome model, any class of contracts $\mathcal{C}$ with representation error of at most $\epsilon$ with respect to $\Clinear$ also has representation error of at most $\epsilon$ with respect to $\Cbounded$. Thus, $\pdim{\mathcal{C}} = \Omega(\min(2 \cdot \log(1/\epsilon), 2 \cdot \log n ) = \Omega(\min(\log(1/\epsilon), \log n))$.
\end{proof}

\section{Sample Complexity Results}\label{sec:sample_complexity_results}
{
First, in \Cref{sec:sample_complexity_upper_bounds}, we use the pseudo-dimension upper bounds established in \Cref{sec:pdim_upper_bounds} to establish upper bounds on the sample complexity of learning linear and bounded contracts, and demonstrate how those results can be translated to sample- and time-efficient algorithms. 
Then, in \Cref{sec:sample_complexity_lower_bounds}, we provide almost matching lower bounds for those problems by a careful reduction to the problem of distinguishing between sufficiently close distributions. 
Thus, we demonstrate that the upper bounds on sample complexity that we obtained in \Cref{sec:sample_complexity_upper_bounds} using the pseudo-dimension are essentially tight. 
Finally, for the class of unbounded contracts, we establish a negative result, proving that no learning algorithm can achieve finite sample complexity in this setting.
}

\subsection{Upper Bounds on the Sample Complexity} \label{sec:sample_complexity_upper_bounds}
{In this section, we use the pseudo-dimension upper bounds established in \Cref{sec:pdim_upper_bounds} to derive sample- and time-efficient algorithms for learning 
{near-optimal linear and bounded contracts.}
}

{To derive sample- and time-efficient algorithms, we introduce the concept of an approximation oracle, which provides a near-optimal contract for an empirical distribution  of agent types.
This approach decouples the learning problem from the computational challenge of identifying the optimal contract for a given (known) distribution. 
For many contract classes and distributions, such oracles can be constructed 
\cite[e.g.,][]{CastiglioniM021,GSW21,alon2021contracts},
and we will demonstrate how to leverage them effectively. 
}

\begin{definition}[Approximation Oracle]\label{def:approximation_oracle}
For any class of contracts $\mathcal{C}$ and a constant $\rho \geq 0$, we define an oracle $\mathcal{O}(\mathcal{C}, \rho)$ as follows:
The input to $\mathcal{O}(\mathcal{C}, \rho)$ is a list of agent types 
$\theta_0, \ldots, \theta_{N-1}$, where $\theta_i = (f_i, c_i)$, possibly with repetitions.
The output of $\mathcal{O}(\mathcal{C}, \rho)$ is a contract $t \in \mathcal{C}$ such that $\mathbb{E}_{\theta \sim \widehat{\mathcal{D}}}[u_p(\theta, t)] \geq \mathrm{OPT}(\mathcal{C}, \widehat{\mathcal{D}}) - \rho$, where $\widehat{\mathcal{D}}$ is the uniform distribution over the agent types in the list, assigning each type $\theta$ a probability mass equal to its frequency in the list.
\end{definition}

{
Using the approximation oracle of \Cref{def:approximation_oracle}, 
we can  establish a sample complexity bound for any subclass of bounded contracts. 
}

\begin{lemma}[Learning with an Approximation Oracle]\label{lem:learning_with_oracle}
 Set the rewards as $r_0 = 0$ and $r_1, r_2, \ldots, r_{m-1} \leq 1$.
 Fix two contract classes  $\bar{\mathcal{C}} \subseteq \mathcal{C} \subseteq \Cbounded$. Assume that $\bar{\mathcal{C}}$ has representation error of at most $\epsilon/3$ with respect to $\mathcal{C}$.
For any $\epsilon > 0$, $\delta > 0$, and $\rho \geq 0$, there exists an algorithm $\mathcal{A}$ that, 
given black-bock access to a single query to the oracle $\mathcal{O}(\bar{\mathcal{C}}, \rho)$,  satisfies the following properties: $\mathcal{A}$ takes as input samples from any distribution $\mathcal{D}$ over the agent types, with the number of samples bounded by:
\begin{align*}
    N = O\left( \frac{1}{\epsilon^2} 
    \left(\pdim{\bar{\mathcal{C}}} 
      + \log \left(\frac{1}{\delta} \right) \right)\right).
\end{align*} 
With probability at least $1-\delta$,  $\mathcal{A}$ returns a contract ${t} \in \bar{\mathcal{C}}$ with the following guarantee:
\begin{align*}
    \mathbb{E}_{\theta \sim \mathcal{D}}[u_p(\theta, {t})] \geq \mathrm{OPT}(\mathcal{C}, \mathcal{D}) - \epsilon - \rho.
\end{align*}
\end{lemma}

\begin{proof}
    Note that for any bounded contract $t$, we have $|u_p(\theta, t)| \leq 1$ for all agent types $\theta = (f,c)$ since 
     $u_p(\theta, t) = \sum_{j \in \setoutcomes} f_{i^{\star}(\theta, t) ,j} (r_j - t_j)$ and
     $r_j - t_j  \in [-1, 1]$ because both $r_j$ and $t_j$ are bounded in $[0, 1]$ for all outcomes $j \in \setoutcomes$ and 
     $f_{i^{\star}(\theta, t)}$ is a probability distribution.

    Consequently, by \Cref{theorem:uniform_learnability}, the class of contracts $\bar{\mathcal{C}}$ is $(\epsilon/3, \delta, N)$-uniformly learnable for some 
    $N = O\left( ({1}/{\epsilon^2}) \left(\pdim{\bar{\mathcal{C}}} + \log \left(1 / {\delta}\right) \right)\right)$.

Consider the algorithm that takes as input a set of $N$ samples $S = \{\theta_0, \ldots, \theta_{N-1}\}$, where $\theta_i = (f_i, c_i)$, and returns the contract $\widehat{t}$ that 
the approximation oracle $\mathcal{O}( \bar{\mathcal{C}}, \rho ) $ would return. 
Additionally, let $\widehat{t^\star}$ and ${t^\star}$ be the optimal contracts in $\bar{\mathcal{C}}$ for the empirical and true distribution respectively. 
Observe that:
\begin{align*}
\mathbb{E}_{\theta \sim \mathcal{D}} [u_p(\theta, \widehat{t})] 
&\geq\mathbb{E}_{\theta \sim \mathcal{\widehat{D}}} [u_p(\theta, \widehat{t})] -  \epsilon/3 && (\text{by uniform learnability}) \\
&\geq\mathbb{E}_{\theta \sim \mathcal{\widehat{D}}} [u_p(\theta, \widehat{t^{\star}})] -  \epsilon/3 - \rho && (\text{by the approximation oracle}) \\
&\geq\mathbb{E}_{\theta \sim \mathcal{\widehat{D}}} [u_p(\theta, {t^{\star}})] -  \epsilon/3 - \rho && (\text{because $\widehat{t^\star}$ optimal for $\widehat{\mathcal{D}}$}) \\
&\geq\mathbb{E}_{\theta \sim \mathcal{{D}}} [u_p(\theta, {t^{\star}})] -  2 \epsilon/3 - \rho  && (\text{by uniform learnability})  \\
& = \mathrm{OPT}(\mathcal{B}_{\epsilon/3}, \mathcal{D}) -  2 \epsilon/3 - \rho && (\text{because ${t^\star}$ optimal for ${\mathcal{D}}$})\\
& \le  \mathrm{OPT}(\Cbounded, \mathcal{D}) -  2 \epsilon/3 - \rho -   \epsilon/3 && (\text{by \Cref{thm:discretization}}) \\ 
& = \mathrm{OPT}(\mathcal{C}, \mathcal{D}) -  \epsilon - \rho  
\end{align*}
Thus, the contract $\widehat{t}$ returned by the oracle satisfies the desired guarantee.
Note that our algorithm just needs to perform a single query to the oracle $\mathcal{O}( \bar{\mathcal{C}}, \rho ) $.
\end{proof}

\begin{remark}
Note that an analogous argument
cannot be made for the class of $\Cunbounded$ because the principal's utility can be arbitrarily negative for contracts with very high payments. 
\end{remark}

\paragraph{Linear Contracts.}
Next, we show how \Cref{lem:learning_with_oracle} can be combined with an existing approximation oracle for the discretization of linear contracts to derive a time- and sample- efficient algorithm.

\begin{theorem}[Efficient Learning for $\mathcal{C}_{\mathrm{linear}}$]\label{thm:efficient_linear}
 Set the rewards as $r_0 = 0$ and $r_1, r_2, \ldots, r_{m-1} \leq 1$.
For any $\epsilon > 0$ and $\delta > 0$, there exists an algorithm $\mathcal{A}$ satisfying the following properties:
$\mathcal{A}$ takes as input samples from any distribution $\mathcal{D}$ over the agent types, with the number of samples bounded by:
\begin{align*}
    N = O\left(\frac{1}{\epsilon^2}\left( \log \left(\frac{1}{\epsilon}\right) + \log \left(\frac{1}{\delta}\right) \right)\right).
\end{align*} 
With probability at least $1-\delta$, $\mathcal{A}$ returns a contract ${t} \in \mathcal{C}_{\mathrm{linear}}$ with the following guarantee:
\begin{align*}
    \mathbb{E}_{\theta \sim \mathcal{D}}[u_p(\theta, {t})] \geq \mathrm{OPT}(\mathcal{C}_{\mathrm{linear}}, \mathcal{D}) - \epsilon.
\end{align*}
Additionally, the running time of $\mathcal{A}$ is polynomial in the size of the action space $n$, the number of outcomes $m$, $1/\epsilon$, and $\log(1/\delta)$. 
\end{theorem}

\begin{proof}
     By \Cref{theorem:t_contracts_true_pd}, the pseudo-dimension of $\mathcal{L}_{\epsilon/3}$ is  $\pdim{\mathcal{L}_{\epsilon/3}} = O(1/\epsilon)$, and by \Cref{lemma:t_contracts_representation_error}, the representation error of $\mathcal{L}_{\epsilon/3}$ is $\epsilon/3$.
   The sample complexity result follows by applying \Cref{lem:learning_with_oracle} to $\bar{\mathcal{C}} = \mathcal{L}_{\epsilon/3}$ and $\mathcal{C} = \Clinear$ with oracle $\mathcal{O}(\mathcal{L}_{\epsilon/3}, 0)$.

  It remains to show that the oracle $\mathcal{O}(\mathcal{L}_{\epsilon/3}, 0)$ can be implemented in polynomial time. This follows since $\mathcal{L}_{\epsilon/3} = \{0, \epsilon/3, 2\epsilon/3, \ldots, 1\}$ contains 
  $\lceil 3/\epsilon \rceil + 1$ 
  contracts and for each of those contracts we can compute the agent's best response and the principal's utility in polynomial time with respect to $n$ and $m$ {by simply calculating the agent's utility for all possible actions}. 
  The result follows since $N$ is polynomial in $1/\epsilon$ and $\log(1/\delta)$.
\end{proof}

\paragraph{Bounded Contracts.}

Next, we show how \Cref{lem:learning_with_oracle} can be combined with an approximation oracle for $\Cbounded$ to derive a 
sample-efficient algorithm for $\Cbounded$.

\begin{theorem}[Efficient learning for $\Cbounded$]\label{thm:bounded_efficient}
 Set the rewards as $r_0 = 0$ and $r_1, r_2, \ldots, r_{m-1} \leq 1$.
For any $\epsilon > 0$, $\delta > 0$, and $\rho \geq 0$, there exists an algorithm $\mathcal{A}$ that, 
given black bock access to a single query to the oracle $\mathcal{O}(\Cbounded, \rho)$, satisfies the following properties: $\mathcal{A}$ takes as input samples from any distribution $\mathcal{D}$ over the agent types, with the number of samples bounded by:
\begin{align*}
    N = O\left(\frac{1}{\epsilon^2} \left( m\log n + m\log m
      + \log \left(\frac{1}{\delta}\right) \right)\right).
\end{align*} 
With probability at least $1-\delta$,  $\mathcal{A}$ returns a contract ${t} \in \Cbounded$ with the following guarantee:
\begin{align*}
    \mathbb{E}_{\theta \sim \mathcal{D}}[u_p(\theta, {t})] \geq \mathrm{OPT}(\Cbounded, \mathcal{D}) - \epsilon - \rho.
\end{align*}
\end{theorem}

\begin{proof}
By \Cref{theorem:general_contract_pdim}, the pseudo-dimension of $\Cbounded$ is  $\pdim{\Cbounded} = O (m\log n + m \log m)$. We also clearly have that $\Cbounded$ has representation error of $0$ with respect to $\Cbounded$. Therefore, the result follows by applying \Cref{lem:learning_with_oracle} with $\bar{\mathcal{C}} = \Cbounded$ and $\mathcal{C} = \Cbounded$ with oracle $\mathcal{O}(\Cbounded, \rho)$.
\end{proof}

\paragraph{Discretized Bounded Contracts.} 

Next, we show how we can leverage the discretization of bounded contracts of \Cref{thm:discretization} to obtain a sample-complexity bound that is independent of $n$.

\begin{theorem}[Sample complexity for $\Cbounded$ using discretization]\label{thm:cone_efficient}
 Set the rewards as $r_0 = 0$ and $r_1, r_2, \ldots, r_{m-1} \leq 1$.
For any $\epsilon > 0$, $\delta > 0$, and $\rho \geq 0$, there exists an algorithm $\mathcal{A}$ that, 
given black-bock access to a single query to the oracle $\mathcal{O}(\Bepsilon, \rho)$,  satisfies the following properties: $\mathcal{A}$ takes as input samples from any distribution $\mathcal{D}$ over the agent types, with the number of samples bounded by:
\begin{align*}
    N = O\left( \frac{1}{\epsilon^2}
    \left(m \log \left( \frac{1}{\epsilon}\right) + m \log m 
      + \log \left(\frac{1}{\delta} \right) \right)\right).
\end{align*} 
With probability at least $1-\delta$,  $\mathcal{A}$ returns a contract ${t} \in \Cbounded$ with the following guarantee:
\begin{align*}
    \mathbb{E}_{\theta \sim \mathcal{D}}[u_p(\theta, {t})] \geq \mathrm{OPT}(\Cbounded, \mathcal{D}) - \epsilon - \rho.
\end{align*}
\end{theorem}

\begin{proof}
    By \Cref{lemma:se_contracts_pseudo-dimension}, the pseudo-dimension of $\mathcal{B}_{\epsilon/3}$ is  $\pdim{\mathcal{B}_{\epsilon/3}} = O(m \log (1 / \epsilon) + m \log m)$ and $\mathcal{B}_{\epsilon/3}$ has representation error of $\epsilon/3$ with respect to $\Cbounded$. The result then follows by applying \Cref{lem:learning_with_oracle} with $\bar{\mathcal{C}} = \mathcal{B}_{\epsilon/3}$ and $\mathcal{C} = \Cbounded$ with oracle $\mathcal{O}(\mathcal{B}_{\epsilon/3}, \rho)$.
\end{proof}

\begin{remark}
A natural way to implement the approximation oracle $\mathcal{O}(\Bepsilon, 0)$ is to construct the whole sequence of $|\Bepsilon| = (m/\epsilon)^{O(m)}$ contracts, compute the principal's empirical utility for each of them, and return the optimal contract in the sequence.
\end{remark}

\subsection{Lower Bounds on the Sample Complexity}\label{sec:sample_complexity_lower_bounds}
{
In this section, we establish nearly matching lower bounds for the sample complexity of finding near-optimal contracts, demonstrating that the upper bounds derived using the pseudo-dimension in \Cref{sec:sample_complexity_upper_bounds} are essentially tight. Finally, we prove that for the class of unbounded contracts, no learning algorithm with finite sample complexity exists.
}

\paragraph{Linear Contracts.} 
{
We will reduce the problem of learning an almost optimal linear contract to that of distinguishing between two sufficiently close Bernoulli distributions. 
We will use the following well-known result, see \cite[e.g.][Chapter 4]{DBLP:books/daglib/0035708}.
We also provide a proof in \Cref{sec:proofs_for_sample_complexity}.
}

\begin{restatable}[Sample Complexity Lower Bound for a Distribution over Two Elements]
{lemma}{samplelowerboundtwoelemenets}\label{lem:bernoulli_lb}
 Fix $0 < \epsilon < 1/4$, $0 < \delta < 1/2$, and let $\mathcal{X} = \{0, 1\}$. Consider two Bernoulli distributions $\mathcal{D}_1 = \mathrm{Ber}\left( 1/2 - \epsilon \right)$ and $\mathcal{D}_2 = \mathrm{Ber}\left( 1/2 + \epsilon \right)$ defined over $\mathcal{X}$. Suppose that there is an algorithm that, given $N$ samples from $\mathcal{D}_i$ for any unknown $i \in \{1, 2\}$, correctly decides whether $i = 1$ or $i=2$ with probability at least $1-\delta$. Then, the number of samples that the algorithm uses must be at least $N \geq \Omega((1/\epsilon^2) \log (1/\delta))$.
\end{restatable}

{We establish the following sample complexity lower bound by reducing the problem of learning an approximately optimal linear contract to the fundamental hypothesis testing problem above.}

\begin{theorem}[Sample Complexity Lower Bound for $\Clinear$]\label{thm:sample_lower_bound_linear}
Suppose that $n \geq 2$, $m \geq 2$, $r_0 = 0$, and $r_1 \geq 1$. Consider  any learning algorithm that given $N$ samples from any agent's type distribution $\mathcal{D}$, with probability at least $1-\delta$, returns a linear contract $t$ with $\mathbb{E}_{\theta \sim \mathcal{D}}[u_p(\theta, t)] \geq \mathrm{OPT}(\mathcal{D}, \Clinear) - \epsilon$. Then, the sample complexity is 
 \[
N = \Omega\left( \dfrac{1}{\epsilon^2} \cdot \log\left( \dfrac{1}{\delta} \right) \right).
\]   
\end{theorem}
\begin{proof}
    We first define two agent's types, $\theta_1$ and $\theta_2$. Let $\theta_1$ be a type defined as follows:
    \begin{table}[H]
        \centering
        \begin{tabular}{|c|cc|c|}
        \hline
        & $r_0 = 0$  & $r_1 = 1$ & \textbf{cost} \\ \hline
        \textbf{action 0:} & $1/2$ & $1/2$ & $c_0 = 0$ \\ 
        \textbf{action 1:} & $1/2$ & $1/2$ & $c_1 = 1/4$\\ \hline
        \end{tabular}
    \end{table}
    Let $\theta_2$ be a type defined as follows:
\begin{table}[H]
        \centering
        \begin{tabular}{|c|cc|c|}
        \hline
        & $r_0 = 0$  & $r_1 = 1$ & \textbf{cost} \\ \hline
        \textbf{action 0:} & $1$ & $0$ & $c_0 = 0$ \\ 
        \textbf{action 1:} & $1/2$ & $1/2$ & $c_1 = 1/4$\\ \hline
        \end{tabular}
    \end{table}
    Note that for $\theta_1$, the agent's best response is always action $0$, and so $u_p(\theta_1, \alpha) = (1/2) \cdot (1-\alpha)$. For $\theta_2$, the agent's best response is action $0$ for $\alpha < 1/2$, and action $1$ for $\alpha\geq 1/2$. Hence, $u_p(\theta_2, \alpha) = 0$ for $\alpha < 1/2$, and $u_p(\theta_2, \alpha) = (1/2) \cdot (1-\alpha)$ for $\alpha \geq 1/2$.
    
    Consider two distributions over the agent's type space, $\mathcal{D}_1$ and $\mathcal{D}_2$. The distribution $\mathcal{D}_1$ puts probability mass $1/2-4 \epsilon$ on type $\theta_1$ and probability mass $1/2+4 \epsilon$ on type $\theta_2$. The distribution $\mathcal{D}_2$ puts probability mass $1/2+4 \epsilon$ on type $\theta_1$ and probability mass $1/2-4  \epsilon$ on type $\theta_2$.

    We will show that for all $0 \leq \alpha < 1/2$, we have $u_p(\mathcal{D}_1, \alpha) < \mathrm{OPT}(\mathcal{D}_1, \Clinear) - \epsilon$. Thus, for the distribution $\mathcal{D}_1$, by the assumption, the algorithm $\mathcal{A}$ must return a linear contract in $[1/2,1]$ with probability at least $1-\delta$. 
    We will also show that for all $1/2 \leq \alpha \leq 1$, we have $u_p(\mathcal{D}_2, \alpha) < \mathrm{OPT}(\mathcal{D}_2, \Clinear) - \epsilon$.  Thus, for the distribution $\mathcal{D}_2$, the algorithm $\mathcal{A}$ must return a linear contract in $[0,1/2)$ with probability at least $1-\delta$. By \Cref{lem:bernoulli_lb} we know that distinguishing between $\mathcal{D}_1$ and $\mathcal{D}_2$ with probability at least $1-\delta$ requires at least $\Omega((1/(\epsilon)^2) \log (1/\delta))$ samples, as needed.

    The expected principal's utility for $\mathcal{D}_1$ is $u_p(\mathcal{D}_1, \alpha) = (1/2-4 \epsilon) \cdot u_p(\theta_1, \alpha) + (1/2+4  \epsilon) \cdot u_p(\theta_2, \alpha)$. For $\alpha < 1/2$, we have $u_p(\mathcal{D}_1, \alpha) = (1/2-4  \epsilon) \cdot (1/2) \cdot (1-\alpha)$. For $\alpha \geq 1/2$, we have $u_p(\mathcal{D}_1, \alpha)=(1/2-4  \epsilon) \cdot (1/2) \cdot (1-\alpha) + (1/2+4  \epsilon) \cdot (1/2) \cdot (1-\alpha) = (1/2) \cdot (1-\alpha)$. In particular, $\mathrm{OPT}(\mathcal{D}_1, \Clinear) = u_p(\mathcal{D}_1, 1/2) = (1/2) \cdot (1-1/2) = 1/4$. Moreover, for every $\alpha < 1/2$, we have $u_p(\mathcal{D}_1, \alpha) \leq u_p(\mathcal{D}_1, 0) = (1/2-4\epsilon) \cdot (1/2) \cdot (1-0) = 1/4 - 2\epsilon < \mathrm{OPT}(\mathcal{D}_1, \Clinear) - \epsilon$.

    The expected principal's utility for $\mathcal{D}_2$ is $u_p(\mathcal{D}_2, \alpha) = (1/2+4 \epsilon) \cdot u_p(\theta_1, \alpha) + (1/2-4  \epsilon) \cdot u_p(\theta_2, \alpha)$. For $\alpha < 1/2$, we have $u_p(\mathcal{D}_2, \alpha) = (1/2+4  \epsilon) \cdot (1/2) \cdot (1-\alpha)$. For $\alpha \geq 1/2$, we have $u_p(\mathcal{D}_2, \alpha)=(1/2+4  \epsilon) \cdot (1/2) \cdot (1-\alpha) + (1/2-4  \epsilon) \cdot (1/2) \cdot (1-\alpha) = (1/2) \cdot (1-\alpha)$. In particular, $\mathrm{OPT}(\mathcal{D}_2, \Clinear) = u_p(\mathcal{D}_2, 0) = (1/2+4\epsilon) \cdot (1/2) \cdot (1-0) = 1/4 + 2\epsilon$. Moreover, for every $\alpha \geq 1/2$, we have $u_p(\mathcal{D}_1, \alpha) \leq u_p(\mathcal{D}_1, 1/2) = (1/2) \cdot (1-1/2) = 1/4 < \mathrm{OPT}(\mathcal{D}_1, \Clinear) - \epsilon$. This proves the claim, and so the result follows.
\end{proof}

\paragraph{Bounded Contracts.}
{Building on our technique for linear contracts, we reduce the problem of learning an approximately optimal bounded contract to that of distinguishing between distributions over 
$2n$ elements. 
We bound the sample complexity of distinguishing between such distributions using a a standard idea from learning theory 
\cite[e.g.,][Theorem 10]{diakonikolas2019lecture4}. 
We provide a proof of the lemma in \Cref{sec:proofs_for_sample_complexity}.}

\begin{restatable}[Sample Complexity Lower Bound for a Distribution over $2n$ Elements]
{lemma}{samplelowerboundmanyelemenets}\label{lem:discrete_lb}
Fix $0 < \epsilon < 1/4$ and let $\mathcal{X} = \{1_0, 2_0, 3_0, \ldots, n_0\} \cup \{1_1, 2_1, 3_1, \ldots, n_1\}$.
    For every vector $z \in \{-1, +1\}^n$, define a distribution $\mathcal{D}^{(z)}$ over $\mathcal{X}$ by putting probability mass $(1- \epsilon \cdot z_i)/(2n)$ on $i_0$ and probability mass $(1 + \epsilon \cdot z_i)/(2n)$ on $i_1$. Suppose that there is an algorithm that, given $N$ samples from $\mathcal{D}^{(z)}$ for any unknown $z \in \{-1,+1\}^n$, with probability at least $3/4$, outputs a vector $z' \in \{-1,+1\}^n$ such that $\sum_{i=1}^n \mathbbm{1}[z_i \neq z_i'] \leq (1/4) \cdot n$. Then, the number of samples that the algorithm uses is at least $N = \Omega\left({n}/{\epsilon^2}\right)$.  
\end{restatable}

{
We prove the following sample complexity lower bound by reducing the problem of learning an approximately optimal bounded contract to this multi-element hypothesis testing problem. The details of the reduction are deferred to \Cref{sec:proofs_for_sample_complexity}.}

\begin{restatable}[Sample Complexity Lower Bound for $\Cbounded$]
{theorem}{samplecomplexitylbboundedcontracts}\label{thm:sample_complexity_lb_bounded_contracts}
Suppose that $n \geq 2$ and $m \geq 2$.
 Set the rewards as $r_0 = 0$ and $r_1, r_2, \ldots, r_{m-1} \geq 1$.
 Consider  any learning algorithm $\mathcal{A}$ that given $N$ samples from any agent's type distribution $\mathcal{D}$, with probability at least $1- \delta$ for $0 < \delta \leq 1/4$, returns a bounded contract $t$ with $\mathbb{E}_{\theta \sim \mathcal{D}}[u_p(\theta, t)] \geq \mathrm{OPT}(\mathcal{D}, \Cbounded) - \epsilon$. Then, the sample complexity is:
 \[
N = \Omega\left( \dfrac{1}{\epsilon^2}  \left( m + \log\left( \dfrac{1}{\delta} \right) \right)\right).
\] 
\end{restatable}

\paragraph{Unbounded Contracts.}
{
Unlike linear and bounded contracts, the class of unbounded contracts presents a fundamental difficulty.  \Cref{thm:sample_impossibility_unbounded_contracts} establishes the impossibility of a learning algorithm with sufficiently small error. We defer the proof to \Cref{sec:proofs_for_sample_complexity}.
}

\begin{restatable}[Impossibility Result for Unbounded Contracts]
{theorem}{sampleimpossibilityunboundedcontracts}\label{thm:sample_impossibility_unbounded_contracts}
 If $m \geq 3$ and $n \geq 2$, then for any $\epsilon < 1/4$ and any $\delta < 1$, there is no learning algorithm with finite sample complexity that for every distribution $\mathcal{D}$ over the agent's types, with probability at least $1-\delta$, outputs a contract $t$ with  
$\mathbb{E}_{\theta \sim \mathcal{D}}[u_p(\theta, t)] \geq \mathrm{OPT}(\mathcal{D}, \Cunbounded) - \epsilon$. 
\end{restatable}

\section{Extensions to Other Settings}\label{sec:extensions}

\subsection{Combinatorial Contracts}\label{sec:combinatorial_actions}
In this section, we extend our results to the combinatorial actions setting introduced in \citet*{Duetting2021CombinatorialC}. 
{ We demonstrate how pseudo-dimension enables us to leverage structural insights specific to these settings, achieving improvements over the {naive} extension {of our results} to this domain.
}

\paragraph{Combinatorial Actions Model.}

In this section, we define the combinatorial actions model.
In this setting, there is a ground set of actions $[\bar{n}] = \{0, 1, \ldots, \bar{n}-1\}$, and the agent can take any subset $S \subseteq [\bar{n}]$ of those actions. 
The agent's type $(\bar{f}, \bar{c})$ in this model comprises of a production function $\bar{f}_j : 2^{[\bar{n}]} \to \Delta^m$ and a cost function $\bar{c} : 2^{[\bar{n}]} \to [0,1]$.
If the agent with type $({\bar{f},\bar{c}})$ chooses to take action set $S$, outcome $j$ occurs with probability $\bar{f}_j(S)$ and the agent incurs cost equal to 
$\bar{c}(S)$.

The agent with type $\bar{\theta} = (\bar{f}, \bar{c})$ in the combinatorial actions model, with a set of $\numcombinatorialactions$ combinatorial actions, is equivalent to the agent with type ${\theta = (f, c)}$ in the non-combinatorial setting with $\numcombinatorialactionsexpanded = 2^{\numcombinatorialactions}$ actions, where each of the $2^{\numcombinatorialactions}$ possible sets of actions in the combinatorial setting corresponds to a single action $S$ in the non-combinatorial setting.
The transition probabilities are related as $f_{S,j} = \bar{f}_j(S)$ for every outcome $j \in \setoutcomes$ and corresponding costs are $c_S = \bar{c}(S)$.

Given the equivalence of the two models via the correspondence above, we extend all definitions from the model described in \Cref{sec:contract_prelims} to the combinatorial actions model.
In particular, we can restate \Cref{theorem:linear_contract_pdim} for the setting with combinatorial actions.

\begin{proposition}\label{corollary:pdim_combinatorial_actions}
The pseudo-dimension of the class of linear contracts in the combinatorial actions model with $\numcombinatorialactions$ combinatorial actions 
is $\pdim{\Clinear}  = O (\numcombinatorialactions)$.
\end{proposition}

\paragraph{Tighter Bounds with Structure.}

In the combinatorial actions setting it is natural to impose some form of ``decreasing marginal returns'' on the structure of the agent's production function $\bar{f}$, which in turn imply ``decreasing marginal returns'' on the principal's reward function.

\begin{definition}[Properties of Set Functions] \label{def:set_functions}
    Let $f : 2^{[\numcombinatorialactions]} \to \mathbb{R}$ be a set function. We say that:
    \begin{itemize}
        \item $f$ is additive if $f(S \cup T) = f(S) + f(T)$ for any disjoint $S, T \subseteq [\numcombinatorialactions]$,
        \item $f$ is submodular if $f(S \cup \{i\}) - f(S) \geq f(T \cup \{i\}) - f(T)$ for any $S \subseteq T \subseteq [\numcombinatorialactions]$ and $i \notin T$,
        \item $f$ is supermodular if $-f$ is submodular,
        \item $f$ is gross substitutes if for any two vectors $p \ge  q \in \mathbb{R}^{\numcombinatorialactions}_{\ge 0}$ and any $S \in \arg \max_{S' \subseteq [\numcombinatorialactions]}
f(S') - \sum_{i \in S'} p_i 
$ there is some $T \in \arg \max_{T' \subseteq [\numcombinatorialactions]} 
f(T') - \sum_{i \in T'} q_i 
$ and $\{ i \in S \,|\, q_i = p_i \} \subseteq T$.
    \end{itemize}
\end{definition}

It is well known that any additive function is gross subsitutes, submodular, and supermodular. Moreover, any gross substitutes function is submodular.

\Cref{corollary:linear_contracts_critical_values} gives a bound of $2^{\bar{n}}$ on the number of critical values for any agent type in the combinatorial actions model. \citet*{Duetting2021CombinatorialC,dutting2024combinatorial} provide improved bounds under the gross substitutes and submodular conditions, respectively.

\begin{theorem}[\cite{Duetting2021CombinatorialC}, \cite{dutting2024combinatorial}] \label{thm:combinatorial_critical_values}
    For any agent type $\bar{\theta} = (\bar{f}, \bar{c})$ in the combinatorial actions model,
    we consider the principal's reward as a function of the subset of actions takes by the agent: $r_p(\bar{\theta}, S) = r \cdot \bar{f}(S)$.
    We have the following:
    \begin{enumerate}[label=(\roman*)]
        \item If the principal's reward function $r_p(\bar{\theta}, \cdot)$ is gross substitutes and the cost function $\bar{c}$ is additive, then the number of critical values is $|C_{\bar{f},\bar{c}}| \leq n(n-1)/2$.
        \item If  the principal's reward function $r_p(\bar{\theta}, \cdot)$ is submodular and the cost function $\bar{c}$ is supermodular, then the number of critical values is  $|C_{\bar{f},\bar{c}}| \leq n$.
    \end{enumerate}
    Moreover, in both cases, the agent's best response can be computed in time polynomial in $\bar{n}$ and $m$. 
\end{theorem}

As the upper bound in \Cref{theorem:linear_contract_pdim} only depends on the number of critical values, 
\Cref{thm:combinatorial_critical_values} immediately yields improved pseudo-dimension bounds for linear contracts in two special cases.

\begin{corollary}[Pseudo-Dimension Upper Bounds with Combinatorial Actions]\label{cor:comb_better}
    Let $\Theta$ be a type space in the combinatorial actions model equal to either:
    \begin{enumerate}[label=(\roman*)]
        \item $\Theta = \{\bar{\theta} = (\bar{f},\bar{c}) : r_p(\bar{f}, \cdot) \text{ is gross substitutes and } \bar{c} \text{ is additive}\}$, or
        \item $\Theta = \{\bar{\theta} = (\bar{f},\bar{c}) : r_p(\bar{f}, \cdot) \text{ is submodular and } \bar{c} \text{ is supermodular}\}$,
    \end{enumerate}
    then, the pseudo-dimension of the class of linear contracts is $\pdim{\Clinear, \Theta} = O(\log \bar{n})$.
\end{corollary}

\paragraph{Efficient Learning Algorithms.} \label{section:computational_aspects_combinatorial_actions}

Next, we establish an efficient learning algorithm for this model.

\begin{theorem}[Efficient Learning with Combinatorial Actions]\label{thm:efficient_combinatorial}
Let $\Theta$ be a subset of the agent's type space in the combinatorial actions model and let $k$ be an upper bound on the number of critical values in $\Theta$, i.e., $|C_{\bar{f}, \bar{c}}| \leq k$ for all $\bar{\theta} = (\bar{f}, \bar{c}) \in \Theta$.
For any $\epsilon > 0$, $\delta > 0$, there exists an algorithm $\mathcal{A}$ satisfying the following properties: $\mathcal{A}$ takes as input samples from any distribution $\mathcal{D}$ over the agent types in the combinatorial actions model, with the number of samples bounded by:
\begin{align*}
    N = O\left(\left(\frac{1}{\epsilon}\right)^2 \left( \log (k)    + \log \left(\frac{1}{\delta}\right) \right)\right).
\end{align*} 
With probability at least $1-\delta$,  $\mathcal{A}$ returns a contract $\alpha \in \mathcal{C}_{\mathrm{linear}}$ with the following guarantee:
\begin{align*}
    \mathbb{E}_{\bar{\theta} \sim \mathcal{D}}[u_p(\bar{\theta}, {\alpha})] \geq \mathrm{OPT}(\mathcal{C}_{\mathrm{linear}}, \mathcal{D}) - \epsilon.
\end{align*}
Additionally, if $\mathcal{A}$ is given black-box access 
to the demand oracle that returns the agent's best response for any contract $\alpha \in \mathcal{C}_{\mathrm{linear}}$ and agent type $\bar{\theta}$, then the running time of $\mathcal{A}$ is polynomial in 
{$k$, $\bar{n}$, $m$, $1/\epsilon$ and $\log(1/\delta)$.} 
\end{theorem}

\begin{remark}
    Similarly to \Cref{cor:comb_better}, for the special cases where the type space $\Theta$ is either:
    \begin{enumerate}[label=(\roman*)]
        \item $\Theta = \{(\bar{f},\bar{c}) : \bar{f} \text{ is gross substitutes and } \bar{c} \text{ is additive}\}$, or
        \item $\Theta = \{(\bar{f},\bar{c}) : \bar{f} \text{ is submodular } \text{and } \bar{c} \text{ is supermodular}\}$,
    \end{enumerate}
    the sample complexity guarantee in \Cref{thm:efficient_combinatorial} becomes:
\begin{align*}
    N = O\left(\left(\frac{1}{\epsilon}\right)^2 \left( \log (\bar{n})    + \log \left(\frac{1}{\delta}\right) \right)\right).
\end{align*}
and the running time of the algorithm $\mathcal{A}$ is polynomial in $\bar{n}$, $m$, $1/\epsilon$, and $1/\delta$.
\end{remark}

The proof of \Cref{thm:efficient_combinatorial} relies on the following result of \citet*{dutting2024combinatorial}.

\begin{theorem}[\cite{dutting2024combinatorial}]\label{thm:enumerate_critical}
Let $\Theta$ be a subset of the agent's type space in the combinatorial actions model and let $k$ be an upper bound on the number of critical values in $\Theta$, i.e., $|C_{\bar{f}, \bar{c}}| \leq k$ for all $\bar{\theta} = (\bar{f}, \bar{c}) \in \Theta$.
There exists an algorithm that for any agent's type $\bar{\theta} = (\bar{f}, \bar{c}) \in \Theta$ enumerates all critical values in $C_{\bar{f}, \bar{c}}$.
If the algorithm is given black-box access 
to the demand oracle that returns the agent's best response for any linear contract $\alpha \in \mathcal{C}_{\mathrm{linear}}$ and any agent's type $\bar{\theta} \in \Theta$, then the running time of the algorithm is polynomial in the number of critical values $k$, the number of combinatorial actions $\bar{n}$, the number of outcomes $m$, $1/\epsilon$, and $\log(1/\delta)$. 
\end{theorem}

{In addition to this result, we use the following lemma, which establishes that any empirically optimal linear contract in the combinatorial actions model must correspond to either zero or a critical value of one of the agent types in the given sample set. }

\begin{lemma} \label{lemma:empirical_optimal_contract_at_critical_point}
Fix any multiset $S = \{\bar{\theta}_0, \ldots, \bar{\theta}_{N-1} \}$ of agent types in the combinatorial actions model, where $\bar{\theta}_i = (\bar{f}_i, \bar{c}_i)$. Let $\alpha \in \mathcal{C}_{\mathrm{linear}}$ be the linear contract maximizing $(1/|S|) \cdot \sum_{\bar{\theta} \in S} u_p(\bar{\theta}, \alpha)$. Then, $\alpha$ is either $0$ or a critical value for one of the agent types in $S$, i.e., either $\alpha = 0$ or $\alpha \in C_{\bar{f}, \bar{c}}$ for some $\bar{\theta} = (\bar{f}, \bar{c}) \in S$.
\end{lemma}

\begin{proof}
Suppose for the purpose of contradiction that there is a linear contract $\alpha^\star \in \mathcal{C}_{\mathrm{linear}}$ such that $(1/|S|) \cdot \sum_{\bar{\theta} \in S} u_p(\bar{\theta}, \alpha^\star) > (1/|S|) \cdot \sum_{\bar{\theta}\in S} u_p(\bar{\theta}, \alpha)$ for $\alpha = 0$ and all $\alpha \in C_{\bar{f},\bar{c}}$ for all $\bar{\theta} = (\bar{f},\bar{c}) \in S$.

Consider some $\bar{\theta}_i = (\bar{f}_i,\bar{c}_i) \in S$. By \Cref{corollary:linear_contracts_critical_values}, since $\alpha^\star > 0$ and $\alpha^\star$ is not a critical value for any of the agent's types in $S$, 
there must be some $\bar{\epsilon}_i > 0$ such that
     $r_p (\bar{\theta}_i, \alpha^\star - \epsilon) = r_p (\bar{\theta}_i, \alpha^\star) \text{ for all $0 < \epsilon \leq \bar{\epsilon}_i$}$.
Let  $\epsilon = \min_{i \in [N]} \bar{\epsilon}_i$. We observe that:
\begin{align*}
    \frac{1}{|S|} \sum_{\bar{\theta} \in S} u_p(\bar{\theta}, \alpha^\star - \epsilon) & =  \frac{1}{|S|} \sum_{\bar{\theta} \in S} (1-\alpha^\star + \epsilon) \cdot r_p(\bar{\theta}, \alpha^\star - \epsilon)
    \nonumber \\ 
    & = \frac{1}{|S|} \sum_{\bar{\theta} \in S} (1-\alpha^\star + \epsilon) \cdot r_p(\bar{\theta}, \alpha^\star)
    && (\text{since $\epsilon \leq \bar{\epsilon}_i$ for all $i \in [N]$}) \\ 
    & > \frac{1}{|S|} \sum_{\bar{\theta} \in S} (1-\alpha^\star) \cdot r_p(\bar{\theta}, \alpha^\star)
    \nonumber && (\text{since $\epsilon > 0$}) \\ 
    & = \frac{1}{|S|} \sum_{\bar{\theta} \in S} u_p(\bar{\theta}, \alpha^\star)
\end{align*}
which contradicts the assumption. This concludes the proof.
\end{proof}

We are now ready to prove \Cref{thm:efficient_combinatorial}.

\begin{proof}[Proof of \Cref{thm:efficient_combinatorial}]

 By \Cref{theorem:linear_contract_pdim}, the pseudo-dimension of $\mathcal{C}_{\mathrm{linear}}$ for the agent's type space $\Theta$ is  $\pdim{\Clinear, \Theta} = O(k)$. 
     {Additionally, note that for all linear contracts, we have $|u_p(\bar{\theta}, \alpha)| \leq 1$ for all agent types $\bar{\theta} = (\bar{f},\bar{c})$, since $u_p(\bar{\theta}, \alpha) = (1-\alpha) \cdot r_p(\bar{\theta}, \alpha) = (1-\alpha) \cdot \sum_{j \in \setoutcomes} f_{j}(S) \cdot r_j$ for some set of action $S \subseteq [\numcombinatorialactions]$ and $r_j \in [0,1]$ for all $j \in \setoutcomes$ and $f_i$ is a probability distribution.}
    Consequently, by \Cref{theorem:uniform_learnability}, the class of contracts $\Clinear$ is $(\epsilon/2, \delta, N)$-uniformly learnable over the type space $\Theta$ for some $N = O\left((1/{\epsilon^2}) \left( \log(k)  + \log \left({1}/{\delta}\right) \right)\right)$.

    Consider the algorithm that takes as input a set of $N$ samples $S = \{\bar{\theta}_0, \ldots, \bar{\theta}_{N-1}\}$, where $\bar{\theta}_i = (\bar{f}_i, \bar{c}_i)$, and returns the empirically optimal linear contract $\alpha \in \Clinear$ that maximizes $(1/|S|) \cdot \sum_{\bar{\theta} \in S} u_p(\bar{\theta}, \alpha)$. By \Cref{lem:empirical_maximization}, we have $\mathbb{E}_{\bar{\theta} \sim \mathcal{D}} [u_p(\bar{\theta}, \alpha)]  \geq \mathrm{OPT}(\Clinear, \mathcal{D}) - 2 \cdot (\epsilon/2) =  \mathrm{OPT}(\mathcal{C}_{\mathrm{linear}}, \mathcal{D}) - \epsilon$.
    Thus, the empirically optimal linear contract $\alpha$ satisfies the desired guarantee.

  It remains to show that the algorithm can be implemented in polynomial time. 
  By \Cref{lemma:empirical_optimal_contract_at_critical_point}, we know that the empirically optimal contract $\alpha$ is a critical value for one of the agent's types in the sample.
  By \Cref{thm:enumerate_critical}, we can enumerate the critical values for every agent's type in the sample.
  Note that there are at most $N \cdot k$ of those critical values in total.
  The result follows since for each of those critical values we compute the agent's best response and the principal's utility in polynomial time with respect to $n$ and $m$ {by simply calculating the agent's utility for all possible actions}. 
  The result follows since $N$ is polynomial in $k$, $1/\epsilon$, and $\log(1/\delta)$.
\end{proof}

\subsection{Menus of Contracts} \label{sec:menus}
In this section, we show how pseudo-dimension can be used to derive sample-efficient algorithms for learning menus of contracts. This highlights the flexibility of our approach and its applicability 
beyond the scenarios studied in the previous sections.

First, we define menus of contracts of a bounded size.

\begin{definition}[Menu of Contracts]
A menu of contracts of size \( K \) is a collection of \( K \) contracts, denoted as \( M = (t^0, \dots, t^{K-1}) \), where each contract \( t^k = (t^k_0, \dots, t^k_{m-1}) \in \mathbb{R}_{\geq 0}^m \) for \( k \in [K] \).

An agent with type \( \theta = (f, c) \) selects a contract \( t^{k^\star} \) and corresponding action \( i^\star \) to maximize her expected utility. This is formalized by solving:
\[
(k^\star, i^\star) \in \argmax_{k \in [K], i \in \setactions} u_a(\theta, t^k, i), \quad \text{where } u_a(\theta, t^k, i) = \sum_{j \in \setoutcomes} f_{i,j} t^k_j - c_i.
\]
As is standard in the literature, the agent is assumed to break ties in favor of the principal.

Given the agent's chosen contract \( k^\star \) and best response action \( i^\star \), the principal's expected utility from the menu \( M \) is:
\[
u_p(\theta, M) = \sum_{j \in \setoutcomes} f_{i^\star, j} (r_j - t^{k^\star}_j).
\]

We denote the set of all menus with \( K \) bounded contracts as \( \mathcal{M}_{\mathrm{bounded}}^K \) and the set of all menus with \( K \) unbounded contracts as \( \mathcal{M}_{\mathrm{unbounded}}^K \).
\end{definition}

We derive a pseudo-dimension bound for the classes of menus of unbounded and bounded contracts. The proof is deferred to \Cref{sec:proofsforextensions}.

\begin{restatable}[Pseudo-Dimension Upper Bound for $\mathcal{M}_{\mathrm{bounded}}^K$ and $\mathcal{M}_{\mathrm{unbounded}}^K$]
{theorem}{pdimupperboundsformenus}\label{theorem:menu_general_contracts_pdim}
The pseudo-dimension of $\mathcal{M}_{\mathrm{unbounded}}^K$, and hence also of $\mathcal{M}_{\mathrm{bounded}}^K$, is at most 
$O (K \cdot \numoutcomes \cdot  \log (K \cdot \numoutcomes \cdot \numactions)) $.
\end{restatable}

We can use the pseudo-dimension bound to learn menus of contracts as follows.

\begin{theorem}[Efficient Learning for Menus of Contracts]\label{thm:menus_efficient}
 Set the rewards as $r_0 = 0$ and $r_1, r_2, \ldots, r_{m-1} \leq 1$.
For any $\epsilon > 0$, $\delta > 0$, and $\rho \geq 0$, there exists an algorithm $\mathcal{A}$ that, 
given black bock access to a single query to the oracle $\mathcal{O}(\mathcal{M}_{\mathrm{bounded}}^K, \rho)$, satisfies the following properties: $\mathcal{A}$ takes as input samples from any distribution $\mathcal{D}$ over the agent types, with the number of samples bounded by:
\begin{align*}
    N = O\left(\left(\frac{1}{\epsilon}\right)^2 \left( K \cdot m \cdot \log (K \cdot m \cdot n) 
      + \log \left(\frac{1}{\delta}\right) \right)\right).
\end{align*} 
With probability at least $1-\delta$,  $\mathcal{A}$ returns a contract ${t} \in \Cbounded$ with the following guarantee:
\begin{align*}
    \mathbb{E}_{\theta \sim \mathcal{D}}[u_p(\theta, {t})] \geq \mathrm{OPT}(\Cbounded, \mathcal{D}) - \epsilon - \rho.
\end{align*}
\end{theorem}
\begin{proof}
    Observe that for any menu of bounded contracts $M \in \mathcal{M}_{\mathrm{bounded}}^K$, the principal's utility is bounded in the range $[-1, 1]$ for any agent type. 
    A closer examination of the proof of \Cref{lem:learning_with_oracle} reveals that it applies not only to learning a single bounded contract but also to classes of menus of bounded contracts. Hence, we conclude the proof using \Cref{theorem:menu_general_contracts_pdim}.
\end{proof}

\subsection{Online Learning} \label{section:online}

In this section, we show how our results for offline learning of contracts can be extended to online learning. We first define the online learning variant of our model.

\begin{definition}[Online Learning Variant of Our Model] \label{def:online_learning_problem}
Fix any class of contracts $\mathcal{C}$.
We define an online learning algorithm for the contract class $\mathcal{C}$ as follows.
We assume that there is some fixed but unknown probability distribution $\mathcal{D}$ over the agent type space $\Thetaall$.
We consider a sequence of $T$ rounds for some $T > 0$. 
In each round $i = 1, \ldots, T$, the following events happen in the given order:
\begin{itemize}
    \item The algorithm commits to a contract $t_{i} \in \mathcal{C}$.
    \item An agent type $\theta_i$ is drawn independently from $\mathcal{D}$.
    \item The algorithm observes the type $\theta_{i}$.
    \item The algorithm receives a reward equal to $u_p(\theta_i, t_{i})$.
\end{itemize}
The chosen contract in round $i$ can depend only on the types observed in the previous round, i.e., $\theta_1, \ldots, \theta_{i-1}$. We define the regret of an online learning algorithm as:
\begin{align*}
    R_T = \sup_{\mathcal{D}} \left\{ T \cdot \mathrm{OPT}(\mathcal{C}, \mathcal{D})  - \mathbb{E}_{\theta_1, \ldots, \theta_T \sim \mathcal{D}}\left[ \sum_{i=1}^T u_p(\theta_i, t_i) \right]  \right\}.
\end{align*}
That is, the regret is the difference between the total reward of the optimal contract for $\mathcal{D}$ and the total expected reward of the algorithm for the worst-case distribution $\mathcal{D}$.
\end{definition}

The online learning problem studied in
\citet*{ZhuBYWJJ23} differs from the problem described above in its feedback model. In their problem, in round $i$, the algorithm observes only the realized outcome sampled given the agent's best response to contract $t_i$, not the type $\theta_i$. Consequently, the contracts defined in the following rounds can only depend on the outcomes observed in the previous rounds. In our model, the algorithm observes the full agent's type $\theta_i$. Note that the feedback model of \cite{ZhuBYWJJ23} can be simulated in our model: We can simply compute the agent's best response and sample an outcome from the corresponding distribution over outcomes. As we show in this section, the stronger feedback model allows us to obtain better regret guarantees compared to those in \cite{ZhuBYWJJ23}.

{
To analyze the regret in the online learning setting, we establish a general connection between the pseudo-dimension of a contract class and the achievable regret. Specifically, we show that if a contract class $\bar{\mathcal{C}}$ provides a good approximation (low representation error) for the target class $\mathcal{C} \subseteq \Cbounded$, then the regret can be bounded in terms of the pseudo-dimension of $\bar{\mathcal{C}}$. 
This result forms the foundation of our regret guarantees: by applying this lemma to the efficient approximations of linear and bounded contracts derived in \Cref{sec:pdim_upper_bounds}, we obtain improved upper bounds on regret for these contract classes.
}

\begin{lemma}[Regret Upper Bound via Approximation]\label{lem:online_learning_reduction}
Fix any class of contracts $\mathcal{C} \subseteq \Cbounded$.
Let $\bar{\mathcal{C}}$ be a contract class with representation error of at most $O(1/\sqrt{T})$ with respect to $\mathcal{C}$.
Then, we can get regret $R_T = O(\sqrt{T\cdot \pdim{\bar{\mathcal{C}}}} + \sqrt{T \log T})$ for learning over the contract class $\mathcal{C}$ in the online learning problem of \Cref{def:online_learning_problem}.
\end{lemma}
\begin{proof}

        Fix some $T > 0$. Let $\theta_1, \ldots, \theta_T$ be the agent types drawn from $\mathcal{D}$ in the $T$ rounds. 

    The idea behind the algorithm $\mathcal{A}$ is that, in each round $i$ it returns the empirically optimal bounded contract $t_i$ from that class $\bar{\mathcal{C}}$
    given the past $i - 1$ samples. 

We set the failure probability $\delta = T^{-3/2}$.
    By \Cref{theorem:uniform_learnability}(\ref{enum:epsilon_unif_learnability}) with probability at least $1 - \delta$, for the contract returned on round $i$ we have: 
    \begin{align*}
    \mathbb{E}_{\theta_{i} \sim \mathcal{D}} [u_p(\theta_{i}, t_{i}) | \theta_1, \ldots, \theta_{i-1}] &\geq \mathrm{OPT}({\bar{\mathcal{C}}}_{i}, \mathcal{D}) - O\left( (1/\sqrt{i}) \cdot \left( \sqrt{\pdim{\bar{\mathcal{C}}}} + \sqrt{\log(1/\delta)} \right) \right) \\
        &\geq \mathrm{OPT}(\bar{\mathcal{C}}, \mathcal{D}) - O\left( (1/\sqrt{i}) \cdot \left(  \sqrt{\pdim{{\bar{\mathcal{C}}}}} + \sqrt{\log T} \right) \right) \\
        &\geq \mathrm{OPT}(\mathcal{C}, \mathcal{D}) - O(1/\sqrt{T}) - O\left( (1/\sqrt{i}) \cdot \left(  \sqrt{\pdim{{\bar{\mathcal{C}}}}} + \sqrt{\log T} \right) \right) \\
        &\geq \mathrm{OPT}(\mathcal{C}, \mathcal{D}) -  O\left( (1/\sqrt{i}) \cdot \left(  \sqrt{\pdim{{\bar{\mathcal{C}}}}} + \sqrt{\log T} \right) \right) 
    \end{align*}
    
    By the union bound the inequality above holds simultaneously for all $i = 1, \ldots, T$ with probability at least $1-\delta T = 1-1/\sqrt{T}$. 
    If the inequality above does not hold, which happens with probability at most $\delta T$, note that $\mathrm{OPT}(\mathcal{C}, \mathcal{D}) \leq 1$ and $u_p(\theta_i, t_i) \geq 0$, which means we can bound $\mathrm{OPT}(\mathcal{C}, \mathcal{D}) - u_p(\theta_i, t_i) \leq 1$.
    Thus, we obtain the following regret bound:
    \begin{align*}
        T \cdot \mathrm{OPT}(\mathcal{C}, \mathcal{D}) 
 - \mathbb{E}\left[ \sum_{i = 1}^T u_p(\theta_{i}, t_{i}) \right] &\leq 
 (\delta T) \cdot T +  (1-\delta T)  \sum_{i=1}^T O\left(\frac{\sqrt{\pdim{\bar{\mathcal{C}}}} + \sqrt{\log T}}{\sqrt{i}}\right) \\
 &\leq \sqrt{T} + O\left(\sqrt{\pdim{\bar{\mathcal{C}}}}  + \sqrt{\log T}\right) \cdot O \left( \sum_{i=1}^{T} \frac{1}{\sqrt{i}} \right) \\
  &= \sqrt{T} + O\left(\sqrt{\pdim{\bar{\mathcal{C}}}}  + \sqrt{\log T}\right) \cdot O \left( \int_1^{T+1}  \frac{1}{\sqrt{\lfloor x \rfloor}} \,\mathrm{d}x
 \right) \\
 &\leq \sqrt{T} + O\left(\sqrt{\pdim{\bar{\mathcal{C}}}}  + \sqrt{\log T}\right) \cdot O \left(\int_{1}^{T+1} \frac{1}{\sqrt{x}} \,\mathrm{d}x\right) \\
 &=  \sqrt{T} + O\left(\sqrt{\pdim{\bar{\mathcal{C}}}}  + \sqrt{\log T}\right) \cdot O\left( 2 \sqrt{T+1} - 2 \right) \\
 &= O\left(\sqrt{T \cdot \pdim{\bar{\mathcal{C}}}}  + \sqrt{T \log T}\right)
    \end{align*}
    which completes the proof.
\end{proof}

\citet{ZhuBYWJJ23} show that with their feedback model, the optimal regret one can achieve for learning linear contracts is $\tilde{O}(T^{2/3})$. In our model, this can be improved to $\tilde{O}(T^{1/2})$.

\begin{theorem}[Regret Upper Bound for $\Clinear$]
    We can get regret $R_T =  O(\sqrt{T \log T})$ for learning over the class of linear contracts $\Clinear$ in the online learning problem of \Cref{def:online_learning_problem}.
\end{theorem}

\begin{proof}
    Recall that $\mathcal{L}_{1/\sqrt{T}}$ has pseudo-dimension $O(\log \sqrt{T}) = O(\log T)$ and representation error $1/\sqrt{T}$, as shown in \Cref{lemma:t_contracts_representation_error}. Therefore, the result follows by \Cref{lem:online_learning_reduction}.
\end{proof}

\citet{ZhuBYWJJ23} show that with their feedback model, the optimal regret one can achieve for learning bounded contracts is at most $\tilde{O}(\sqrt{m} \cdot T^{1-1/(2m+1)})$ and at least $\Omega(T^{1-1/(m+2)})$. In the following theorem, we show that in our model, this can be improved to $\tilde{O}(\sqrt{m} \cdot T^{1/2})$.

\begin{theorem}[Regret Upper Bound for $\Cbounded$]
    We can get regret $R_T =  O\left(   \sqrt{m \cdot T \log T} \right)$ for learning over the class of bounded contracts $\Cbounded$ in the online learning problem of \Cref{def:online_learning_problem}.
\end{theorem}

\begin{proof}
    Recall that $\mathcal{B}_{1/\sqrt{T}}$ has pseudo-dimension $ O(m \log(\sqrt{T})) = O(m \log T)$ and representation error $1/\sqrt{T}$, as shown in \Cref{lemma:se_contracts_pseudo-dimension}. Therefore, the result follows by \Cref{lem:online_learning_reduction}.
\end{proof}

Note that the lower bound of regret of $\Omega(T^{1-1/(m+2)})$ for the {bandit} feedback model of \citet*{ZhuBYWJJ23} implies that any learning algorithm requires an exponential number of samples to achieve constant regret, even if the distribution over agent types is a point mass on a single agent type.\footnote{The standard approach to converting a regret bound in an online learning setting to a sample complexity guarantee is to require the average regret in the online setting to be below the approximation error in the offline setting and solve for the number of required rounds, \( T \).
This means that if the regret is at least \( \Omega(T^{1-1/(m+2)}) \), then for the average regret, which is at least \( \Omega(T^{1-1/(m+2)}/T) \), to be smaller than \(\epsilon\), we require \( T \geq \Omega((1/\epsilon)^{m+2}) \).
} 
By comparison, with the full feedback model, our learning algorithm only requires a polynomial number of samples to achieve a constant regret, for any distribution over agent types. 
Thus, in our eyes, \citet*{ZhuBYWJJ23} and all other works that follow this feedback model focus primarily on the problem of learning the distribution over outcomes for the agent's action, rather than the problem of learning the agent type distribution.

\section{Conclusion and Future Work}

In this work, we formalize a framework for learning nearly-optimal contracts in an offline setting, where samples are drawn from the agent type distribution. 
A central tool in our analysis is the \emph{pseudo-dimension}, which measures the intrinsic complexity of a contract class.
We leverage structural insights on optimal contract design to derive essentially tight tradeoffs between representation error, pseudo-dimension and sample complexity for linear, bounded and unbounded contracts.
This enables the principal to balance contract expressivity with sample complexity.

We also extend our results to the combinatorial actions model. We show refined bounds that link the pseudo-dimension (and consequently sample complexity) of optimal contracts in that model to the number of critical values. Prior work had linked the number of critical values to the computational complexity of finding (near-)optimal contracts. Our work establishes a formal link between this concept and pseudo-dimension. As another extension we consider menus of contracts, and show that the pseudo-dimension scales linearly with menu size. This shows a very benign scaling behavior, when going from single contracts to menus of contracts.
Last but not least, we adapt our offline framework to the online learning setting under expert advice, and show polynomial sample complexity. 
This is in contrast to prior work, which showed that, with bandit feedback, exponentially many samples are required,  even if the agent is of a fixed but unknown type.
This reveals a key distinction between the two feedback models: while bandit feedback focuses on learning the distribution over outcomes for the actions of a single agent, we focus on learning the distribution over agent types.

Our work suggests several directions for future work.
A significant challenge lies in scenarios where even a polynomial number of samples from the agent's type distribution is infeasible.
Examining the few-sample regime, where the principal has access to only a constant number of samples, could provide meaningful insights; see \cite[e.g.,][]{DBLP:conf/sigecom/DhangwatnotaiRY10,DBLP:conf/ec/CorreaDFS19,DBLP:conf/focs/DuttingKLR024} for an analysis of this regime in related settings. 
It remains an open question how, and under what conditions, learning algorithms can be designed to guarantee a constant-factor approximation of the optimal contract with such limited data.

Another avenue for future research involves refining our results for the pseudo-dimension of approximations of the class of bounded contracts. While our bounds are tight in the region where $\epsilon < 1 /(24 \cdot m^2)$, a gap remains between our lower and upper bounds in the case where $\epsilon > 1/(24 \cdot m^2)$, leaving room for further improvement.

It would also be interesting to study other contract classes, beyond those studied in this work. 
Open questions in this domain include determining whether our pseudo-dimension upper bound for menus of contracts of bounded size (\Cref{theorem:menu_general_contracts_pdim}) is tight, and whether the pseudo-dimension of more expressive classes, such as menus of unbounded size or menus of randomized contracts, can be bounded.

The action query model \citep{ChenEtAl2024,dutting2021complexity} presents another promising avenue for analysis. In this model, the agent's type is fixed, and the principal observes samples from the distribution over outcomes induced by each action. A key open question is whether learning algorithms can be designed for this feedback model when the agent's type distribution is also unknown.

Finally, an important special case of the model studied in this paper is where the agent's type corresponds to their cost per unit of effort---a setting studied within the Bayesian framework \citep{alon2021contracts,AlonDLT23,CastiglioniCLXS25}. An open question is whether this additional structure can be leveraged to obtain improved pseudo-dimension and sample complexity bounds for the key contract classes analyzed in this work.

\section*{Acknowledgments}

We would like to thank Yurong Chen, Yoav Gal-Tzur, Franciszek Malinka, Nadav Merlis, Vishnu V. Narayan, and Maya Schlesinger for their helpful feedback and invaluable discussions. We also thank the audience at the FILOFOCS workshop (Paris, 2024) for their comments during the presentation of a preliminary version of this paper.

This project has been partially funded by the European Research Council (ERC) under the European Union's Horizon 2020 research and innovation program (grant agreement No. 866132), by an Amazon Research Award, by the NSF-BSF (grant number 2020788), by the Israel Science Foundation Breakthrough Program (grant No.2600/24), and by a grant from TAU Center for AI and Data Science (TAD).

\bibliographystyle{plainnat}
\bibliography{references}

\begin{thebibliography}{52}
\providecommand{\natexlab}[1]{#1}
\providecommand{\url}[1]{\texttt{#1}}
\expandafter\ifx\csname urlstyle\endcsname\relax
  \providecommand{\doi}[1]{doi: #1}\else
  \providecommand{\doi}{doi: \begingroup \urlstyle{rm}\Url}\fi

\bibitem[Alon et~al.(2021)Alon, D{\"u}tting, and {Talgam-Cohen}]{alon2021contracts}
Tal Alon, Paul D{\"u}tting, and Inbal {Talgam-Cohen}.
\newblock Contracts with private cost per unit-of-effort.
\newblock In \emph{EC 2021}, pages 52--69, 2021.

\bibitem[Alon et~al.(2023)Alon, D{\"{u}}tting, Li, and {Talgam{-}Cohen}]{AlonDLT23}
Tal Alon, Paul D{\"{u}}tting, Yingkai Li, and Inbal {Talgam{-}Cohen}.
\newblock Bayesian analysis of linear contracts.
\newblock In \emph{EC 2023}, page~66, 2023.

\bibitem[Babaioff et~al.(2006)Babaioff, Feldman, and Nisan]{babaioff2006combinatorial}
Moshe Babaioff, Michal Feldman, and Noam Nisan.
\newblock Combinatorial agency.
\newblock In \emph{EC 2006}, pages 18--28, 2006.

\bibitem[Bacchiocchi et~al.(2024)Bacchiocchi, Castiglioni, Marchesi, and Gatti]{bcmg23}
Francesco Bacchiocchi, Matteo Castiglioni, Alberto Marchesi, and Nicola Gatti.
\newblock Learning optimal contracts: How to exploit small action spaces.
\newblock In \emph{ICLR 2024}, 2024.

\bibitem[Balcan(2020)]{DBLP:books/cu/20/Balcan20}
Maria{-}Florina Balcan.
\newblock Data-driven algorithm design.
\newblock In \emph{Beyond the Worst-Case Analysis of Algorithms}, pages 626--645. Cambridge University Press, 2020.

\bibitem[Balcan et~al.(2016)Balcan, Sandholm, and Vitercik]{balcan16}
Maria-Florina Balcan, Tuomas Sandholm, and Ellen Vitercik.
\newblock Sample complexity of automated mechanism design.
\newblock In \emph{NIPS 2016}, pages 2091--2099, 2016.

\bibitem[Balcan et~al.(2018)Balcan, Sandholm, and Vitercik]{BalcanSV18}
Maria{-}Florina Balcan, Tuomas Sandholm, and Ellen Vitercik.
\newblock A general theory of sample complexity for multi-item profit maximization.
\newblock In \emph{{EC} 2018}, pages 173--174, 2018.

\bibitem[Balcan et~al.(2021)Balcan, DeBlasio, Dick, Kingsford, Sandholm, and Vitercik]{BalcanDDKSV21}
Maria{-}Florina Balcan, Dan~F. DeBlasio, Travis Dick, Carl Kingsford, Tuomas Sandholm, and Ellen Vitercik.
\newblock How much data is sufficient to learn high-performing algorithms? generalization guarantees for data-driven algorithm design.
\newblock In \emph{STOC 2021}, pages 919--932, 2021.

\bibitem[Bates et~al.(2024)Bates, Jordan, Sklar, and Soloff]{bates2022}
Stephen Bates, Michael~I. Jordan, Michael Sklar, and Jake~A. Soloff.
\newblock Principal-agent hypothesis testing, 2024.
\newblock URL \url{https://arxiv.org/abs/2205.06812}.

\bibitem[Beyeler et~al.(2024)Beyeler, Brero, Lubin, and Seuken]{brero2021machinelearningpowerediterativecombinatorial}
Manuel Beyeler, Gianluca Brero, Benjamin Lubin, and Sven Seuken.
\newblock {iMLCA:} {M}achine learning-powered iterative combinatorial auctions with interval bidding.
\newblock In \emph{EC 2024}, page 136, 2024.

\bibitem[Burkett and Rosenthal(2024)]{burkett2024}
Justin Burkett and Maxwell Rosenthal.
\newblock Data-driven contract design.
\newblock \emph{Journal of Economic Theory}, 221:\penalty0 105900, 2024.

\bibitem[Carroll(2015)]{caroll2015}
Gabriel Carroll.
\newblock {Robustness and Linear Contracts}.
\newblock \emph{American Economic Review}, 105\penalty0 (2):\penalty0 536--563, 2015.

\bibitem[Castiglioni et~al.(2021)Castiglioni, Marchesi, and Gatti]{CastiglioniM021}
Matteo Castiglioni, Alberto Marchesi, and Nicola Gatti.
\newblock Bayesian agency: Linear versus tractable contracts.
\newblock In \emph{EC 2021}, pages 285--286, 2021.

\bibitem[Castiglioni et~al.(2022)Castiglioni, Marchesi, and Gatti]{castiglioni2022designing}
Matteo Castiglioni, Alberto Marchesi, and Nicola Gatti.
\newblock Designing menus of contracts efficiently: the power of randomization.
\newblock In \emph{EC 2022}, pages 705--735, 2022.

\bibitem[Castiglioni et~al.(2025)Castiglioni, Chen, Li, Xu, and Zuo]{CastiglioniCLXS25}
Matteo Castiglioni, Junjie Chen, Minming Li, Haifeng Xu, and Song Zuo.
\newblock A reduction from multi-parameter to single-parameter bayesian contract design.
\newblock In \emph{SODA 2025}, pages 1795--1836, 2025.

\bibitem[Chen et~al.(2024)Chen, Chen, Deng, and Huang]{ChenEtAl2024}
Yurong Chen, Zhaohua Chen, Xiaotie Deng, and Zhiyi Huang.
\newblock Are bounded contracts learnable and approximately optimal?
\newblock In \emph{EC 2024}, pages 315--344, 2024.

\bibitem[Cohen et~al.(2022)Cohen, Deligkas, and Koren]{cohen2022}
Alon Cohen, Argyrios Deligkas, and Moran Koren.
\newblock Learning approximately optimal contracts.
\newblock In \emph{SAGT 2022}, pages 331--346, 2022.

\bibitem[Cole and Roughgarden(2014)]{ColeR14}
Richard Cole and Tim Roughgarden.
\newblock The sample complexity of revenue maximization.
\newblock In \emph{STOC 2014}, pages 243--252. {ACM}, 2014.

\bibitem[Correa et~al.(2019)Correa, D{\"{u}}tting, Fischer, and Schewior]{DBLP:conf/ec/CorreaDFS19}
Jos{\'{e}} Correa, Paul D{\"{u}}tting, Felix~A. Fischer, and Kevin Schewior.
\newblock Prophet inequalities for {I.I.D.} random variables from an unknown distribution.
\newblock In \emph{EC 2019}, pages 3--17, 2019.

\bibitem[Curry et~al.(2024)Curry, Thoma, Chakrabarti, McAleer, Kroer, Sandholm, He, and Seuken]{curry2024automateddesignaffinemaximizer}
Michael Curry, Vinzenz Thoma, Darshan Chakrabarti, Stephen McAleer, Christian Kroer, Tuomas Sandholm, Niao He, and Sven Seuken.
\newblock Automated design of affine maximizer mechanisms in dynamic settings.
\newblock In \emph{AAAI 2024}, pages 9626--9635, 2024.

\bibitem[{Deo-Campo Vuong} et~al.(2024){Deo-Campo Vuong}, Dughmi, Patel, and Prasad]{deo2024supermodular}
Ramiro {Deo-Campo Vuong}, Shaddin Dughmi, Neel Patel, and Aditya Prasad.
\newblock On supermodular contracts and dense subgraphs.
\newblock In \emph{SODA 2024}, pages 109--132, 2024.

\bibitem[Dhangwatnotai et~al.(2010)Dhangwatnotai, Roughgarden, and Yan]{DBLP:conf/sigecom/DhangwatnotaiRY10}
Peerapong Dhangwatnotai, Tim Roughgarden, and Qiqi Yan.
\newblock Revenue maximization with a single sample.
\newblock In \emph{EC 2010}, pages 129--138, 2010.

\bibitem[Diakonikolas(2019)]{diakonikolas2019lecture4}
Ilias Diakonikolas.
\newblock Lecture 4: Concentration inequalities, 2019.
\newblock URL \url{http://www.iliasdiakonikolas.org/teaching/Fall19/scribes/lec4.pdf}.

\bibitem[D{\"u}tting et~al.(2019)D{\"u}tting, Roughgarden, and {Talgam-Cohen}]{dutting2019simple}
Paul D{\"u}tting, Tim Roughgarden, and Inbal {Talgam-Cohen}.
\newblock Simple versus optimal contracts.
\newblock In \emph{EC 2019}, pages 369--387, 2019.

\bibitem[D\"utting et~al.(2021{\natexlab{a}})D\"utting, Ezra, Feldman, and Kesselheim]{Duetting2021CombinatorialC}
Paul D\"utting, Tomer Ezra, Michal Feldman, and Thomas Kesselheim.
\newblock Combinatorial contracts.
\newblock \emph{FOCS 2021}, pages 815--826, 2021{\natexlab{a}}.

\bibitem[D\"utting et~al.(2021{\natexlab{b}})D\"utting, Roughgarden, and {Talgam-Cohen}]{dutting2021complexity}
Paul D\"utting, Tim Roughgarden, and Inbal {Talgam-Cohen}.
\newblock The complexity of contracts.
\newblock \emph{SIAM Journal on Computing}, 50\penalty0 (1):\penalty0 211--254, 2021{\natexlab{b}}.

\bibitem[D{\"{u}}tting et~al.(2023)D{\"{u}}tting, Ezra, Feldman, and Kesselheim]{duetting2022multi}
Paul D{\"{u}}tting, Tomer Ezra, Michal Feldman, and Thomas Kesselheim.
\newblock Multi-agent contracts.
\newblock In \emph{STOC 2023}, pages 1311--1324, 2023.

\bibitem[D\"utting et~al.(2023)D\"utting, Guruganesh, Schneider, and Wang]{DuettingGuruganeshSchneiderWang23}
Paul D\"utting, Guru Guruganesh, Jon Schneider, and Joshua Wang.
\newblock Optimal no-regret learning for one-side lipschitz functions.
\newblock In \emph{ICML 2023}, pages 8836--8850, 2023.

\bibitem[D\"utting et~al.(2024)D\"utting, Feldman, and {Gal-Tzur}]{dutting2024combinatorial}
Paul D\"utting, Michal Feldman, and Yoav {Gal-Tzur}.
\newblock Combinatorial contracts beyond gross substitutes.
\newblock In \emph{SODA 2024}, pages 92--108, 2024.

\bibitem[D{\"{u}}tting et~al.(2024{\natexlab{a}})D{\"{u}}tting, Feldman, Tzur, and Rubinstein]{DBLP:journals/corr/abs-2403-09794}
Paul D{\"{u}}tting, Michal Feldman, Yoav~Gal Tzur, and Aviad Rubinstein.
\newblock The query complexity of contracts, 2024{\natexlab{a}}.
\newblock URL \url{https://doi.org/10.48550/arXiv.2403.09794}.

\bibitem[D{\"{u}}tting et~al.(2024{\natexlab{b}})D{\"{u}}tting, Kesselheim, Lucier, Reiffenh{\"{a}}user, and Singla]{DBLP:conf/focs/DuttingKLR024}
Paul D{\"{u}}tting, Thomas Kesselheim, Brendan Lucier, Rebecca Reiffenh{\"{a}}user, and Sahil Singla.
\newblock Online combinatorial allocations and auctions with few samples.
\newblock In \emph{{FOCS} 2024}, pages 1231--1250, 2024{\natexlab{b}}.

\bibitem[D\"utting et~al.(2025)D\"utting, Ezra, Feldman, and Kesselheim]{DEFK24}
Paul D\"utting, Tomer Ezra, Michal Feldman, and Thomas Kesselheim.
\newblock Multi-agent combinatorial contracts.
\newblock In \emph{SODA 2025}, page 1857–1891, 2025.

\bibitem[Dütting et~al.(2024)Dütting, Feldman, and Talgam-Cohen]{DuettingFTC24}
Paul Dütting, Michal Feldman, and Inbal Talgam-Cohen.
\newblock Algorithmic contract theory: A survey.
\newblock \emph{Foundations and Trends® in Theoretical Computer Science}, 16\penalty0 (3-4):\penalty0 211--412, 2024.

\bibitem[Ezra et~al.(2024)Ezra, Feldman, and Schlesinger]{ezra2023Inapproximability}
Tomer Ezra, Michal Feldman, and Maya Schlesinger.
\newblock On the (in)approximability of combinatorial contracts.
\newblock \emph{ITCS 2024}, pages 44:1--44:22, 2024.

\bibitem[Frazier et~al.(2014)Frazier, Kempe, Kleinberg, and Kleinberg]{frazier14}
Peter Frazier, David Kempe, Jon Kleinberg, and Robert Kleinberg.
\newblock Incentivizing exploration.
\newblock In \emph{EC 2014}, page 5–22, 2014.

\bibitem[{Goldman Sachs}(2022)]{goldmansachs2022creatorEconomy}
{Goldman Sachs}.
\newblock The creator economy could approach half a trillion dollars by 2027, 2022.
\newblock URL \url{https://www.goldmansachs.com/insights/articles/the-creator-economy-could-approach-half-a-trillion-dollars-by-2027}.

\bibitem[Grossman and Hart(1983)]{grossman1983}
Sanford~J. Grossman and Oliver~D. Hart.
\newblock An analysis of the principal-agent problem.
\newblock \emph{Econometrica}, 51\penalty0 (1):\penalty0 7--45, 1983.

\bibitem[Guruganesh et~al.(2021)Guruganesh, Schneider, and Wang]{GSW21}
Guru Guruganesh, Jon Schneider, and Joshua~R. Wang.
\newblock Contracts under moral hazard and adverse selection.
\newblock In \emph{{EC} 2021}, pages 563--582, 2021.

\bibitem[Guruganesh et~al.(2023)Guruganesh, Schneider, Wang, and Zhao]{GuruganeshSW023}
Guru Guruganesh, Jon Schneider, Joshua~R. Wang, and Junyao Zhao.
\newblock The power of menus in contract design.
\newblock In \emph{EC 2023}, pages 818--848, 2023.

\bibitem[Ho et~al.(2014)Ho, Slivkins, and Vaughan]{ho2014adaptive}
Chien-Ju Ho, Aleksandrs Slivkins, and Jennifer~Wortman Vaughan.
\newblock Adaptive contract design for crowdsourcing markets: Bandit algorithms for repeated principal-agent problems.
\newblock In \emph{EC 2014}, pages 359--376, 2014.

\bibitem[Holmström(1979)]{holmstrom1979}
Bengt Holmström.
\newblock Moral hazard and observability.
\newblock \emph{The Bell Journal of Economics}, 10\penalty0 (1):\penalty0 74--91, 1979.

\bibitem[Li et~al.(2001)Li, Long, and Srinivasan]{DBLP:journals/jcss/LiLS01}
Yi~Li, Philip~M. Long, and Aravind Srinivasan.
\newblock Improved bounds on the sample complexity of learning.
\newblock \emph{J. Comput. Syst. Sci.}, 62\penalty0 (3):\penalty0 516--527, 2001.

\bibitem[Morgenstern and Roughgarden(2015)]{morgenstern2015pseudodimension}
Jamie Morgenstern and Tim Roughgarden.
\newblock The pseudo-dimension of near-optimal auctions, 2015.

\bibitem[{Royal Swedish Academy of Sciences}(2016)]{Nobel2016}
{Royal Swedish Academy of Sciences}.
\newblock Scientific background on the 2016 {N}obel {P}rize in {E}conomic, 2016.
\newblock URL \url{https://www.nobelprize.org/prizes/economic-sciences/2016/press-release/}.

\bibitem[Scheid et~al.(2024)Scheid, Tiapkin, Boursier, Capitaine, Moulines, Jordan, El-Mhamdi, and Oliviero~Durmus]{scheid24incentivizedLearning}
Antoine Scheid, Daniil Tiapkin, Etienne Boursier, Aymeric Capitaine, Eric Moulines, Michael Jordan, El-Mahdi El-Mhamdi, and Alain Oliviero~Durmus.
\newblock Incentivized learning in principal-agent bandit games.
\newblock In \emph{ICML 2024}, pages 43608--43631, 2024.

\bibitem[Schorfheide and Wolpin(2012)]{schorfheide2012}
Frank Schorfheide and Kenneth~I. Wolpin.
\newblock On the use of holdout samples for model selection.
\newblock \emph{American Economic Review}, 102\penalty0 (3):\penalty0 477–81, 2012.

\bibitem[Schorfheide and Wolpin(2016)]{schorfheide2016}
Frank Schorfheide and Kenneth~I. Wolpin.
\newblock To hold out or not to hold out.
\newblock \emph{Research in Economics}, 70\penalty0 (2):\penalty0 332--345, 2016.

\bibitem[Soumalias et~al.(2024{\natexlab{a}})Soumalias, Heiss, Weissteiner, and Seuken]{soumalias2024pricesbidsvalueseverything}
Ermis Soumalias, Jakob Heiss, Jakob Weissteiner, and Sven Seuken.
\newblock Prices, bids, values: Everything, everywhere, all at once, 2024{\natexlab{a}}.
\newblock URL \url{https://arxiv.org/abs/2411.09355}.

\bibitem[Soumalias et~al.(2024{\natexlab{b}})Soumalias, Zamanlooy, Weissteiner, and Seuken]{soumalias2024machinelearningpoweredcourseallocation}
Ermis Soumalias, Behnoosh Zamanlooy, Jakob Weissteiner, and Sven Seuken.
\newblock Machine learning-powered course allocation.
\newblock In \emph{EC 2024}, page 1099, 2024{\natexlab{b}}.

\bibitem[Soumalias et~al.(2024{\natexlab{c}})Soumalias, Weissteiner, Heiss, and Seuken]{Soumalias_2024}
Ermis~Nikiforos Soumalias, Jakob Weissteiner, Jakob Heiss, and Sven Seuken.
\newblock Machine learning-powered combinatorial clock auction.
\newblock In \emph{AAAI 2024}, page 9891–9900, 2024{\natexlab{c}}.

\bibitem[Tsybakov(2009)]{DBLP:books/daglib/0035708}
Alexandre~B. Tsybakov.
\newblock \emph{Introduction to Nonparametric Estimation}.
\newblock Springer Series in Statistics. Springer, Springer New York, NY, USA, 2009.

\bibitem[Zhu et~al.(2023)Zhu, Bates, Yang, Wang, Jiao, and Jordan]{ZhuBYWJJ23}
Banghua Zhu, Stephen Bates, Zhuoran Yang, Yixin Wang, Jiantao Jiao, and Michael~I. Jordan.
\newblock The sample complexity of online contract design.
\newblock In \emph{EC 2023}, page 1188, 2023.

\end{thebibliography}

\appendix
\section*{Appendix}

\section{Proofs Omitted from Section~\ref{sec:contract_prelims}}\label{sec:proofs_for_contracts_prelims}

\begin{restatable}[Linear Contracts in the Binary Outcome Model]
    {lemma}{equivalenceofbinaryoutcome} \label{lem:binary_outcome}
        Let $R > 0$. Let $\Theta_{\mathrm{general}}$  denote the space of all agent types for the general $m$-outcome model with rewards $r_0 = 0$ and $r_1 = R$ and $r_2, \ldots, r_{m-1} \leq R$.
     Let $\Theta_{\mathrm{binary}}$ denote the space of all agent types for the binary outcome model with rewards $r_0' =0$ and $r_1' = R$.
     Then:
    \begin{enumerate}[label=(\alph*)]
        \item For every distribution over $\Theta_{\mathrm{binary}}$, a linear contract is optimal among all unbounded contracts.\label{enum:lin_opt}
        \item There is an onto map $\psi : \Theta_{\mathrm{general}} \to \Theta_{\mathrm{binary}}$ such that $u_p(\theta, \alpha \cdot r) = u_p(\psi(\theta), \alpha \cdot r')$ for every agent type $\theta \in \Theta_{\mathrm{general}}$ and every linear contract $\alpha \in [0,1]$.\label{enum:binary_general}
    \end{enumerate}
\end{restatable}

\begin{proof}
    (a) 
   Let $t = (t_0, t_1) \in \Cunbounded$ be the optimal contract for a given distribution $\mathcal{D}$ over $\Theta_{\text{binary}}$. By definition, we have $t_0 \geq 0$. If $t_1 \geq R$, then $u_p(\theta, t) \leq 0$ since $r'_0 - t_0 \leq 0$ and $r'_1 - t_1 \leq 0$, implying that $u_p(\mathcal{D}, t) \leq u_p(\mathcal{D}, (0,0))$. Hence, we can assume without loss of generality that $t_1 \leq R$.
    
    Suppose that $t_0 > 0$. Consider the contract $t' = (0, t_1) \in \Cunbounded$. Let $\theta \in \Theta_{\mathrm{binary}}$, and let $i = i^\star(\theta, t)$ and $i' = i^\star(\theta, t')$. We must have:
    \begin{align*}
        t_0 \cdot f_{i,0} + t_1 \cdot f_{i,1} - c_i =  u_a(\theta, t, i)&\geq u_a(\theta, t, i') = t_0 \cdot f_{i',0} + t_1 \cdot f_{i',1} - c_{i'}
       \intertext{and}
        t_1 \cdot f_{i', 1} - c_{i'} =  u_a(\theta, t', i') &\geq u_a(\theta, t', i) = t_1 \cdot f_{i,1} - c_i.
    \end{align*}
   By summing the above inequalities, we get $t_0 \cdot f_{i', 0} \leq t_0 \cdot f_{i, 0}$, implying $f_{i',0} \leq f_{i,0}$. Since $f_{i', 1} = 1 - f_{i', 0}$ and $f_{i, 1} = 1 - f_{i,0}$, it follows that:
    \begin{align*}
        u_p(\theta, t') = f_{i',0} \cdot (0 - 0) + f_{i',1} \cdot (R - t_1) \geq f_{i,0} \cdot (0 - t_0) + f_{i,1} \cdot (R - t_1) = u_p(\theta, t).
    \end{align*}
    Thus, contract $t'$ generates a higher utility for the principal for every fixed agent type, and therefore for any distribution over agent types. Hence, we can assume that $t_0 = 0$. Then, we have $t = (0, t_1) = (t_1 / R) \cdot r'$, which shows that $t$ is a linear contract.
    This concludes the proof of part (a).

    (b) Fix an agent type $\theta = (f,c) \in \Theta_{\mathrm{general}}$. We define an agent type $\psi(\theta) = (f', c') \in \Theta_{\mathrm{binary}}$ by letting $f_{i, 1}' = (f_{i} \cdot r)/R$ and $f_{i,0}' = 1 - f_{i,1}'$ and $c_{i}' = c_i$ for every action $i \in [n]$.
    In terms of the agent's utility, for every linear contract $\alpha \in [0,1]$, we get:
    \begin{align*}
        u_a(\theta, \alpha \cdot r, i) = f_{i} \cdot r - c_{i} = f_{i}' \cdot r' - c_{i}' = u_a(\psi(\theta), \alpha \cdot r', i).
    \end{align*}
    It follows that $i^\star(\theta, \alpha \cdot r) = i^\star(\psi(\theta),\alpha \cdot r')$. Thus, letting $i = i^\star(\theta, \alpha \cdot r)$, we get:
    \begin{align*}
        u_p(\theta, \alpha \cdot r) = f_{i} \cdot (1-\alpha) \cdot r = f_{i}' \cdot (1-\alpha) \cdot r' = u_p(\psi(\theta), \alpha \cdot r', i).
    \end{align*}
    This completes the proof of the lemma.
\end{proof}

\section{Proofs Omitted from Section~\ref{sec:learning_theory_preliminaries}}\label{sec:proofs_from_learning_theory_preliminaries}

\empiricalmaximization*
\begin{proof}
Let $p^\star \in \mathcal{P}$ be the optimal {parametrization of $\mathcal{F}$} such that $\mathbb{E}_{i \sim \mathcal{D}}[\mathcal{F}(i, p^\star)] = \mathrm{OPT}(\mathcal{F}, \mathcal{D})$.
By the definition of uniform learnability, we have that with probability at least $1-\delta$, the following holds: $|\mathbb{E}_{i \sim \mathcal{D}}[\mathcal{F}(i, p)] - (1/|S|) \cdot \sum_{i \in S} \mathcal{F}(i, p)| \leq \epsilon$ for all parameters $p \in \mathcal{P}$. Thus, with probability at least $1-\delta$, it holds that:
\begin{align*}
    \mathbb{E}_{i \sim \mathcal{D}}[\mathcal{F}(i,\widehat{p})] &\geq \frac{1}{|S|} \cdot \sum_{i \in S} \mathcal{F}(i, \widehat{p}) - \epsilon && (\text{by uniform learnability of } \mathcal{F})\\
    &\geq \frac{1}{|S|} \cdot \sum_{i \in S} \mathcal{F}(i, p^\star) - \epsilon && (\text{by the choice of } \widehat{p})\\
    &\geq \mathbb{E}_{i \sim \mathcal{D}}[\mathcal{F}(i,p^\star)] - 2 \cdot \epsilon && (\text{by uniform learnability of } \mathcal{F}) \\
    &= \mathrm{OPT}(\mathcal{F}, \mathcal{D}) -2 \cdot \epsilon && (\text{by the choice of } p^\star)
\end{align*}
and the result follows.
\end{proof}

\pdimproperties*

\begin{proof}
We first prove part \ref{enum:pdim_size}. The pseudo-dimension of a class is defined as the maximum cardinality $k$ of a set of points that can be shattered by the class, i.e., 
there must exist a set 
$ S = \{ (i_1, \ldots, i_k \}$ of $k$ and a set of real numbers ${\tau^1, \tau^2, \dots, \tau^k}$, such that for all $T \subseteq S$
there exists a function $f$ in the class such that for all $i \in T$, it holds that
$ f(i) \ge \tau^i $ and for all $j \notin T$, it holds that
$ f(j) < \tau^{j} $.
However, any set of size $k$ supports exactly $2^k$ distinct labelings, and a single function $f$, from any class, only induces one labeling. 
Thus, we have $2^k \le |\mathcal{F} |$, and hence $\pdim{\mathcal{F}} = k \leq \log |\mathcal{F}|$.

Now we prove part \ref{enum:pdim_mono}.
Suppose that the pseudo-dimension of $\mathcal{F}$ is $k$. 
Then, for any set of inputs $S = \{i_1, \ldots , i_{k} \}$
and corresponding thresholds $\tau^1, \ldots , \tau^{k} \in \mathbb{R}$, 
there exists at least one $T \subseteq S$ such that there exists no function $f$ in the class satisfying:
$f(i) \ge \tau_i$ for all $i \in T$ and $f(j) < \tau_{j}$ for all $j \in S \setminus T$. 
However, in that case, the same property will be true for the class $\mathcal{F'}$, since any function in $\mathcal{F'}$ is also in $\mathcal{F}$.
\end{proof}

\section{Proofs Omitted from Section~\ref{sec:pdim}}\label{sec:proofs_for_pdim}

\pdimlbcontractconstruction*
\begin{proof}
  Consider a distribution $\mathcal{D}$ that is uniform over $m$ agent types $\theta^{(0)}, \theta^{(1)}, \ldots, \theta^{(m-1)}$, defined as follows.
  Agent type $\theta^{(0)} = (f^{(0)}, c^{(0)})$ is defined as:
    \begin{table}[H]
    \centering
    \renewcommand{\arraystretch}{1.5}
    \begin{tabular}{|c|>{\centering\arraybackslash}m{2cm}>{\centering\arraybackslash}m{2cm}|>{\centering\arraybackslash}m{2cm}|}
        \hline
        & \textbf{outcome $0$} & \textbf{outcome $j \in \setoutcomes \setminus \{0\}$} & \textbf{cost} \\ \hline
        \textbf{action $0$:} & 
        \vspace{0.2cm}$\displaystyle f^{(0)}_{0,0} = 1$\vspace{0.2cm} & 
        \vspace{0.2cm}$\displaystyle f^{(0)}_{0,j} = 0$\vspace{0.2cm} & 
        \vspace{0.2cm}$\displaystyle c_0 = 0$\vspace{0.2cm} \\ 
        \textbf{action $i \geq 1$:} & 
        $\displaystyle f^{(0)}_{i,0} = 1$\vspace{0.2cm} & 
        $\displaystyle f^{(0)}_{1,j} = 0$\vspace{0.2cm} & 
       $\displaystyle c_i = \infty$\vspace{0.2cm} \\ \hline
    \end{tabular}
\end{table}
In words, any action chosen by agent of type $\theta^{(0)}$ leads to outcome $0$ with probability $1$.
At the same time, the agent's cost for action $0$ is $0$, and her cost for any other action is $\infty$.
   For any $j \in \setoutcomes \setminus \{0\}$, we define agent type $\theta^{(j)} = (f^{(j)}, c^{(j)})$ as follows: 
    \begin{table}[H]
    \centering
    \renewcommand{\arraystretch}{1.5}
    \begin{tabular}{|c|>{\centering\arraybackslash}m{2.3cm}>{\centering\arraybackslash}m{2.3cm}>{\centering\arraybackslash}m{2.3cm}|>{\centering\arraybackslash}m{2cm}|}
        \hline
        & \textbf{outcome $0$} & \textbf{outcome $j$} & \textbf{outcome $j' \in \setoutcomes \setminus \{ 0, j \}$} & \textbf{cost} \\ \hline
        \textbf{action $0$:} & 
        \vspace{0.2cm}$\displaystyle f^{(j)}_{0,0} = 1$\vspace{0.2cm} & 
        \vspace{0.2cm}$\displaystyle f^{(j)}_{0,j} = 0$\vspace{0.2cm} & 
        \vspace{0.2cm}$\displaystyle f^{(j)}_{0,j'} = 0$\vspace{0.2cm} & 
        \vspace{0.2cm}$\displaystyle c_0 = 0$\vspace{0.2cm} \\ 
        \textbf{action $1$:} & 
        $\displaystyle f^{(j)}_{1,0} = 0$\vspace{0.2cm} & 
        $\displaystyle f^{(j)}_{1,j} = 1$\vspace{0.2cm} & 
        $\displaystyle f^{(j)}_{1,j'} = 0$\vspace{0.2cm} & 
        $\displaystyle c_1 = \alpha_j$\vspace{0.2cm} \\ 
        \textbf{action $i \geq 2$:} & 
        $\displaystyle f^{(j)}_{i,0} = 1$\vspace{0.2cm} & 
        $\displaystyle f^{(j)}_{i,j} = 0$\vspace{0.2cm} & 
        $\displaystyle f^{(j)}_{i,j'} = 0$\vspace{0.2cm} & 
        $\displaystyle c_i = \infty$\vspace{0.2cm} \\ \hline
    \end{tabular}
\end{table}
In words, action $1$ leads to outcome $j$ with probability $1$, while any other action leads to outcome $0$ with probability $1$. 
The cost of action $0$ is zero, the cost of action $1$ is $\alpha_j$ and the cost of any other action is $\infty$. Thus, the agent only considers actions $0$ and $1$. Consider the contract $t^\star = (0, \alpha_1, \alpha_2, \ldots, \alpha_{m-1}) \in \Cbounded$ and observe the following:
    \begin{align*}
        \mathbb{E}_{\theta \sim \mathcal{D}}[u_p(\theta, t^\star)] =  \sum_{j=0}^{m-1} \frac{1}{m} \cdot u_p(\theta^{(j)}, t^\star) = \sum_{j=1}^{m-1} \frac{1}{m}\left(r_j-\alpha_j\right)
    \end{align*}
     Let $t \in \mathcal{C}$ be such that
     $\mathbb{E}_{\theta \sim \mathcal{D}}[u_p(\theta, t)] \geq \mathrm{OPT}(\mathcal{D}, \Cbounded) - \epsilon/m$. Recall that by our assumption on $\mathcal{C}$, there must be such a contract $t$.
    Suppose that $t_0 > \epsilon$. 
    Note that $i^\star(\theta^{(0)}, t) = 0$, regardless of $t$, since the costs of all actions other than $0$ are set to $\infty$. We have:
\begin{align*}
       \mathbb{E}_{\theta \sim \mathcal{D}}[u_p(\theta, t)] 
       & = \sum_{j=0}^{m-1} \frac{1}{m} \cdot u_p(\theta^{(j)}, t) \\
       & < -\frac{\epsilon}{m} + \sum_{j=1}^{m-1} \frac{1}{m} \cdot u_p(\theta^{(j)}, t) \\
       &\leq  - \frac{\epsilon}{m}  +  \mathbb{E}_{\theta \sim \mathcal{D}}[u_p(\theta, t^\star)], 
\end{align*}
where the first inequality follows from the fact that the contract will pay $t_0 > \epsilon$ to agent $\theta^{(0)}$, 
and the second inequality from the fact that for any other agent type $\theta^{(j')}$, transfer $t^{\star}_{j'}$ is the minimum transfer required to induce action $1$ from the agent.
This, however, contradicts our assumption that $\mathbb{E}_{\theta \sim \mathcal{D}}[u_p(\theta, t)] \geq \mathrm{OPT}(\mathcal{D}, \Cbounded) - \epsilon/m$. 
    Thus, it must hold that $t_0 \leq \epsilon$.
    If for $\ell \in \setoutcomes \setminus \{0\}$ we have $t_\ell < \alpha_\ell$, then $i^\star(\theta^{(\ell)}, t) = 0$, and therefore:
    \begin{align*}
          \mathbb{E}_{\theta \sim \mathcal{D}}[u_p(\theta, t)]
          &\leq \sum_{j \notin \{0,\ell\}} \frac{1}{m}\left(r_j-\alpha_j\right) 
          \\
          &= \mathbb{E}_{\theta \sim \mathcal{D}}[u_p(\theta, t^\star)] - \frac{1}{m}\left(r_{\ell}-\alpha_{\ell}\right) \\
          &< \mathbb{E}_{\theta \sim \mathcal{D}}[u_p(\theta, t^\star)] - \frac{\epsilon}{m},
    \end{align*}
    where the last inequality follows from the assumption that $\epsilon <  \min_{j \in [m] \setminus \{0\}} (r_j - \alpha_j)$.
    Finally, if $t_\ell > \alpha_{\ell} + \epsilon$, then $i^\star(\theta^{(\ell)}, t) = 1$, and therefore:
    \begin{align*}
         \mathbb{E}_{\theta \sim \mathcal{D}}[u_p(\theta, t)]
              < \sum_{j \notin \{0,\ell\}} \frac{1}{m}\left(r_j-\alpha_j\right) + \frac{1}{m}\left(r_{\ell}-\alpha_{\ell} - \epsilon\right)
          = \mathbb{E}_{\theta \sim \mathcal{D}}[u_p(\theta, t^\star)] - \frac{\epsilon}{m}
    \end{align*}
Thus, these two observations imply that for any $t$ such that $\mathbb{E}_{\theta \sim \mathcal{D}}[u_p(\theta, t)] \geq \mathrm{OPT}(\mathcal{D}, \Cbounded) - \epsilon/m$, it must also hold that $t_\ell \in [\alpha_\ell, \alpha_\ell + \epsilon]$, as needed.
\end{proof}

\pdimlbtypeconstruction*

\begin{proof}
The agent's type \(\theta^{(\alpha, S,j)} = (f^{(\alpha, j)}, c^{(\alpha, S)})\) is defined by the production function:\begin{table}[H]
    \centering
    \renewcommand{\arraystretch}{1.25}
    \begin{tabular}{|c|>{\centering\arraybackslash}m{3cm}>{\centering\arraybackslash}m{3cm}>{\centering\arraybackslash}m{3cm}|}
        \hline
        & \textbf{outcome $0$} & \textbf{outcome $j$} & \textbf{outcome $j' \in \setoutcomes \setminus \{j, 0\}$} \\ \hline 
        \textbf{action $0$:} & 
        \vspace{0.2cm}$\displaystyle f^{(\alpha,j)}_{0,0} = 1$\vspace{0.2cm}  & 
        \vspace{0.2cm}$\displaystyle f^{(\alpha,j)}_{0,j} = 0$\vspace{0.2cm}  & 
        \vspace{0.2cm}$\displaystyle f_{0,j'}^{(\alpha,j)} = 0 $\vspace{0.2cm}\\
        \textbf{action $i>0$:} & 
        \vspace{0.2cm}$\displaystyle f^{(\alpha,j)}_{i,0} = \frac{\alpha_{\ell}-\alpha_{i}}{r_j-\alpha_i}$\vspace{0.2cm}  & 
        \vspace{0.2cm}$\displaystyle f^{(\alpha,j)}_{i,j} = \frac{r_j-\alpha_{\ell}}{r_j-\alpha_i}$\vspace{0.2cm}  & 
        \vspace{0.2cm}$\displaystyle f_{i,j'}^{(\alpha,j)} = 0 $\vspace{0.2cm}\\
        \hline
    \end{tabular}
\end{table}
    \noindent and the cost vector:
     \begin{table}[H]
        \centering
        \renewcommand{\arraystretch}{1.25}
         \begin{tabular}{|c|>{\centering\arraybackslash}m{7cm}|}
        \hline
        & \textbf{cost} \\ \hline
        \textbf{action $0$:} & \vspace{0.2cm}$\displaystyle {c^{(\alpha, S)}_{0} = 0}$ \vspace{0.2cm}\\ \textbf{action $i \in S$:} & \vspace{0.2cm}$\displaystyle {c^{(\alpha, S)}_{i} = \sum_{j=1}^i \alpha_i\left(\frac{r_j-\alpha_{\ell}}{r_j-\alpha_j} - \frac{r_j-\alpha_{\ell}}{r_j-\alpha_{j-1}}\right)}$ \vspace{0.2cm}\\ 
        \textbf{action $i' \notin S$:}  & $\displaystyle c^{(\alpha, S)}_{i'} = \infty $\vspace{0.2cm} \\  \hline
        \end{tabular}
    \end{table}

We first anaylze the best responses of an agent with type $\theta^{(\alpha, S, j)}$, given any contract $t$ with $t_j - t_0 \in [\alpha_i, \alpha_{i+1})$ for some $i \in \{0, 1, \ldots, \ell\}$. We show that if $i \in S$, then the best response is $i^\star(\theta^{(\alpha, S, j)}, t) = i$, and if $i \notin S$, then the best response is $i^\star(\theta^{(\alpha, S, j)}, t) < i$.
To show this, we write the agent's utility for action $w \in S$ as:
\begin{align*}
    u_a(\theta^{(\alpha, S, j)}, t, w) = t_0 \cdot f_{w, 0}^{(\alpha, j)} + t_j \cdot f_{w, j}^{(\alpha, j)} - c_w^{(\alpha, S)} 
    = t_0 + (t_j - t_0) \cdot f_{w,j}^{(\alpha, j)} - c_w^{(\alpha, S)}. 
\end{align*}
Note that:
\begin{align*}
    (t_j - t_0) \cdot f_{w,j}^{(\alpha, j)} &= (t_j - t_0) \cdot \left(\frac{1-\alpha_{\ell}}{1-\alpha_{w}}\right)  = (t_j - t_0) \cdot \left( \frac{1-\alpha_{\ell}}{1-\alpha_0} + \sum_{k=1}^{w} \left(\frac{1-\alpha_{\ell}}{1-\alpha_k} - \frac{1-\alpha_{\ell}}{1-\alpha_{k-1}}\right)\right) \\
    c_w^{(\alpha, S)} &= \sum_{k=1}^{w} \alpha_k \left(\frac{1-\alpha_{\ell}}{1-\alpha_k} - \frac{1-\alpha_{\ell}}{1-\alpha_{k-1}}\right)
\end{align*}
Therefore, we have:
\begin{align*}
    u_a(\theta^{(\alpha, S, j)}, t, w) = t_0 + (t_j - t_0) \cdot \left(\frac{1-\alpha_{\ell}}{1-\alpha_0}\right) + \sum_{k=1}^{w} ((t_j - t_0) - \alpha_k) \left(\frac{1-\alpha_{\ell}}{1-\alpha_k} - \frac{1-\alpha_{\ell}}{1-\alpha_{k-1}}\right)
\end{align*}
Since $\alpha_{k-1} < \alpha_{k}$, the term $\frac{1-\alpha_{\ell}}{1-\alpha_k} - \frac{1-\alpha_{\ell}}{1-\alpha_{k-1}}$ is always positive. Additionally, we have $(t_j - t_0) - \alpha_k$ is non-negative for all $\alpha_k \leq t_j - t_0$ and is negative for all $\alpha_k > \alpha$. We also note that the agent's utility for taking action $0$ is $u_a(\theta^{(\alpha, S, j)}, t, 0) = t_0 \cdot f_{0,1}^{(\alpha,j)} - c_0^{(\alpha, S)} = t_0 \cdot 1 - 0 = t_0$ and for taking action $w \notin S$ is $-\infty$ since $c_w = \infty$. This proves the claim about the agent's best responses.

Let us now consider a contract $t$ with $t_j - t_0 \in [\alpha_i, \alpha_i + \rho)$ for some $i \in S$. 
Note that $t_j - t_0 \leq \alpha_i + \rho \leq \alpha_i + (1/3) \cdot (\alpha_{i+1} - \alpha_i) < \alpha_{i+1}$.
Thus, by the observation above, the agent's best response for each contract $t$ with $t_j - t_0 = \alpha_i$ for $i \in S$ is action $i$, i.e., $i^\star(\theta^{(\alpha, S, j)}) = i$. Since $t_j - t_0 < \alpha_i + \rho$ and $t_0 < \rho$, we can bound the principal's utility:
\begin{align*}
    u_p(\theta^{(\alpha, S, j)}, t) &= f_{i,0}^{(\alpha, j)} \cdot (0 - t_0) + f_{i,j}^{(\alpha, j)} \cdot (r_j- t_j) 
    \\
    &= f_{i,j}^{(\alpha, j)} \cdot (r_j- (t_j-t_0)) - t_0\\ 
    &= \left(\frac{r_j-\alpha_{\ell}}{r_j-\alpha_i}\right) \cdot (r_j-(t_j-t_0)) - t_0 \\
   &\geq \left(\frac{r_j-\alpha_{\ell}}{r_j-\alpha_i}\right) \cdot (r_j-\alpha_i - \rho) - \rho \\
   &= (r_j - \alpha_{\ell}) - \left(\frac{r_j-\alpha_{\ell}}{r_j-\alpha_i}\right) \cdot \rho - \rho \\
   &\geq (r_j - \alpha_{\ell}) - 2 \cdot \rho.
\end{align*}

Let us now consider a contract $t$ with $t_j - t_0 \in [\alpha_i, \alpha_{i+1})$ for some $i \notin S$ with $i > 0$. The agent's best response is some action $i' < i$, i.e., $i' = i^\star(\theta^{(\alpha, S, j)}) < i$.
We can thus bound the principal's utility:
\begin{align*}
    u_p(\theta^{(\alpha, S, j)}, t) &= f_{i',0}^{(\alpha,j)} \cdot (0-t_0) 
 + f_{i',1}^{(\alpha,j)} \cdot (r_j-t_j) \\ 
 &= f_{i', j}^{(\alpha, j)} \cdot (r_j - (t_j-t_0)) - t_0\\
 &= \left(\frac{r_j-\alpha_{\ell}}{r_j-\alpha_{i'}}\right) \cdot (r_j-(t_j-t_0)) - t_0 \\
  &\leq \left(\frac{r_j-\alpha_{\ell}}{r_j-\alpha_{i'}}\right) \cdot (r_j-\alpha_i)\\
    &\leq (r_j-\alpha_{\ell}) \cdot \left(\frac{r_j-\alpha_{i}}{r_j-\alpha_{i'}}\right) \\
    &= (r_j-\alpha_{\ell}) \cdot \left(1-\frac{\alpha_{i}-\alpha_{i'}}{r_j-\alpha_{i'}}\right) \\
    &\leq (r_j-\alpha_{\ell}) \cdot \left(1-(\alpha_{i}-\alpha_{i'})\right) \\
    &= (r_j-\alpha_{\ell}) - (r_j-\alpha_{\ell}) \cdot (\alpha_{i} - \alpha_{i'})\\
    &\leq (r_j-\alpha_{\ell}) - 3 \cdot \rho.
\end{align*}
since $\rho \leq (1/3) \cdot (r_j - \alpha_{\ell}) \cdot \min_{0 \leq w \leq \ell-1} (\alpha_{w+1}- \alpha_w)$.

Finally, we consider the principal's utility for a contract $t$ with $t_j - t_0 \in [0, \alpha_1)$. We have that the agent's best response is action $0$, and so:
\begin{align*}
    u_p(\theta^{(\alpha, S, j)}, t) = f_{0,0}^{(\alpha, j)} \cdot (0 - t_0)
     = - t_0 \leq 0 \leq (r_j - \alpha_{\ell}) - 3 \cdot \rho.
\end{align*}
This completes the proof of the lemma.
\end{proof}

\section{Proofs Omitted from Section~\ref{sec:sample_complexity_results}}\label{sec:proofs_for_sample_complexity}

\samplelowerboundtwoelemenets*

\begin{proof}
Suppose there exists an algorithm that can distinguish between $\mathcal{D}_1$ and $\mathcal{D}_2$ with probability at least $1 - \delta$ given $N$ samples. 
Let $\mathcal{D}_1^N$ and $\mathcal{D}_2^N$ be the probability disitrubutions over the sample space $\{0,1\}^N$.
By a standard result, we know that $\delta \geq (1/2) \cdot (1- d_{\mathrm{TV}}(\mathcal{D}_1^N, \mathcal{D}_2^N))$, where $d_{\mathrm{TV}} (\mathcal{D}_1^N, \mathcal{D}_2^N)$ is the total variation distance between $\mathcal{D}_1^N$ and $\mathcal{D}_2^N$, given by $d_{\mathrm{TV}} (\mathcal{D}_1^N, \mathcal{D}_2^N) = \sum_{x \in \{0,1\}^N} (1/2) \cdot |\mathcal{D}_1^N(x) - \mathcal{D}_2^N(x)|$. For a proof, see \citet*{DBLP:books/daglib/0035708}.

Using the {Bretagnolle-Huber inequality} \cite{DBLP:books/daglib/0035708}, we obtain:
\[
1- d_{\mathrm{TV}}(\mathcal{D}_1, \mathcal{D}_2) \geq (1/2) \cdot \exp\left({-D_{\text{KL}}\left(\mathcal{D}_1^N \,\Vert\, \mathcal{D}_2^N\right)}\right),
\]
where $D_{\text{KL}}\left(\mathcal{D}_1^N \,\Vert\, \mathcal{D}_2^N\right)$ denotes the Kullback-Leibler (KL) divergence between $\mathcal{D}_1^N$ and $\mathcal{D}_2^N$, given by $D_{\text{KL}}\left(\mathcal{D}_1^N \,\Vert\, \mathcal{D}_2^N\right) = \sum_{x \in \{0,1\}^N} \mathcal{D}_1^N(x) \cdot \log(\mathcal{D}_1^N(X)/\mathcal{D}_2^N(x))$.

By a standard property of the KL divergence, since the samples are independent and identically distributed (i.i.d.), the KL divergence between the sample distributions scales linearly with $N$:
\[
D_{\text{KL}}(\mathcal{D}_1^N \,\Vert\, \mathcal{D}_2^N) = N \cdot D_{\text{KL}}(\mathcal{D}_1 \,\Vert\, \mathcal{D}_2).
\]
We can combine these observations above as follows:
\[
\delta \geq (1/2)\cdot\left( 1- d_{\mathrm{TV}}(\mathcal{D}_1, \mathcal{D}_2)\right) \geq (1/4) \cdot \exp\left({-N \cdot D_{\text{KL}}\left(\mathcal{D}_1 \,\Vert\, \mathcal{D}_2\right)}\right).
\]
Taking the natural logarithm of both sides gives:
\[
\ln(4\delta) \geq - N \cdot D_{\text{KL}}(\mathcal{D}_1 \,\Vert\, \mathcal{D}_2).
\]
Rewriting the inequality, we get:
\[
N \geq \dfrac{1}{D_{\text{KL}}(\mathcal{D}_1 \,\Vert\, \mathcal{D}_2)} \cdot \ln\left( \dfrac{1}{4\delta} \right).
\]

We can now calculate the KL-divergence of the given Bernoulli distributions. By definition of the KL divergence we have:
\begin{align*}
    D_{\text{KL}}(\mathcal{D}_1 \,\Vert\, \mathcal{D}_2) &= \mathcal{D}_1(0) \cdot \log\left(\frac{\mathcal{D}_1(0)}{\mathcal{D}_2(0)}\right) + \mathcal{D}_1(1) \cdot \log\left(\frac{\mathcal{D}_1(1)}{\mathcal{D}_2(1)}\right)  \\
    &= (1/2+\epsilon)\cdot \log\left(\frac{1/2+\epsilon}{1/2-\epsilon}\right) + (1/2-\epsilon) \cdot \log\left(\frac{1/2-\epsilon}{1/2+\epsilon}\right) \\
    &= (1/2+\epsilon)\cdot \log\left(\frac{1/2+\epsilon}{1/2-\epsilon}\right) - (1/2-\epsilon) \cdot \log\left(\frac{1/2+\epsilon}{1/2-\epsilon}\right) \\
    &= (2\epsilon) \cdot \log\left(\frac{1/2+\epsilon}{1/2-\epsilon}\right) \\
    &= (2\epsilon) \cdot \log\left(1 + \frac{2\epsilon}{1/2-\epsilon}\right) \\
    &\leq (2\epsilon) \cdot \frac{2\epsilon}{1/2-\epsilon} \\
    &\leq 16 \epsilon^2
\end{align*}
where the first inequality follows from the fact that $\log(1+x) \leq x$ for all $x > -1$ and the second inequality from the assumption that $\epsilon < 1/4$.
It follows that:
    \[
N \geq \dfrac{1}{D_{\text{KL}}(\mathcal{D}_1 \,\Vert\, \mathcal{D}_2)} \cdot \ln\left( \dfrac{1}{4\delta} \right) \geq \dfrac{1}{16 \epsilon^2} \cdot \ln\left( \dfrac{1}{4\delta} \right)
\]
which completes the proof.
\end{proof}

\samplelowerboundmanyelemenets*
\begin{proof}
    Consider two Bernoulli distributions $\mathcal{D}_{-1} = \mathrm{Ber}\left( 1/2 - \epsilon/2 \right)$ and $\mathcal{D}_{+1} = \mathrm{Ber}\left( 1/2 + \epsilon/2 \right)$ over $\{0,1\}$. Consider any algorithm $\mathcal{A}$ that, given $\ell$ samples $x_1, \ldots, x_{\ell} \in \{0,1\}$ drawn independently from $\mathcal{D}_{j}$ for some unknown $j \in \{-1,+1\}$, outputs some $\mathcal{A}(x_1, \ldots, x_k) \in \{-1,+1\}$.  By \Cref{lem:bernoulli_lb}, we know that if the algorithm satisfies $\mathcal{A}(x_1, \ldots, x_k) = j$ with probability at least $34/64$, then $\ell \geq C/\epsilon^2$ for some constant $C > 0$. Equivalently, we can say that if $\ell < C/\epsilon^2$, then any algorithm satisfies $\mathcal{A}(x_1, \ldots, x_k) \neq j$ with probability at least $30/64$, and thus $\mathbb{E}[\mathbbm{1}[\mathcal{A}(x_1, \ldots, x_k) \neq j]] \geq 30/64$. 
    
    Suppose that there exists an algorithm satisfying the properties given in the statement of the lemma that takes only $N < (n/30) \cdot (C/\epsilon^2)$ samples, where $C$ is the constant given above. 
    We denote the samples as $x_1, \ldots, x_{N} \in \mathcal{X}$.
    By the assumption, the number of mistakes that the algorithm makes is at most $(1/4) \cdot n$ with probability $3/4$, and at most $n$ with probabiltiy $1/4$. Thus, we must have:
    \begin{align*}
        \mathbb{E}_{x_0, \ldots, x_{N-1} \sim \mathcal{D}^{(z)}}\left[ \sum_{i=1}^n \mathbbm{1}[z_i' \neq z_i] \right] \leq (1/4) \cdot n + (3/4) \cdot ((1/4) \cdot n) = (28/64) \cdot n.
    \end{align*}
    Our strategy to prove the lemma is to derive a contradiction from this inequality.

    Drawing $N$ samples $x_1, \ldots, x_{N}$ from $\mathcal{D}^{(z)}$ is equivalent to the following procedure that is easier to analyze: 
    Let $i_1, \ldots, i_N \in \{1, \ldots, n\}$ be independent random variables drawn uniformly at random from $\{1, \ldots, n\}$.
    Then, let $j_1, \ldots, j_{N} \in \{0, 1\}$ be independent random variables, where for each $k \in \{1, \ldots, N\}$, the random variable
    $j_k$ is drawn from $\mathcal{D}_{-1}$ if $z_{i_k} = -1$ and from $\mathcal{D}_{+1}$ if $z_{i_k} = +1$. Finally, let $x_k = (i_k)_{(j_k)}$. To verify the equivalence, note that for every $i \in \{1, \ldots ,n \}$, we have:
    \begin{align*}
        \mathbb{P}[x_k = i_0] &= \mathbb{P}[i_k = i \text{ and } j_k = 0] && (\text{since $x_k = (i_k)_{(j_k)}$})\\
        &= 
\mathbb{P}[i_k = i] \cdot \mathbb{P}[j_k = 0 \,|\, i_k = i] \\        
&= (1/n) \cdot (1/2 - z_i \cdot \epsilon/2) && (\text{since $j_k \sim \mathcal{D}_{(z_i)} = \mathrm{Ber}(1/2 + z_i \cdot \epsilon/2)$}) \\
        &= (1 - z_i \cdot \epsilon)/(2n) 
    \end{align*}
    In a similar way, we show:
    \begin{align*}
        \mathbb{P}[x_k = i_1] &= \mathbb{P}[i_k = i \text{ and } j_k = 1] = (1/n) \cdot (1/2 + z_i \cdot \epsilon/2) = (1 + z_i \cdot \epsilon)/(2n)
    \end{align*}
    as needed.

    Consider any fixed draw of $i_1, \ldots, i_{N}$. 
    For each $i \in \{1, \ldots, n\}$, let $N_i = |\{ k \in \{1, \ldots, N\}: i_k = i \}|$. Note that there are at least $(29/30) \cdot n$ values of $i$ with $N_i < C/\epsilon^2$ samples. Indeed, if the contrary holds, then we have at least $(1/30) \cdot n$ values of $i$ with $N_i \geq C/\epsilon^2$ samples, which implies $\sum_{i=1}^n N_i \geq ((1/30) \cdot n) \cdot (C/\epsilon^2) \geq (n/30) \cdot (C/\epsilon^2) > N$, which is a contradiction.
    Now, the crucial observation is that if $N_i < C/\epsilon^2$, then, since the algorithm observes less than $C/\epsilon^2$ samples from $\mathcal{D}_{z_i}$, we have $\mathbb{E}_{j_1, \ldots, j_{N}}[\mathbbm{1}[z_i' \neq z_i]] \geq 30/64$ by the argument we made earlier.
    Therefore, we have:
     \begin{align*}
         \mathbb{E}_{x_1, \ldots, x_{N}}\left[ \sum_{i=1}^n \mathbbm{1}[z_i' \neq z_i] \right] &= \mathbb{E}_{i_1, \ldots, i_{N}}\left[ \mathbb{E}_{j_1, \ldots, j_{N}}\left[\sum_{i=1}^n 
 \mathbbm{1}\left[z_i' \neq z_i\right]\,\bigg\rvert\, i_1, \ldots, i_N \right]  \right]\\
         &= \mathbb{E}_{i_1, \ldots, i_{N}}\left[ \sum_{i=1}^n \mathbb{E}_{j_1, \ldots, j_{N}}\left[\mathbbm{1}\left[z_i' \neq z_i\right] \,\rvert\, i_1, \ldots, i_N \right] \right] \\
         &\geq \mathbb{E}_{i_1, \ldots, i_{N}}\left[ \sum_{i=1}^n (30/64) \cdot \mathbbm{1}\left[N_i < C/\epsilon^2\right]  \right] \\
         &\geq \mathbb{E}_{i_1, \ldots, i_{N}}\left[ (29/30) \cdot n \cdot (30/64) \right] \\
         &= (29/64) \cdot n
    \end{align*}
    which contradicts the inequality we derived earlier, and hence completes the proof.
\end{proof}

\samplecomplexitylbboundedcontracts*

\begin{proof}
    Learning bounded contracts in the general $m$-outcome case is at least as hard as in the binary outcome case with $m=2$, and in the binary outcome model, learning bounded contracts is essentially equivalent to learning linear contracts by \Cref{lem:binary_outcome}\ref{enum:lin_opt}. Therefore, we obtain the lower bound of $N \geq \Omega((1/\epsilon^2) \log (1/\delta))$ directly from \Cref{thm:sample_lower_bound_linear}.

    We first define $2(m- 1)$ agent's types, $\theta_1^j$ and $\theta_2^j$ for $j \in [m-1]$. Let $\theta_1^j$ be a type defined as follows:
    \begin{table}[H]
        \centering
        \begin{tabular}{|c|ccc|c|}
        \hline
        & $r_0 = 0$  & $r_j = 1$ & $r_{j'} = 1$ &\textbf{cost} \\ \hline
        \textbf{action 0:} & $1/2$ & $1/2$ & $0$ & $c_0 = 0$ \\ 
        \textbf{action 1:} & $1/2$ & $1/2$ & $0$ & $c_1 = 1/4$\\ \hline
        \end{tabular}
    \end{table}
    Let $\theta_2^j$ be a type defined as follows:
\begin{table}[H]
        \centering
        \begin{tabular}{|c|ccc|c|}
        \hline
        & $r_0 = 0$  & $r_1 = 1$ & $r_{j'} = 0$ &\textbf{cost} \\ \hline
        \textbf{action 0:} & $1$ & $0$ & $0$ & $c_0 = 0$ \\ 
        \textbf{action 1:} & $1/2$ & $1/2$ & $0$ & $c_1 = 1/4$\\ \hline
        \end{tabular}
    \end{table}
    Note that for $\theta_1^j$, the agent's best response is always action $0$, and so $u_p(\theta_1^j, t) = (1/2) \cdot (1-t_j)$ for any $t_j$. For $\theta_2^j$, the agent's best response is action $0$ for $t_j < 1/2$, and action $1$ for $ t_j \geq 1/2$. Hence, $u_p(\theta_2^j, t) = 0$ for $t_j < 1/2$, and $u_p(\theta_2^j, t) = (1/2) \cdot (1-t_j)$ for $t_j \geq 1/2$.

    For every sequence $z_1, \ldots, z_{m-1} \in \{+1, -1\}$, we can define a distribution $\mathcal{D}^{(z)}$ which puts probability mass $(1-16\epsilon z_j)/(2(m-1))$ on type $\theta^j_1$ and probability mass $(1+16\epsilon z_j)/(2(m-1))$ on type $\theta^j_2$. Note that the these terms add up to $1$ and so $\mathcal{D}^{(z)}$ is a valid probability distribution for every sequence $z$ with elements in $\{+1,-1\}$.

    We can verify that for distribution $\mathcal{D}^{(z)}$, the optimal contract sets $t_j = 1/2$ if $z_j = +1$ and $t_j = 0$ if $z_j = -1$.
    Without loss of generality, we can assume that the learning algorithm $\mathcal{A}$ always returns a contract $t$ satisfying $t_0 = 0$ and $t_j \in \{0, 1/2\}$ for all $j \in \{1, \ldots, m-1\}$.
    For any such contract $t$, we can define $w_j \in \{0, 1\}$  depending on the combination of $z_j \in \{+1, -1\}$ and $t_j \in \{0, 1/2\}$. If $z_j = +1$ and $t_j = 0$, we set $w_j = 1$. If $z_j = -1$ and $t_j = 1/2$, we set $w_j = 1$. Otherwise, we set $w_j = 0$.

If $w_j$ is set to $1$, then this means that the contract $t$ did not identify correctly the payment $t_j$ for outcome $j$. 
Note that the agent type distribution is uniform over  
$j \in [m - 1]$ 
Conditioned on the agent being of type $\theta^j_i$ for $i \in \{1, 2\}$, the probability mass assigned to the two types is   $\frac{1}{2} - 16 \epsilon z_j $ and $\frac{1}{2} + 16 \epsilon z_j $, respectively.
Thus, now we are at exactly the setting of the proof in \Cref{thm:sample_lower_bound_linear}.
But in \Cref{thm:sample_lower_bound_linear} we showed that for this (now conditioned on $j$) agent type distribution, the principal's expected utility if  $w_j$ is $1$ is at least $2 \epsilon$ less compared to the principal's utility for the optimal value of $t_j$. 
Given that the agent distribution is uniform over $j \in [m - 1]$,
{for the principal's expected utility conditioned on $w = (w_1, \dots, w_{m-1})$, we have: }
\begin{align*}
    \mathbb{E}_{\theta \sim \mathcal{D}^{(z)}}[u_p(\theta, t) | w] - \mathrm{OPT}(\mathcal{D}^{(z)}, \Cbounded) = \frac{1}{m-1} \cdot \sum_{j=1}^{m-1} (2\epsilon) \cdot w_j
\end{align*}
Thus,  to satisfy $\mathbb{E}_{\theta \sim \mathcal{D}^{(z)}}[u_p(\theta, t)] \geq \mathrm{OPT}(\mathcal{D}^{(z)}, \Cbounded) - \epsilon$, we must have $\sum_{j=1}^{m-1} w_j \leq (1/4) \cdot (m-1)$. 
{In other words, our learning algorithm needs to be able to correctly identify the $z$ value for at least $(3/4) \cdot (m-1)$ indices.}
By \Cref{lem:discrete_lb}, this requires $N \geq \Omega(m/\epsilon^2)$.
\end{proof}

\sampleimpossibilityunboundedcontracts*

\begin{proof}
    We prove the statement for the case where $\numoutcomes=3$ and $\numactions=2$. The general case follows since we can always extend each distribution over three outcomes to include the remaining $\numoutcomes-3$ outcomes with probability $0$, and we can include additional copies of any of the two actions to make the total number of actions equal to $\numactions$.

    Suppose for the sake of contradiction that there is a learning algorithm  $\mathcal{A}$ such that for every distribution $\mathcal{D}$ over the agent's type space $\Thetaall$, given $K$ samples from $\mathcal{D}$, with probability at least $1-\delta$ over the draw of the samples from $\mathcal{D}$, algorithm $\mathcal{A}$ outputs a contract $t$ with $\mathbb{E}_{\theta \sim \mathcal{D}} [u_p(\theta,t)] \geq \mathrm{OPT}(\Cunbounded, \mathcal{D}) - \epsilon$.

    We first define two agent types, $\theta^{(1)}$ and $\theta^{(2)}$. 
    Let $\eta = (1/4 - \epsilon) \cdot (1-\delta^{1/K})$.
    The agent type $\theta^{(1)}$ is given in the following table:
    \begin{table}[H]
        \centering
        \begin{tabular}{|c|ccc|c|}
        \hline
        & $r_0 = 0$ & $r_1 = 0$ & $r_2 = 1$ & \textbf{cost} \\ \hline
        \textbf{action 0:} & $0$ & $1/2$ & $1/2$ & $c_1 = 0$ \\ 
        \textbf{action 1:} & $\eta$ & $0$ & $1-\eta$ & $c_0 = 1/4$\\ \hline
        \end{tabular}
    \end{table}
    The agent type $\theta^{(2)}$ is given in the following table:
    \begin{table}[H]
        \centering
        \begin{tabular}{|c|ccc|c|}
        \hline
        & $r_0 = 0$ & $r_1 = 0$ & $r_2 = 1$ & \textbf{cost} \\ \hline
        \textbf{action 0:} & $1$ & $0$ & $0$ & $c_1 = 0$ \\ 
        \textbf{action 1:} & $1$ & $0$ & $0$ & $c_0 = 0$\\ \hline
        \end{tabular}
    \end{table}
    We consider two distributions over agent types, $\mathcal{D}^{(1)}$ and $\mathcal{D}^{(2)}$. Distribution $\mathcal{D}^{(1)}$ assigns probability $1$ to type $\theta^{(1)}$. Distribution $\mathcal{D}^{(2)}$ assigns probability $\delta^{1/K}$ to type $\theta^{(1)}$ and probability $1-\delta^{1/K}$ to type $\theta^{(2)}$.

    We first analyze the output of algorithm $\mathcal{A}$ for distribution $\mathcal{D}^{(1)}$. Notice that there is only one possible sample that can be drawn from that distribution which means that the output is deterministic. Denote the contract returned by $\mathcal{A}$ as $t$. By the assumption, it must hold that $u_p(\theta^{(1)}, t) = \mathbb{E}_{\theta \sim \mathcal{D}^{(1)}}[u_p(\theta,t)] \geq \mathrm{OPT}(\Cunbounded, \mathcal{D}^{(1)})-\epsilon$. To get a lower bound on the value of $\mathrm{OPT}(\Cunbounded, \mathcal{D}^{(1)})$, we can consider a contract $t^\star$ given by:
    \begin{align*}
        t_0^\star = 1/(4 \cdot \eta), \quad t_1^\star = 0, \quad t_2^\star = 0.
    \end{align*}
     The contract above incentivizes action $1$ since:
    \begin{align*}
    u_a(\theta^{(1)}, t^\star, 0) &= t^\star \cdot f_1 - c_1 = 0 + 0 + 0 - 0 = 0,\\
    u_a(\theta^{(1)}, t^\star, 1) &= t^\star \cdot f_{0} - c_0 = 1/4 + 0 + 0 - 1/4 =  0.
    \end{align*}
    We get:
    \begin{align*}
        \mathrm{OPT}(\Cunbounded, \mathcal{D}^{(1)}) \geq u_p(\theta^{(1)}, t^\star) \geq u_p(\theta^{(1)}, t^\star, 1) = (r-t^\star) \cdot f_1 = -1/4 + 0 + 1 = 3/4.
    \end{align*}
    It  follows that $u_p(\theta^{(1)}, t)\geq \mathrm{OPT}(\Cunbounded, \mathcal{D}^{(1)})-\epsilon \geq 3/4 - \epsilon > 1/2$ since $\epsilon < 1/4$.
    In particular, $t$ must incentivize action $1$ because if the agent takes action $0$, the principal's utility is at most $r \cdot f_0 = 1/2$. Since $t$ incentivizes action $1$, we must have:
    \begin{align*}
        &(1/2) \cdot t_1 + (1/2) \cdot t_2 = u_a(\theta^{(1)}, t, 0) \leq u_a(\theta^{(1)}, t, 1) = \eta \cdot t_0 + (1-\eta) \cdot t_2 - 1/4.
    \end{align*}
    We can rearrange this inequality as follows:
    \begin{align}
         & \eta \cdot t_0 - (1/2) \cdot t_1 + (1/2-\eta) \cdot t_2 \geq 1/4 \nonumber  \\
        \Longrightarrow \quad 
        & \eta \cdot t_0  + (1/2-\eta) \cdot t_2 \geq 1/4  \nonumber  \\
        \Longrightarrow \quad & (1-\eta) \cdot \eta \cdot t_0 + (1/2-\eta) \cdot (1-\eta) \cdot t_2 \geq (1/4) \cdot (1-\eta) \label{eq:one}
    \end{align}
    where the first inequality follows since $t_1 > 0$ and the second one by  multiplying both sides by $(1-\eta)$.
    Moreover, since the principal's utility must be at least $3/4-\epsilon$, we must have
    \begin{align*}
        -\eta \cdot t_0 + (1-\eta) \cdot (1-t_2) = u_p(\theta^{(1)}, t, 1) \geq 3/4 - \epsilon.
    \end{align*}
    We can rearrange this inequality as follows:
    \begin{align}
        & -\eta \cdot t_0 + (1-\eta) \cdot (1-t_2) \geq 3/4-\epsilon \nonumber  \\
        \Longrightarrow \quad & -\eta \cdot (1/2-\eta) \cdot t_0 + (1-\eta) \cdot (1/2-\eta) \cdot t_2 \geq (3/4-\epsilon) \cdot (1/2-\eta) \label{eq:two}
    \end{align}
    where the second inequality follows by multiplying both sides by $(1/2-\eta)$. By adding inequaliaties (\ref{eq:one}) and (\ref{eq:two}), we obtain:
    \begin{align*}
        & (\eta/2) \cdot t_0 + (1/2-\eta) \cdot (1-\eta) \geq 5/8 - \eta - (1/2) \cdot \epsilon + \epsilon \cdot \eta \\
        \Longrightarrow \quad & (\eta/2) \cdot t_0 + (1/2-\eta) \cdot (1-\eta) \geq 5/8 - \eta - (1/2) \cdot \epsilon && (\text{since $\epsilon \geq 0$}) \\
        \Longrightarrow \quad & (\eta/2) \cdot t_0 + (1/2-\eta) \cdot (1-\eta) \geq 1/2 - \eta + (1/2) \cdot (1/4-\epsilon) \\
        \Longrightarrow \quad & (\eta/2) \cdot t_0 + 1/2 - (3/2) \cdot \eta + \eta^2 \geq 1/2 - \eta + (1/2) \cdot (1/4-\epsilon) \\
        \Longrightarrow \quad & (\eta/2) \cdot t_0 \geq (1/2) \cdot \eta - \eta^2 + (1/2) \cdot (1/4-\epsilon) \\
        \Longrightarrow \quad & t_0 \geq 1 - 2 \cdot \eta + (1/\eta) \cdot (1/4-\epsilon) \\
       \Longrightarrow \quad & t_0 \geq (1/\eta) \cdot (1/4-\epsilon) && (\text{since $\eta \leq 1/2$}) \\
        \Longrightarrow \quad & t_0 \geq 2/(1-\delta^{1/K}) && (\text{by the definition of $\eta$})
    \end{align*}
    Let us now consider the second disrtibution, $\mathcal{D}^{(2)}$. Notice that with probability $\delta^{(1/K) \cdot K} = \delta$, all samples are of type $\theta^{(1)}$. In that case, the algorithm $\mathcal{A}$ returns the same contract $t$ as in the case of distribution $\mathcal{D}^{(1)}$. Let us calculate the principal's utility generated by contract $t$ for the second type $\theta^{(2)}$:
    \begin{align*}
        u_p(\theta^{(2)}, t) = (r-t) \cdot f_0 = -t_0 \leq -2/(1-\delta^{1/K})
    \end{align*}
    We conclude:
    \begin{align*}
        u_p(\mathcal{D}^{(2)}, t) &= \delta^{1/K} \cdot u_p(\theta^{(1)}, t) + (1-\delta^{1/K}) \cdot u_p(\theta^{(2)}, t) \\
        &\leq \delta^{1/K} \cdot  1 + (1-\delta^{1/K}) \cdot (-2/(1-\delta^{1/K})) \\
        &\leq -1
    \end{align*}
    Note that by considering the zero contract, we can get that $\mathrm{OPT}(\Cunbounded, \mathcal{D}^{(2)}) \geq 0$, which means that $u_p(\mathcal{D}^{(2)}, t) \leq \mathrm{OPT}(\Cunbounded, \mathcal{D}^{(2)}) - \epsilon$ since $\epsilon < 1/4$.
    Therefore, with probability at least $\delta$, the algorithm returns a contract with representation error higher than $\epsilon$, which contradicts the initial assumption on the algorithm $\mathcal{A}$.
\end{proof}

\section{Proofs Omitted from Section~\ref{sec:extensions}}
\label{sec:proofsforextensions}
\pdimupperboundsformenus*

\begin{proof}
We first show that the class of all menus of contracts of size $K$ is $(M \cdot m, M^2 \cdot n^2)$-delineable, establishing a bound on the pseudo-dimension.

Fix any agent type $\theta = (f,c)$. A menu of $K$ (unbounded) contracts $M \in \mathcal{M}_{\text{unbounded}}^K$ is parametrized by $K$ contracts $ (t^{0}, \ldots, t^{K-1}) $, where each contract is represented as $t^k = (t^k_0, \ldots, t^k_{\numoutcomes - 1}) \in \mathbb{R}_{\ge 0}^m$ for all $k \in [K]$.

Fix a sample of the agent's type $(f, c)$. For this type, we define 
$\binom{\numactions \cdot K}{2}$ hyperplanes. Specifically, for any two pairs $(i,k), (i',k') \in \setactions \times [K]$ (where each pair consists of an action and a contract index), we define the hyperplane:
\begin{align*}
H_{(i,k), (i',k')} = \{ M = (t^1, \dots, t^K) \in \mathbb{R}_{\ge 0 }^{K \cdot m} \,:\, t^k \cdot f_i - c_i = t^{k'} \cdot f_w - c_w\},    
\end{align*}
This hyperplane represents contract menus where the agent is indifferent between selecting the $k$-th contract and taking action $i$ versus selecting the $k'$-th contract and taking action $i'$.

For any contract menu $M$ within a connected component of
$\mathbb{R}_{\ge 0}^{K \cdot m} \setminus \bigcup_{(i,k), (i',k') \in \setactions \times [K]} H_{(i,k),(i',k')}$, 
the agent's optimal contract $t \in M$ and best response action $i^{*}(\theta, t)$ remain constant. Consequently, the principal's utility, given by  
$(r-t) \cdot f_{i^{\star}(\theta, t)}$, is linear across each connected component.

Thus, the class $\mathcal{M}_{\text{unbounded}}^K$ is $(K \cdot m, K^2 \cdot n^2)$-delineable. Applying Theorem \ref{theorem:delineable_pseudo-dimensions}, we obtain the bound $\pdim{\mathcal{M}_{\text{unbounded}}^K} = O (K \cdot \numoutcomes \cdot  \log (K \cdot \numoutcomes \cdot \numactions))$. Finally, since $\mathcal{M}_{\text{bounded}}^K \subseteq \mathcal{M}_{\text{unbounded}}^K$, it follows that $\pdim{\mathcal{M}_{\text{bounded}}^K} = O (K \cdot \numoutcomes \cdot  \log (K \cdot \numoutcomes \cdot \numactions))$.
\end{proof}

\end{document}